\pdfoutput=1
\documentclass[10pt,twocolumn,letterpaper]{article}

\usepackage[
  letterpaper,
  top=0.62in,
  bottom=0.72in,
  left=0.66in,
  right=0.66in,
  columnsep=13.5pt
]{geometry}
\usepackage[T1]{fontenc}
\usepackage[utf8]{inputenc}
\usepackage{lmodern}
\usepackage{microtype}
\usepackage{helvet}
\usepackage{overpic}
\usepackage{multirow}

\usepackage{amsmath,amssymb,amsthm,mathtools}
\usepackage{graphicx}
\usepackage{booktabs}
\usepackage{enumitem}
\usepackage{xcolor}
\usepackage{caption}
\usepackage{subcaption}
\usepackage{dblfloatfix}
\usepackage{balance}

\usepackage[numbers,sort&compress]{natbib}
\usepackage[hidelinks]{hyperref}
\hypersetup{
  pdftitle={Perfectly equidistributed Quasi-Monte Carlo sequences from Artin-Schreier polynomials},
  pdfauthor={Nicolas Bonneel, David Coeurjolly, Victor Ostromoukhov}
}

\usepackage{fancyhdr}
\usepackage{titlesec}

\titleformat{\section}
  {\sffamily\bfseries\normalsize}
  {\thesection.}{0.45em}{}

\titleformat{\subsection}[runin]
  {\sffamily\bfseries\normalsize}
  {\thesubsection.}{0.45em}{}[.]

\titleformat{\subsubsection}[runin]
  {\sffamily\bfseries\small}
  {\thesubsubsection.}{0.45em}{}[.]

\titlespacing*{\section}
  {0pt}{1.2ex plus 0.4ex minus 0.2ex}{0.45ex}

\titlespacing*{\subsection}
  {0pt}{0.9ex plus 0.3ex minus 0.2ex}{0.55em}

\titlespacing*{\subsubsection}
  {0pt}{0.7ex plus 0.2ex minus 0.1ex}{0.55em}

\setlist{nosep,leftmargin=*}
\newtheorem{theorem}{Theorem}
\newtheorem{proposition}[theorem]{Proposition}
\newtheorem{lemma}[theorem]{Lemma}

\theoremstyle{definition}

\theoremstyle{remark}

\newcommand{\cP}{\mathcal{P}}
\newcommand{\cK}{\mathcal{K}}

\newcommand{\PaperTitle}{%
  Perfectly equidistributed Quasi-Monte Carlo sequences from\\
  Artin-Schreier polynomials}

\newcommand{\PaperAuthors}{%
  Nicolas Bonneel\textsuperscript{1,*},
  David Coeurjolly\textsuperscript{1}, and
  Victor Ostromoukhov\textsuperscript{1}}

\newcommand{\PaperAffiliations}{%
  \textsuperscript{1}CNRS, Universit\'e Lyon 1, INSA Lyon, LIRIS, France\\
  \textsuperscript{*}Corresponding author: \href{mailto:nicolas.bonneel@cnrs.fr}{nicolas.bonneel@cnrs.fr}}

\newcommand{\PaperAbstract}{%
To numerically integrate a function, one may resort to Quasi-Monte Carlo estimators, that average integrand values at pseudo-random well-distributed uniform sampling locations. Better uniformity improves the worst-case integration-error bound.
A standard measure of uniformity is given by an integer $t$ value, where $t=0$ yields the best uniformity.
Producing sequences of samples with bounded $t$ values can be achieved with Sobol' recursive construction, that uses coefficients of irreducible polynomials. 
While $b$-dimensional sequences with $t=0$ can be obtained by taking $b$ polynomials of degree $1$ over the Galois Field $\mathrm{GF}(b)$, we show conditions that guarantee $t=0$ for specific higher degree polynomials. In particular, we relate the Sobol' construction to tensorized powers of Pascal matrices when the chosen polynomials only differ by a constant and exhibit simple conditions to guarantee $t=0$ in this case.
We then focus on Artin-Schreier irreducible polynomials, in the form $p_i(x) = x^b - x + c_i$, where $i \in  \{1, \dots, b-1\}$ and $b$ is prime, and we make explicit conditions that always guarantees $t=0$ in $b-1$ dimensions. 
Combining $b$-dimensional Sobol' of degree $1$ and our $(b-1)$-dimensional Artin-Schreier sequence of degree $b$, we provide a fast greedy procedure 
that optimizes the $(2b-1)$-dimensional combined $t$ value, while guaranteeing $t=0$ projection in subspaces.}

\newcommand{\PaperKeywords}{%
Quasi-Monte Carlo; Sobol' sequences; Artin-Schreier polynomials; digital nets; $(0,s)$-sequences}

\newcommand{\makefrontmatter}{%
\twocolumn[
\begin{minipage}{\textwidth}
  \vspace*{-1.5em}
  {\sffamily\bfseries\fontsize{18}{21}\selectfont \PaperTitle\par}
  \vspace{0.65em}
  {\sffamily\fontsize{10.5}{13}\selectfont \PaperAuthors\par}
  \vspace{0.35em}
  {\sffamily\fontsize{8.5}{10.5}\selectfont \PaperAffiliations\par}
  \vspace{0.9em}

  \begin{center}
  \begin{minipage}{0.88\textwidth}
    \small
    \PaperAbstract

    \vspace{0.65em}
    {\sffamily\bfseries Keywords:} \PaperKeywords
  \end{minipage}
  \end{center}

  \vspace{0.65em}
  \hrule height 0.45pt
  \vspace{0.9em}
\end{minipage}
]
}

\begin{document}
\makefrontmatter
\thispagestyle{fancy}


\section{Introduction}

Numerically estimating high-dimensional integrals benefits from Quasi-Monte Carlo sampling. The resulting estimator averages evaluations of the integrand at specific pseudo-random locations 
$$\int_{[0, 1)^s} f(x) \text{d}x \approx \frac{1}{n} \sum_{i=1}^n f(x_i)\,,$$
where $\{x_i\}_{i=1..n}$ uniformly covers the integration domain (we will consider an $s$-dimensional unit hypercube, for simplicity). Uniformity can be measured by a non-negative integer value $t$, where smaller $t$ indicates better uniformity. In this context, a $(t,m,s)$-net is a set of $n=b^m$ points in $s$ dimensions, where if we regularly subdivide the integration domain in rectangles of side lengths $\frac{1}{b^{k_i}}$ for $k_i \in \mathbb{N}, i \in \{1 \dots s\}$, then each rectangle of area $\frac{b^t}{n}$ has exactly $b^t$ points inside (see figure \ref{fig:t1}). This $t$ value amounts to a stratification property of the point set. 
High uniformity is desirable as it improves the integration error in $\mathcal{O}\left(C(s,b) b^t \frac{\left(\log n\right)^s}{n}\right)$ for integrands of bounded Hardy and Krause variation 
~\cite{niederreiter1992random} where $C$ only depends on $s$ and $b$.
A $(t,s)$-sequence is a sequence of sample points where suitable subsets of $b^m$ sample points all satisfy the $(t,m,s)$-net property -- this allows for progressive estimation of an integral, refining the estimated integral value by adding more points while maintaining low integration error.

The Sobol' method is widely used as it produces $(t,s)$-sequences. This construction uses a set of irreducible polynomials $\{p_i\}_{i=1 \dots s}$ in the Galois Field $GF(b)$ and small initialization matrices in $GF(b)$ to produce larger matrices, that, in turn, are used to produce points. The crux of Sobol' sequences lies in determining these polynomials and initialization matrices to ensure low $t$ values. It can be shown that by taking $p_i(x) = x + i$, for $i=0 \dots s$ ($s < b$), the resulting $s$-dimensional sequence has $t=0$ and the matrices produced relate to powers of Pascal matrices~\cite{faure1982discrepance}. While cases for which $t=1$ (when base $b=2$) have been found~\cite{BCIO25}, we are not aware of other conditions guaranteeing $t=0$ (for more general $b$) in the context of Sobol' sequences.
\begin{figure}[t]
\small    
\begin{overpic}[width=1.7cm]{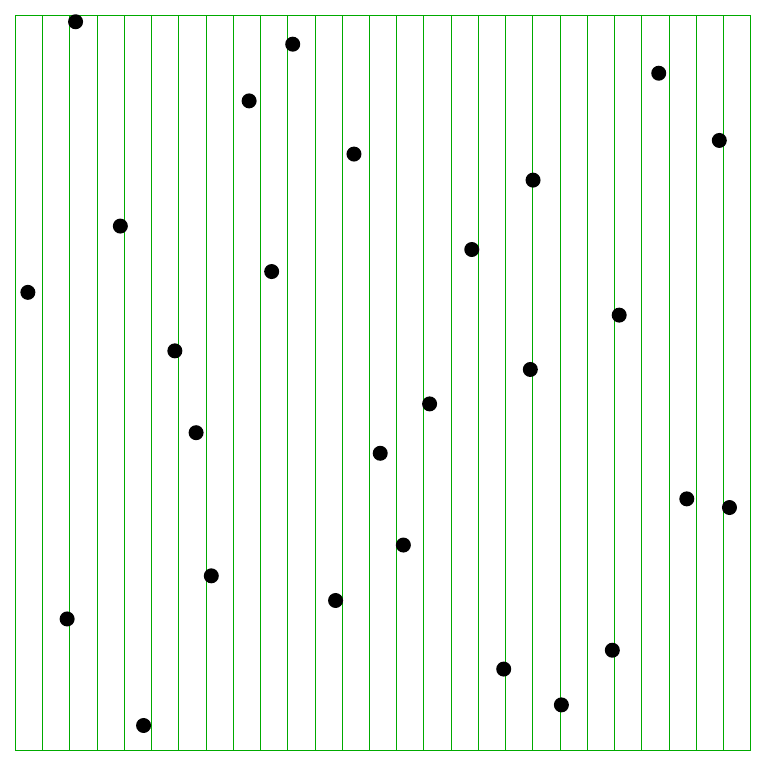}
        \put(350,-25){\rotatebox{0}{$(0,3,2)-$net, $\mathrm{GF}(3)$}}
        \put(350,-165){\rotatebox{0}{$(1,3,2)-$net, $\mathrm{GF}(3)$}}
        \put(350,-295){\rotatebox{0}{$(0,2,2)-$net, $\mathrm{GF}(5)$}}
    \end{overpic}
    \includegraphics[width=1.7cm]{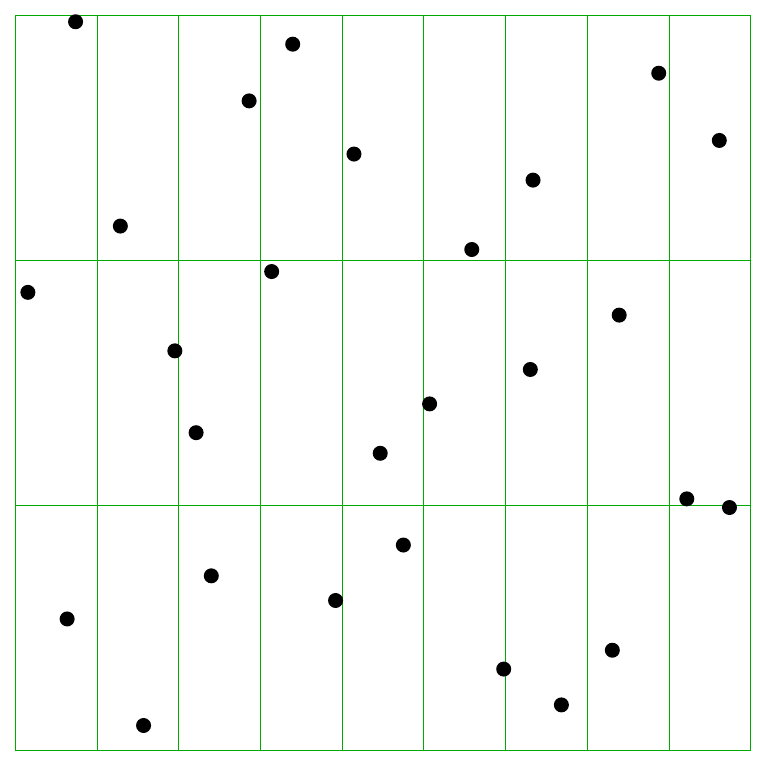}
    \includegraphics[width=1.7cm]{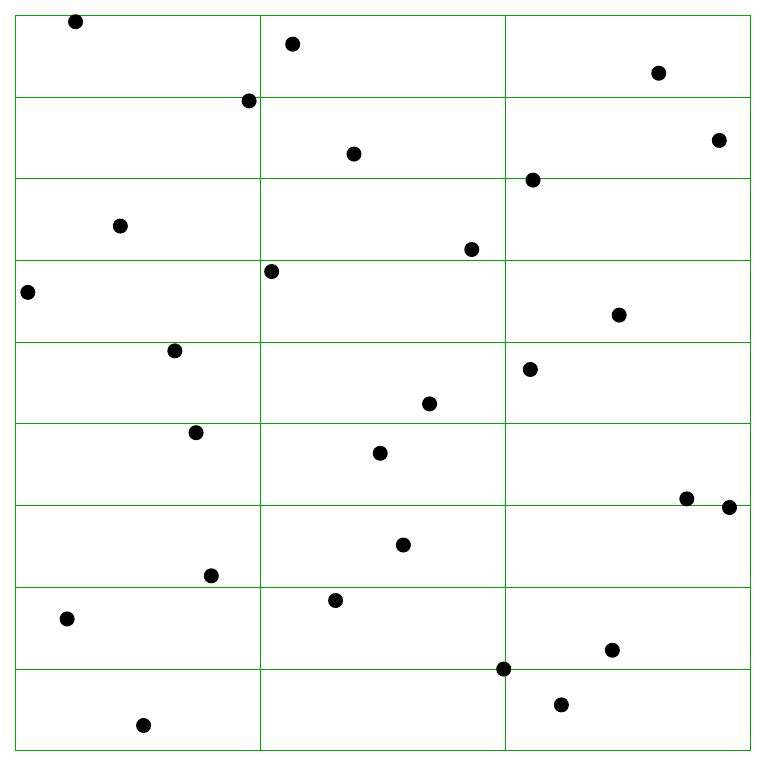}
    \includegraphics[width=1.7cm]{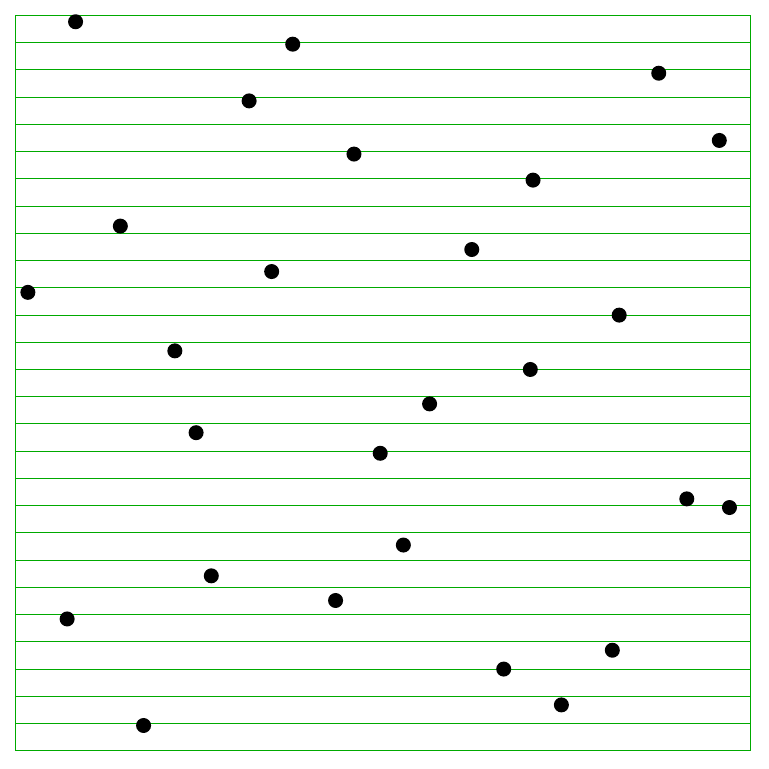}\\
    \vspace{.2cm}
    \hrule
     \vspace{.2cm}
    \includegraphics[width=1.7cm]{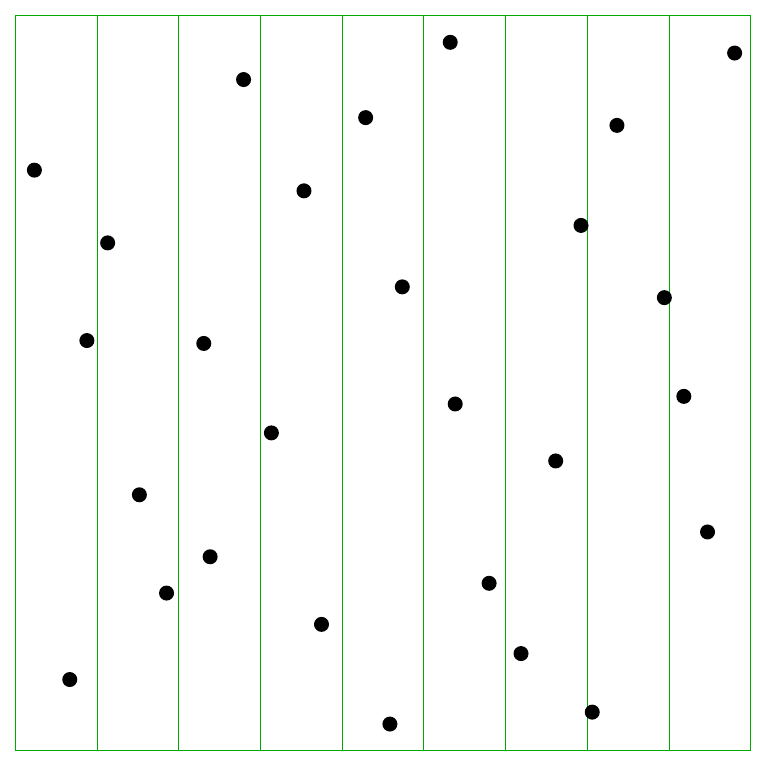}
    \includegraphics[width=1.7cm]{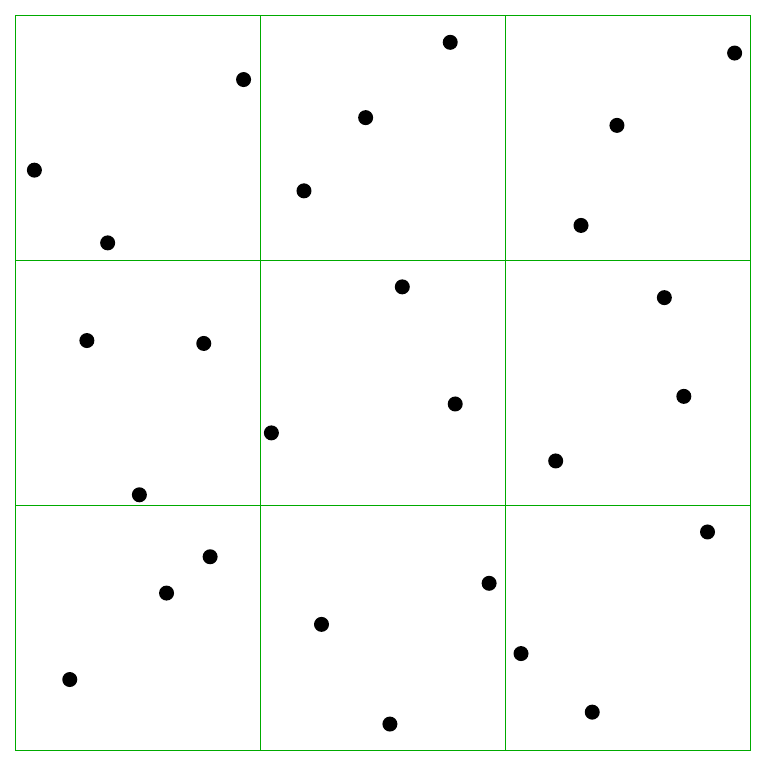}
    \includegraphics[width=1.7cm]{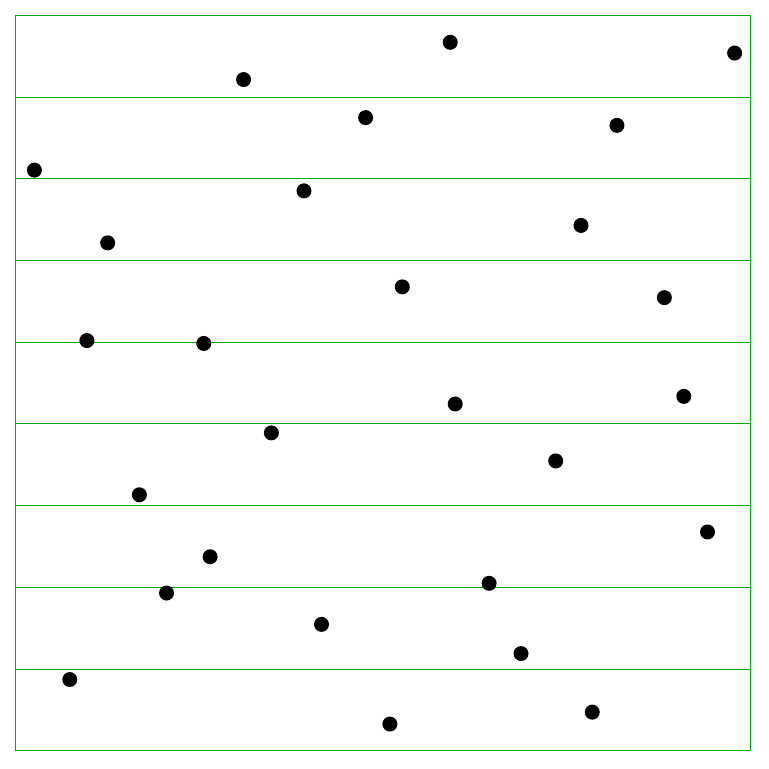}\\
    \vspace{.2cm}
    \hrule
    \vspace{.2cm}
    \includegraphics[width=1.7cm]{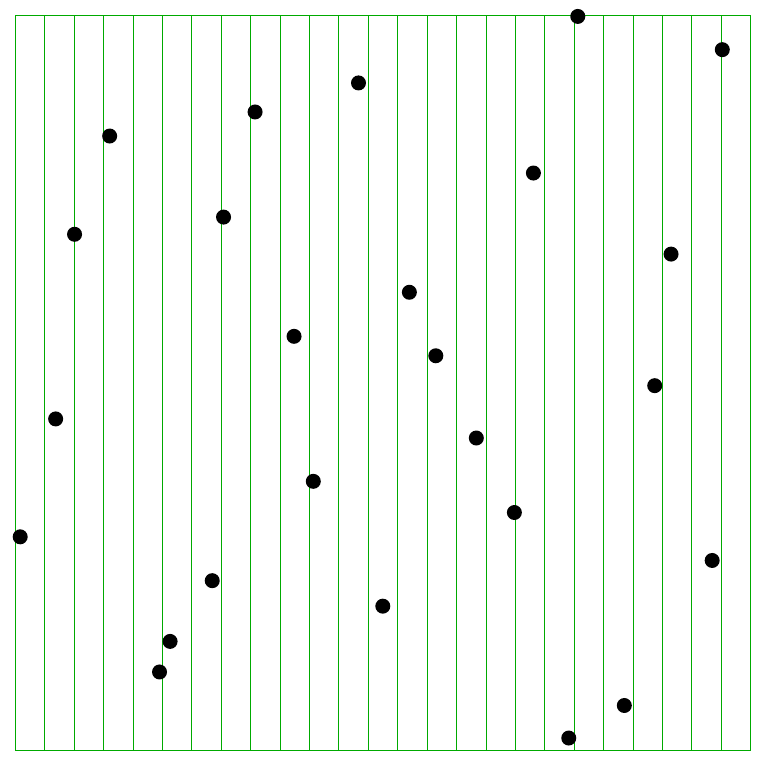}
    \includegraphics[width=1.7cm]{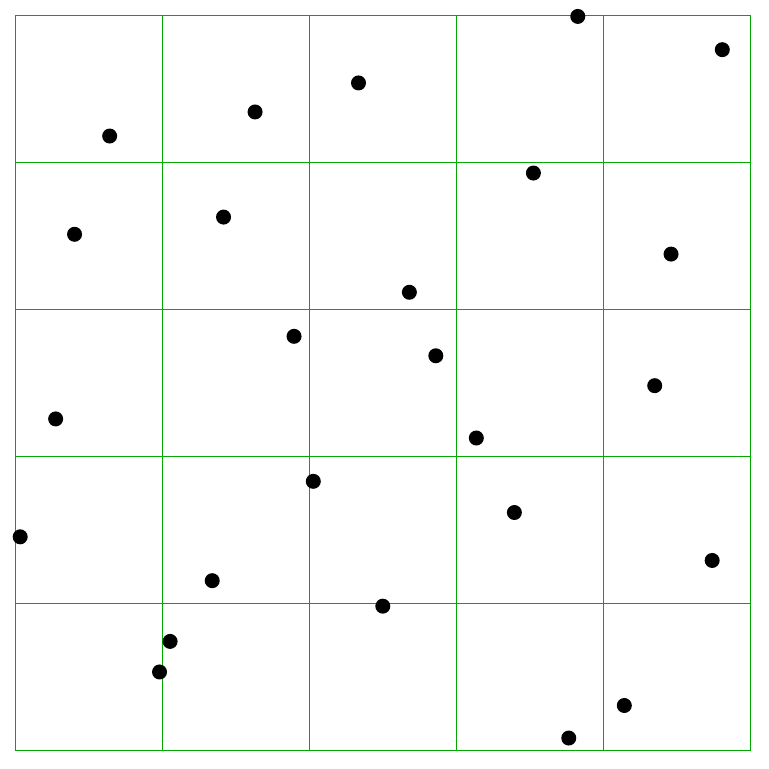}
    \includegraphics[width=1.7cm]{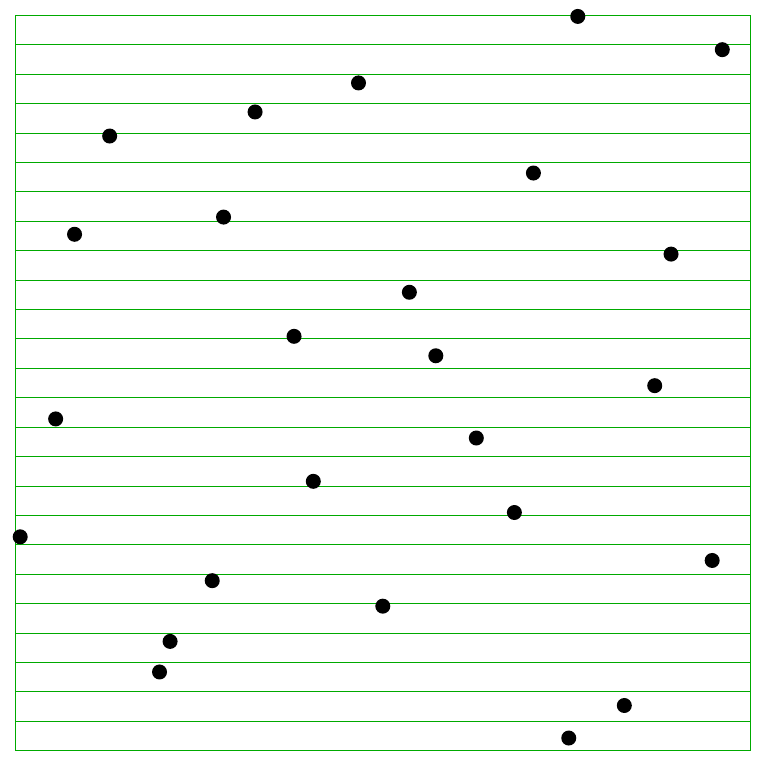}
   
    \caption{A $(t,m,2)$-net in base $b$ is a set of $n=b^m$ samples in 2D such that all $b$-adic intervals of area $b^t/b^m$ contain $b^t$ points. 
From top to bottom, we illustrate a $(0,3,2)$-net of $3^3 = 27$ samples in base $3$ with its 4 stratification constraints with $3^0=1$ sample in each stratum, a $(1,3,2)$-net of $3^3=27$ samples in base $3$ with its 3 stratification constraints with $3^1=3$ samples in each stratum, and a $(0,2,2)$-net of $5^2 = 25$ samples in base $5$ with its 3 stratification constraints  with $5^0=1$ sample in each stratum. Increasing the base and increasing $t$ reduces the number of constraints imposing uniformity. We ensure $t=0$ for arbitrary prime base. 
    \label{fig:t1}}
\end{figure}

In this paper, we show that similarly to the case of affine polynomials, we can guarantee $t=0$ by considering a set of polynomials $p_i(x) = p(x) + c_i$ of degree $e$ that only differ by a constant $c_i$. In this case, we show that the Sobol' constructed matrices also amount to powers of Pascal matrices, tensorized by an $e \times e$ matrix, and that $t=0$ as soon as the first few Sobol' iterations yield $t=0$ point sets. 
To ensure irreducibility of these polynomials and compatibility with Sobol', we then use Artin-Schreier irreducible polynomials of the form 
$p_i(x) = x^b - x + c_i$ in $\mathrm{GF}(b)$ with prime $b$ and $c_i \in \mathrm{GF}(b)^\times$ to obtain $b-1$ irreducible polynomials in base $b$. 
We further find that initializing these Sobol' recurrences with powers of Pascal matrices left and right-multiplied by diagonal matrices always ensures that $t=0$, and provide a greedy solution to explore $t=0$ candidate initializations.
We finally conclude with numerical experiments to evaluate the uniformity of our $(b-1)$-dimensional Artin-Schreier sequences as well as $(2b-1)$-dimensional sequences obtained by concatenating coordinates obtained by the $b$ linear polynomials and the $b-1$ Artin-Schreier polynomials.

\section{Background}
\label{sec:background}

\paragraph*{Digital nets}
To produce a well-distributed set of $b^m$ sample points in $s$ dimensions, the idea of digital nets is to first construct $s$ matrices $\{M_i\}_{i=1\dots s}$ of size $m \times m$ in the Galois Field\footnote{For those unfamiliar with Galois theory, you may consider that for $b$ prime, this amounts to integers modulo $b$} $\mathrm{GF}(b)$. To construct the $j-$th sample, one decomposes the index $j$  in base $b$ to form a column vector of digits (e.g., in base $b=2$, this is the binary representation of $j$). Then each matrix multiplies this column vector to obtain $s$ new column vectors, that are reinterpreted as $s$ coordinates in $[0, 1[$ by considering these vectors as base-$b$ digits with weights $b^{-1}, b^{-2}, \dots$. The uniformity of samples produced with this approach is dictated by the property of matrices $\{M_i\}_{i=1\dots s}$.
For digital nets, the quality factor $t$ can be obtained directly from matrices~\cite{niederreiter1992random,marion2020algorithm}. In particular, a sequence with $t=0$ amounts to ensuring that if one builds a matrix by stacking the $d_i$ first rows of matrix $M_i$, with $\sum_{i=1..s} d_i = m$, the $m\times m$ compound matrix $\cK$ has $\text{det}(\cK) \neq 0$ (in $\mathrm{GF}(b)$)  for all $m$ and all $\{d_i\}$.

\paragraph*{Sobol' construction}
A well-established way of producing good matrices is through Sobol' construction~\cite{sobol1967distribution}.
Given a primitive polynomial $p$ of degree $e$ in base $b$, the original Sobol' approach recursively builds an upper triangular matrix one column at a time, as a combination of the $e$ previous columns. One matrix is produced for each dimension, and thus $s$ dimensions require $s$ polynomials and $e \times e$ matrices to initialize the recursion.
Numerical search over good $e \times e$ matrices was notably performed by Joe and Kuo~\cite{joe2008constructing} for $b=2$, and then Ostromoukhov et al.~\cite{OBCI24} for $b=3$.

The core of our paper uses the construction of Faure and Lemieux~\cite{faure2016irreducible} who obtained the same Sobol' matrices by operating the recursive construction blockwise rather than columnwise. For a Sobol' matrix $M$ and polynomial $p(x) = \sum_{i=0}^e a_i x^i$ of degree $e$, their recurrence reads:
\begin{equation}
\label{eq:faurelemieux}
 M_{i,j} = \left(M_{i,j-1} Q_p + M_{i-1,j-1}\right) F_p 
\end{equation}
where $i,j$ index the $i^{\text{th}}$ block row and $j^{\text{th}}$ block column of matrix $M$, with blocks of size $e \times e$, and $Q_p$ and (inverse of) $F_p$ are matrices\footnote{We identified sign errors in $Q_p$ and $F_p$ in the original paper of Faure and Lemieux~\cite{faure2016irreducible} that do not affect base $b=2$ but affect other bases} given by~\cite{faure2016irreducible,BCIO25}:

\begin{equation}
\label{eq:q}
Q_p = 
\left[
\begin{array}{ccccc}
 -a_0 & 0 & 0 & \dots & 0 \\
 -a_1 & -a_0 & 0 & \dots & 0 \\
 -a_2 & -a_1 & -a_0 & \dots & 0 \\
 \vdots & \vdots & \vdots & \ddots & \vdots \\
 -a_{e-1} & -a_{e-2} & -a_{e-3} & \dots & -a_0 \\
\end{array}
\right]
\end{equation}

\begin{equation}
\label{eq:f}
F^{-1}_p = 
\left[
\begin{array}{ccccc}
 1 & -a_{e-1} & -a_{e-2} & \dots & -a_{1} \\
 0 & 1 & -a_{e-1} & \dots & -a_{2} \\
 0 & 0 & 1 & \dots & -a_{3} \\
 \vdots & \vdots & \vdots & \ddots & \vdots \\
 0 & 0 & 0 & \dots & 1 \\
\end{array}
\right]\,.
\end{equation}

Faure and Lemieux show that using the larger class of irreducible polynomials instead of primitive polynomials, with $b$ a prime power, also yields $(t,s)$-sequences.

\section{Overview}

Our paper investigates conditions that guarantee the highest quality $t=0$ in the context of Sobol' sequences in $\mathrm{GF}(b)$.
In base $b$, one may find $s$ irreducible degree-$e$ polynomials that only differ by a constant (we call them \textit{consecutive polynomials}), although they may not always exist. For instance, in base $b=5$, we may find consecutive degree $e=5$ polynomials $x^5+4x+1$, $x^5+4x+2$, $x^5+4x+3$ and $x^5+4x+4$ that are all irreducible.
In section~\ref{sec:pascal}, we show that given 2 consecutive polynomials, the Sobol' matrices produced are equivalent to powers of Pascal matrices tensorized by $e \times e$ matrices. This relates our construction to Faure's construction of $(0,s)$-sequences that corresponds to the special case where $e=1$~\cite{faure1982discrepance}. However, while in the case $e=1$ the $(0,s)$-sequence condition is automatically satisfied, for $e>1$ this is not always the case.
In section~\ref{sec:conditions}, we  show that testing the $(0, m, s)$-net conditions for the first few submatrices guarantees the $(0,s)$-sequence property up to infinity.
In section~\ref{sec:exist}, we use Artin-Schreier theory to obtain $s=b-1$ degree-$b$ consecutive irreducible polynomials in $\mathrm{GF}(b)$, and show that initializing the Sobol' construction with powers of Pascal matrices always fulfills the conditions for $t=0$.
In section~\ref{sec:numerics} we numerically experiment our sequences with a simple greedy algorithm to explore good initialization matrices, and evaluate discrepancies of the resulting point sets.

\paragraph*{Settings}
We consider a prime base $b$ and $p(x)\in\mathrm{GF}(b)[x]$ a polynomial of degree~$e\ge1$.
We assume that for $\{c_i\in\mathrm{GF}(b)\}_{i=1\dots s}$, the shifted polynomials $p_i = p+c_i$ are distinct irreducible polynomials of degree~$e$. 
We denote $M^{(c_i)}$ their corresponding Sobol' matrix.

\begin{figure*}[!h]\setlength{\tabcolsep}{1pt}
  \begin{tabular}{c|cccccccccccc}
 \multirow{2}{*}{\rotatebox{90}{\quad$GF(5)$}}&  \rotatebox{90}{~\quad$\{x+c_i\}$}& \includegraphics[width=1.5cm]{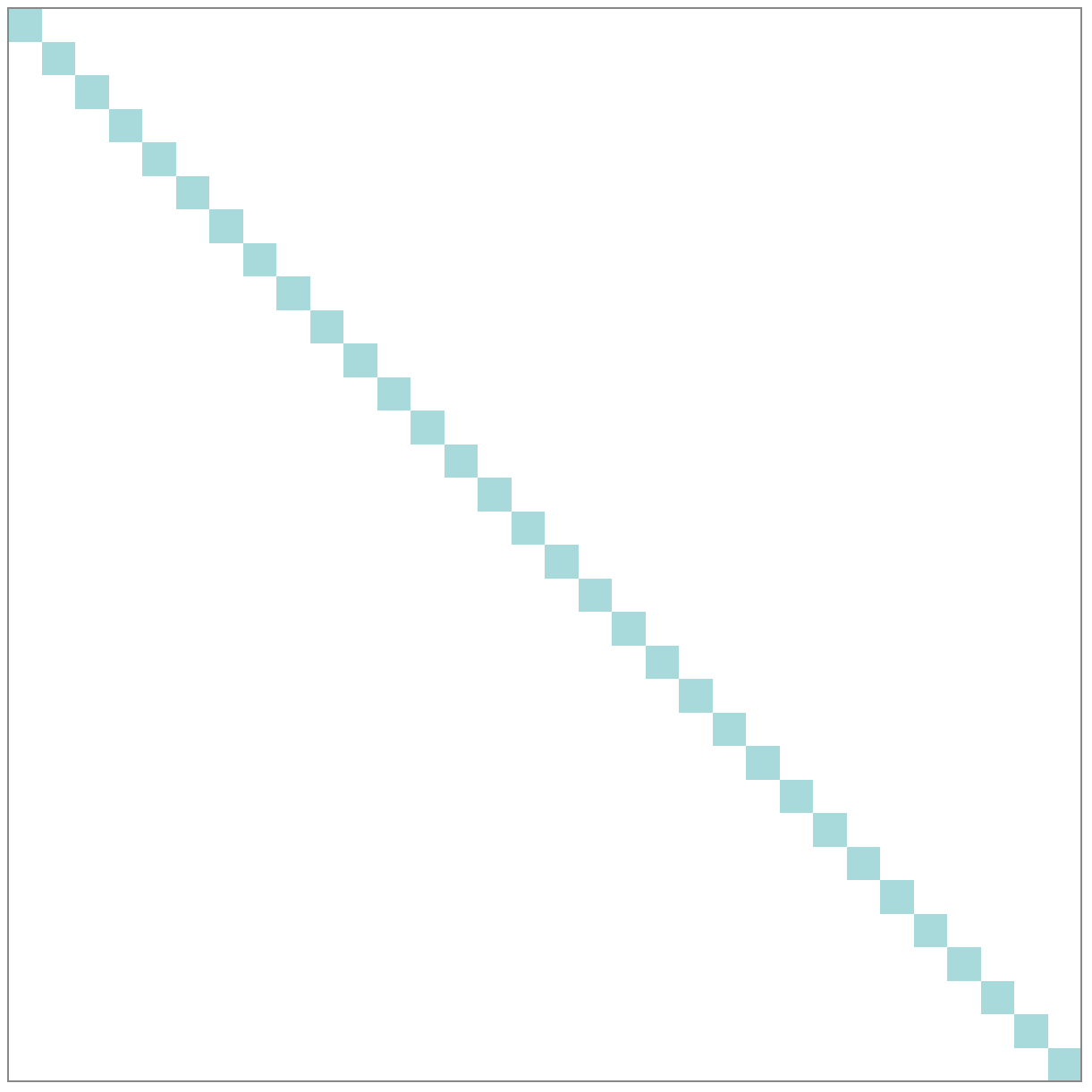}&
    \includegraphics[width=1.5cm]{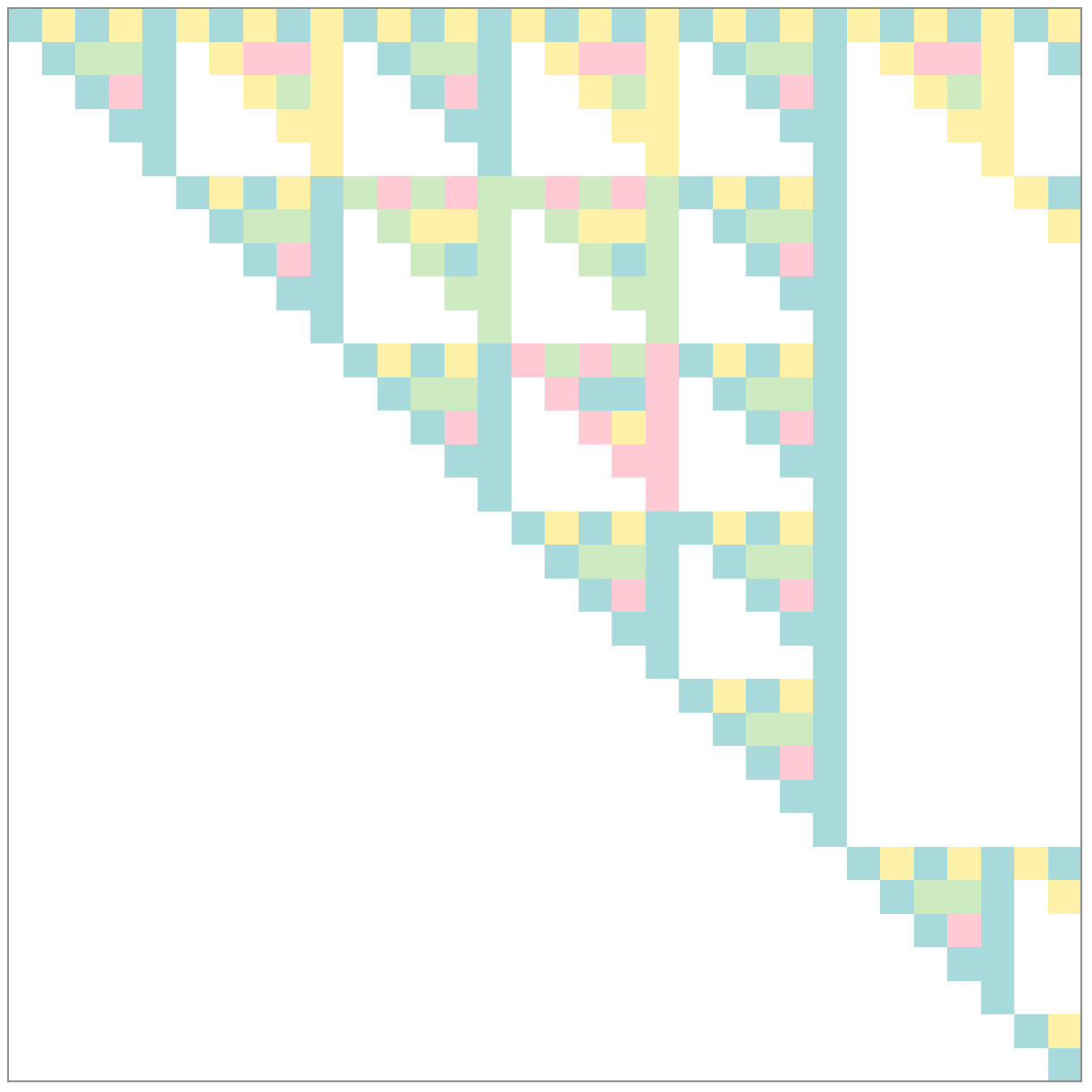}&
    \includegraphics[width=1.5cm]{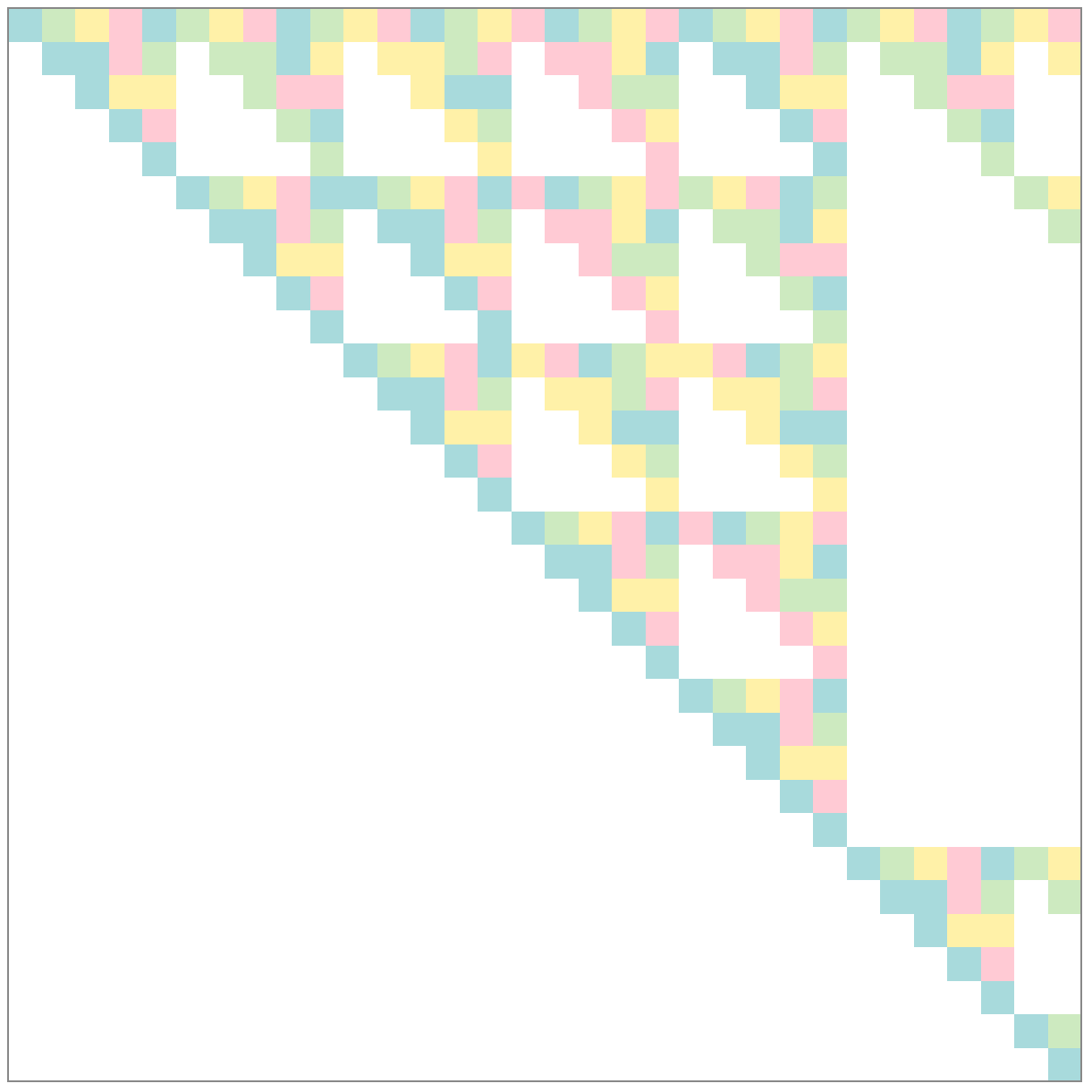}&
    \includegraphics[width=1.5cm]{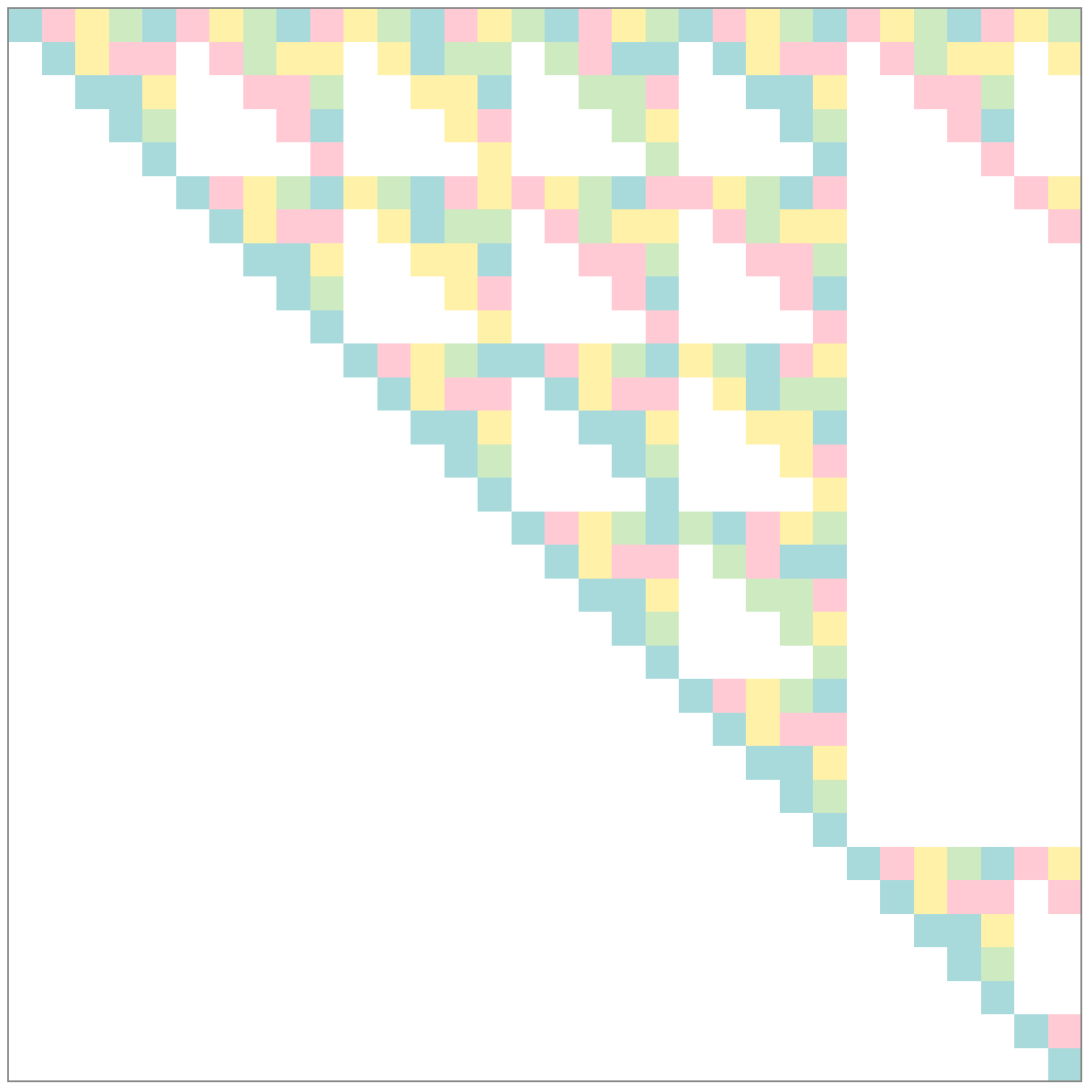}&
    \begin{overpic}[width=1.5cm]{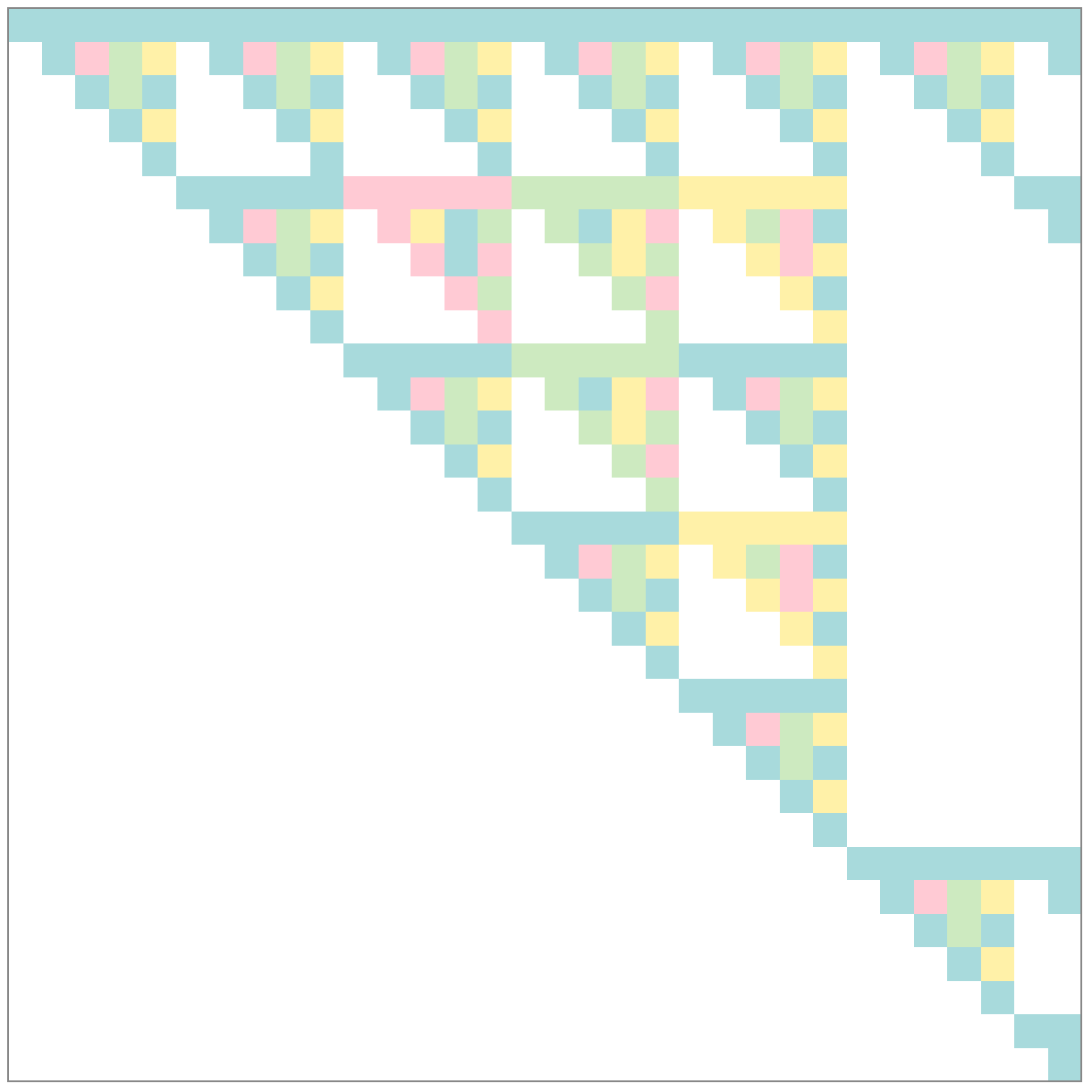}
    \put(600,-100){\includegraphics[width=.4cm]{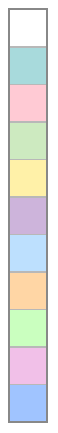}}
    \put(630,-92){10}
    \put(630,-74){9}
    \put(630,-58){8}
    \put(630,-42){7}
    \put(630,-24){6}
    \put(630,-8){5}
    \put(630,8){4}
    \put(630,26){3}
    \put(630,44){2}
    \put(630,62){1}
    \put(630,80){0}  
    \end{overpic}
    &&&&\\
   &\rotatebox{90}{\quad\quad$AS$ } & &
    \includegraphics[width=1.5cm]{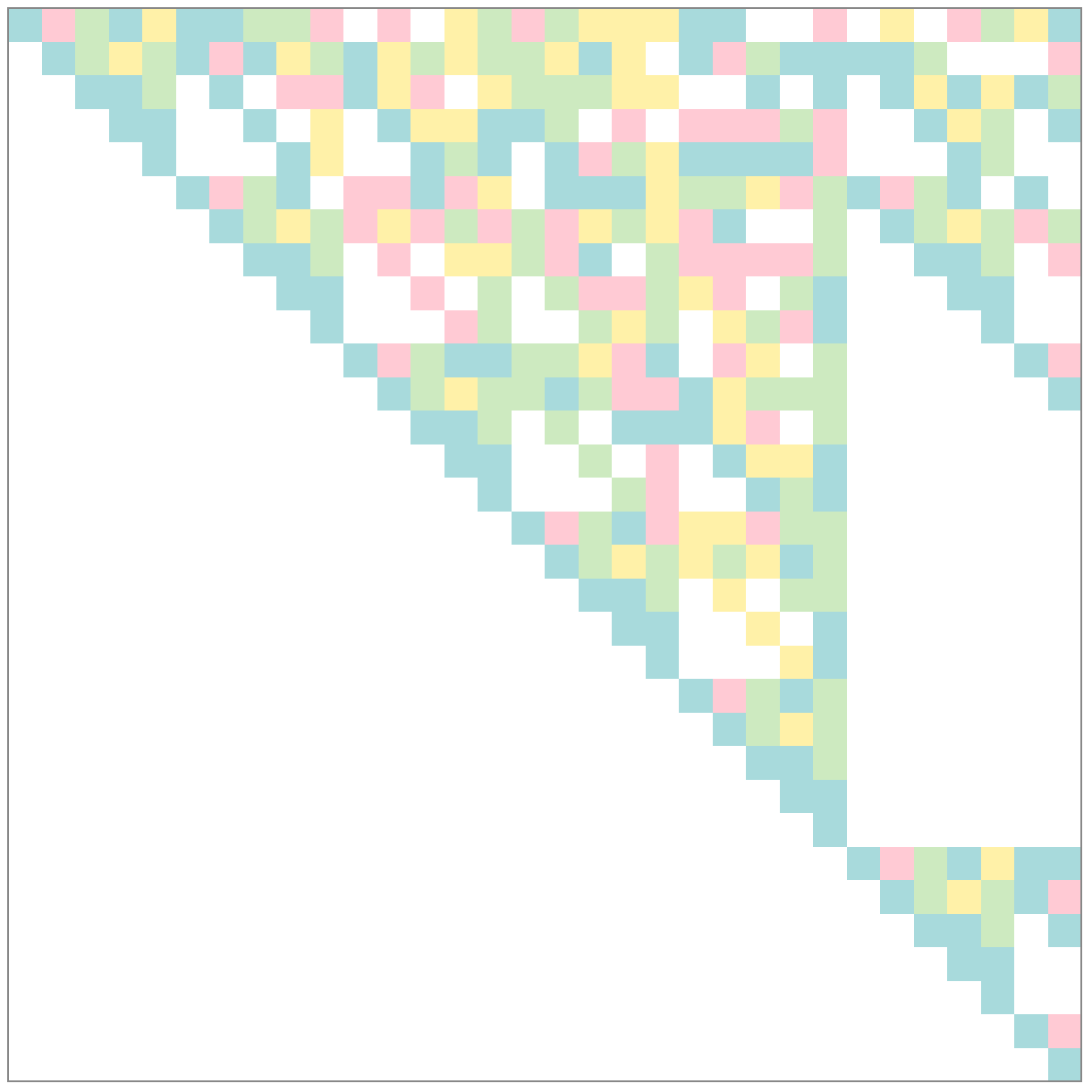}&
    \includegraphics[width=1.5cm]{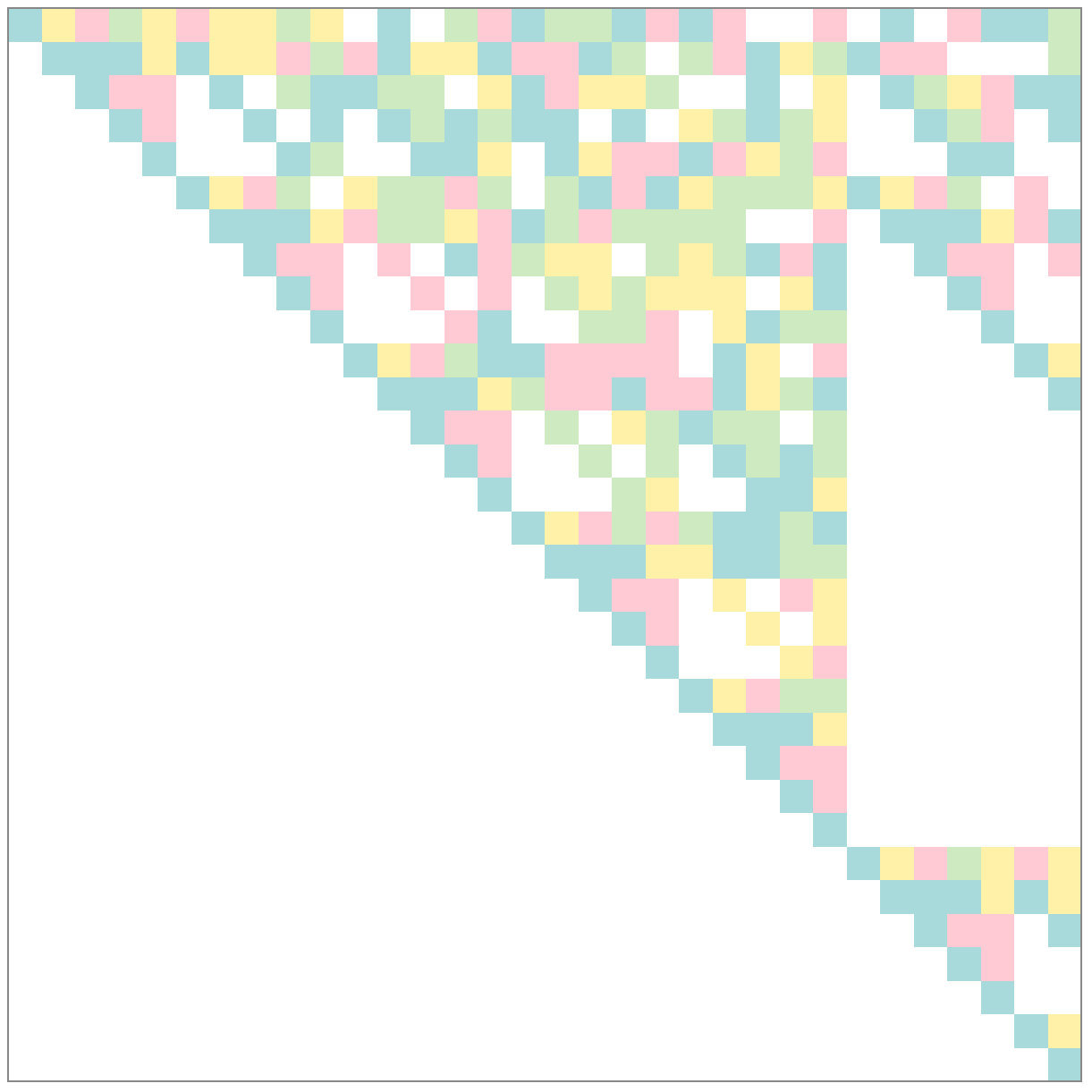}&
    \includegraphics[width=1.5cm]{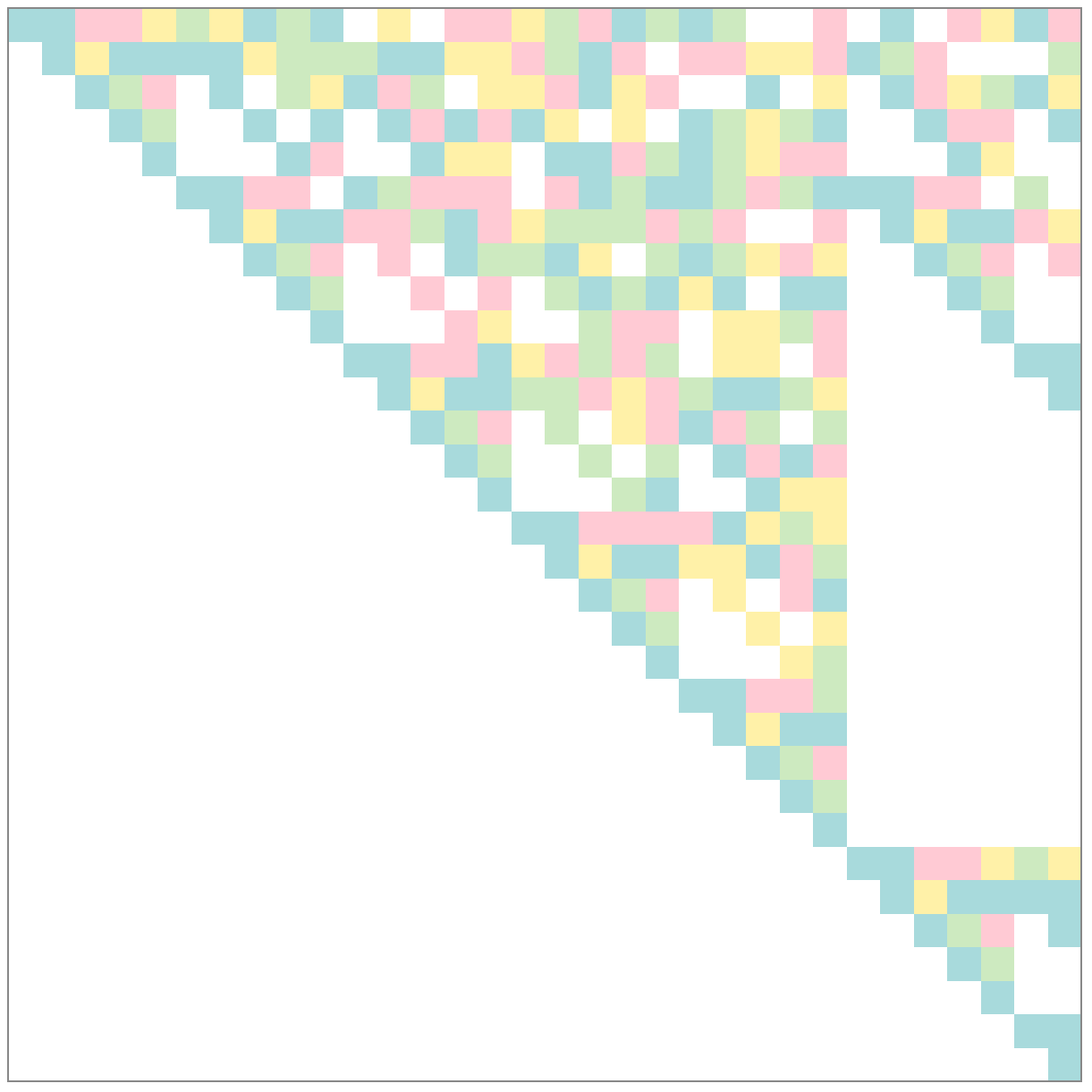}&
    \includegraphics[width=1.5cm]{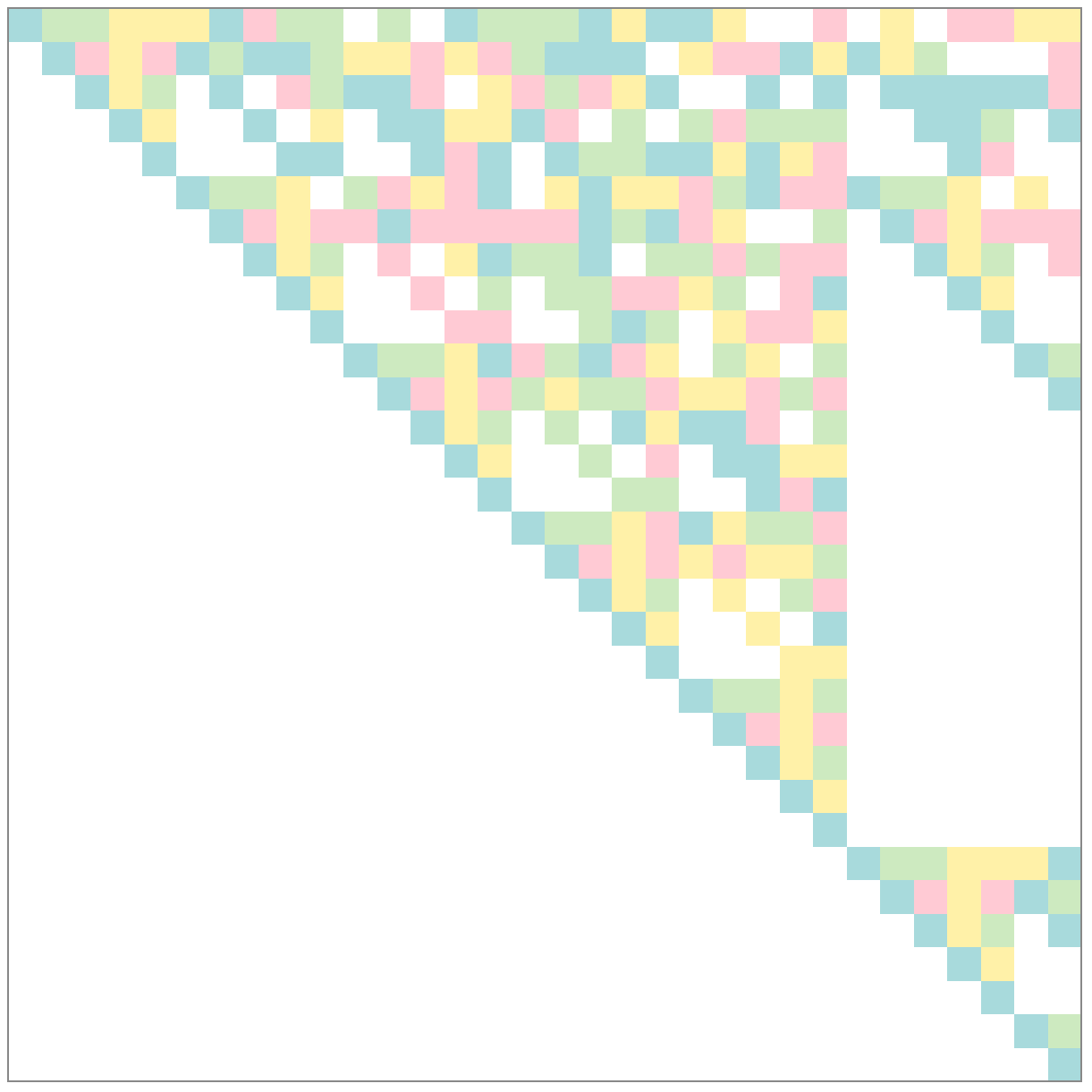}&&&&\\\hline\\
   \multirow{2}{*}{\rotatebox{90}{\quad$GF(7)$}}& \rotatebox{90}{~\quad$\{x+c_i\}$ }&\includegraphics[width=1.5cm]{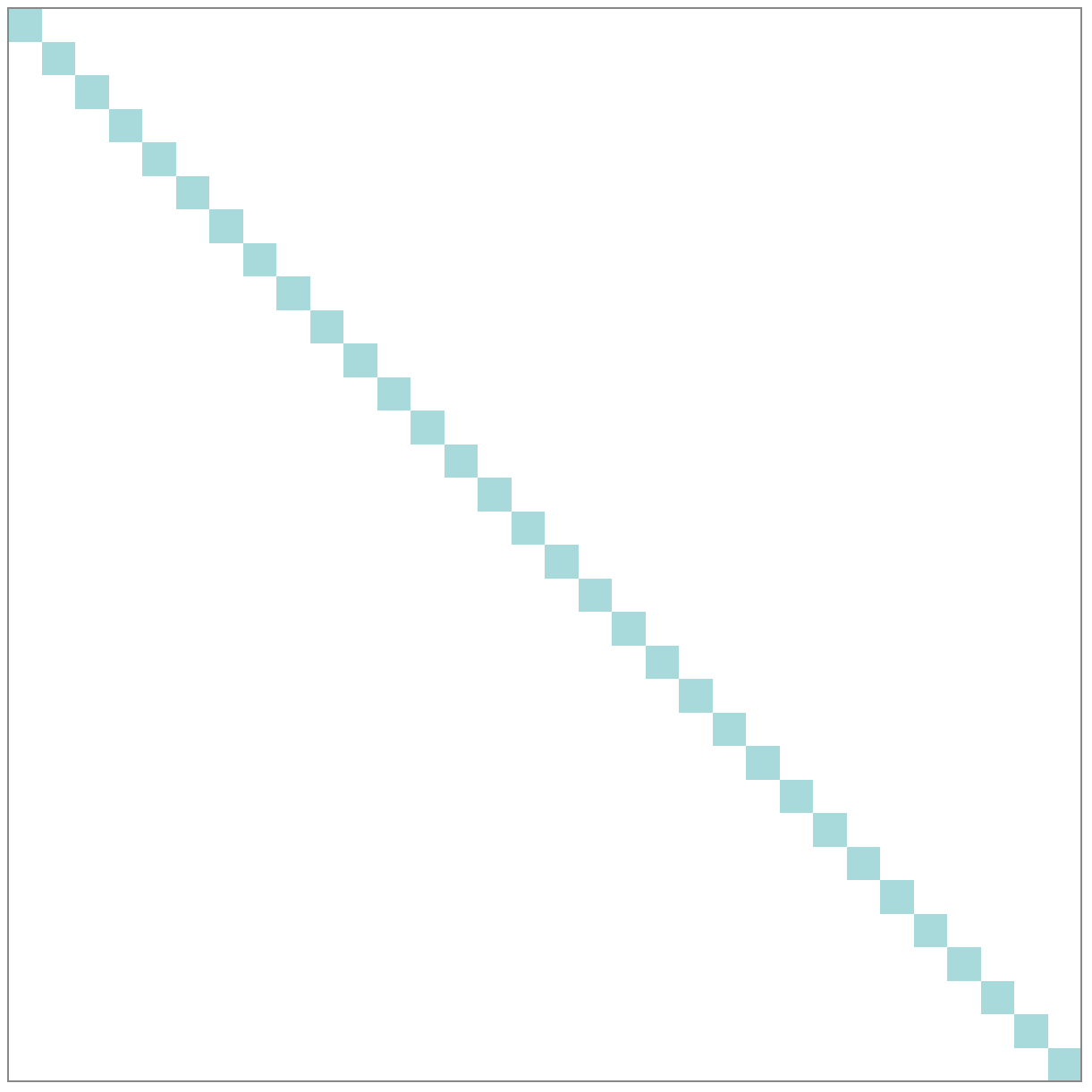}&
    \includegraphics[width=1.5cm]{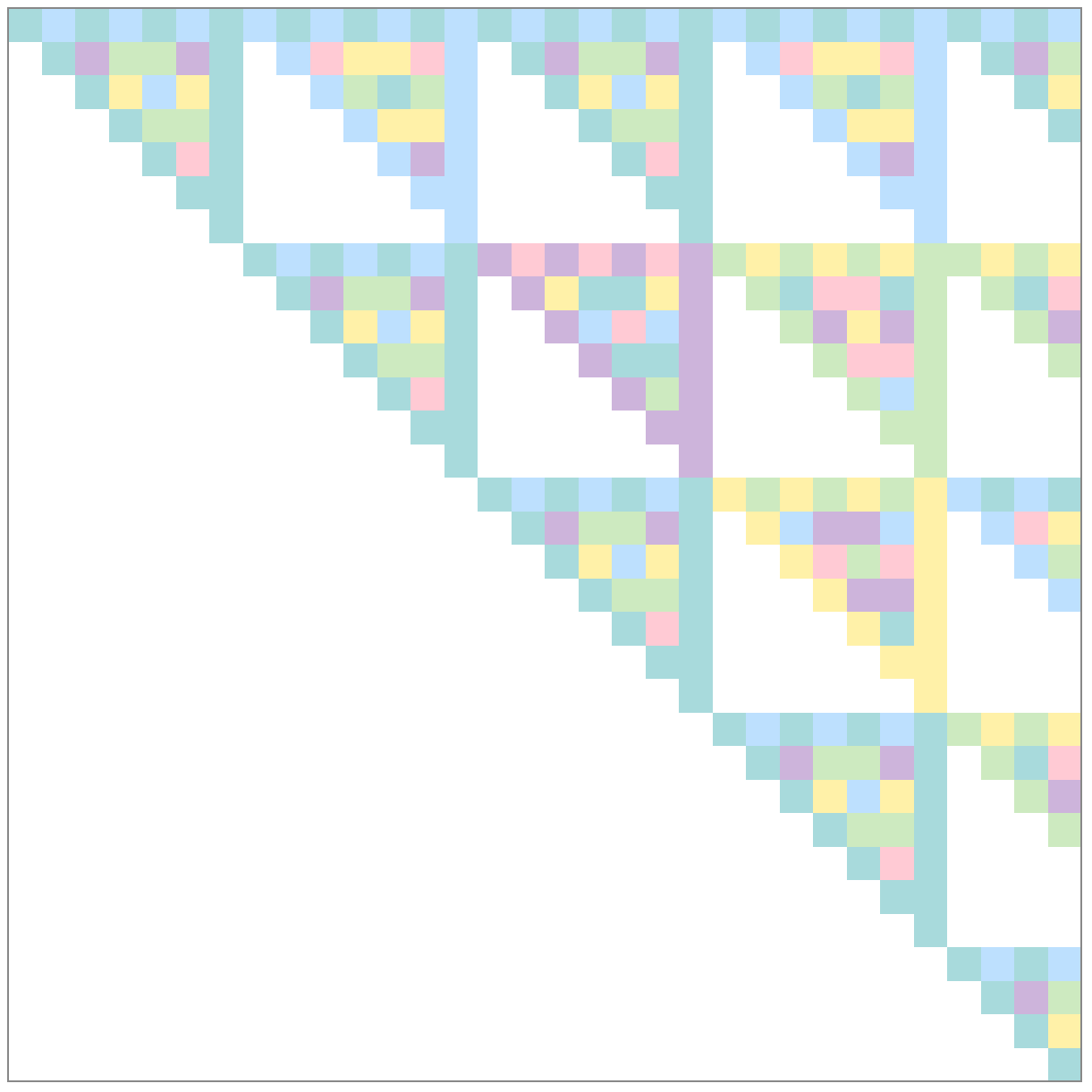}&
    \includegraphics[width=1.5cm]{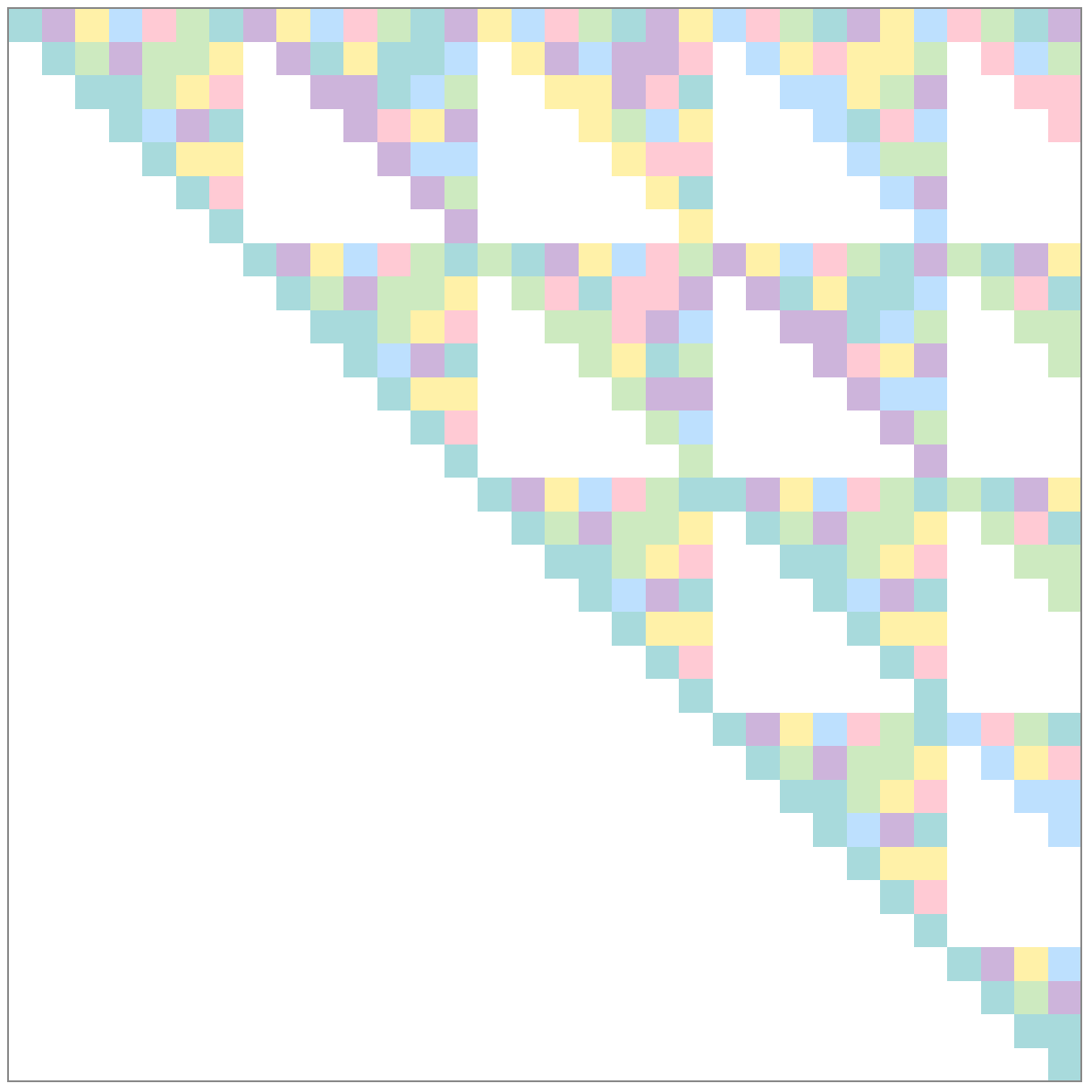}&
    \includegraphics[width=1.5cm]{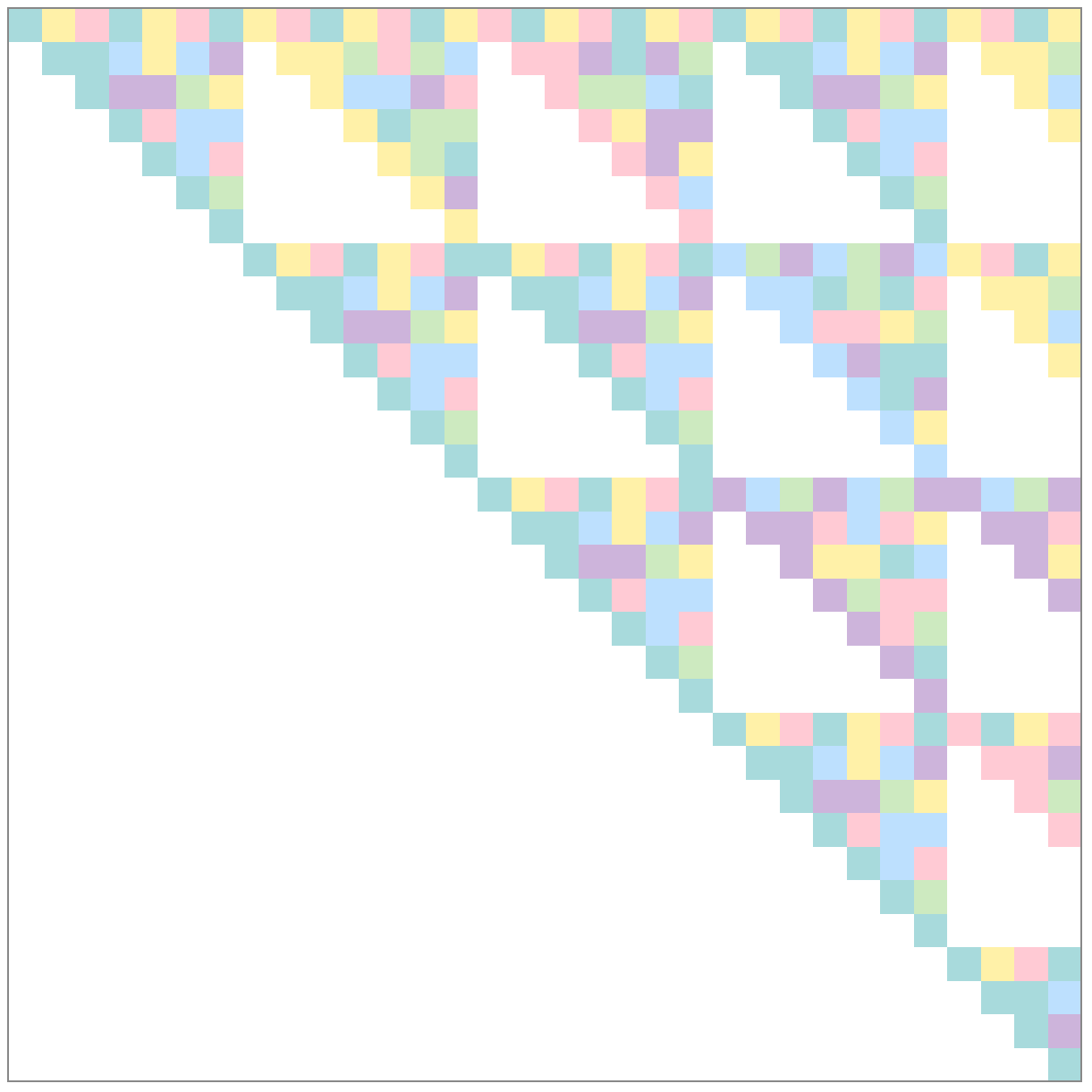}&
    \includegraphics[width=1.5cm]{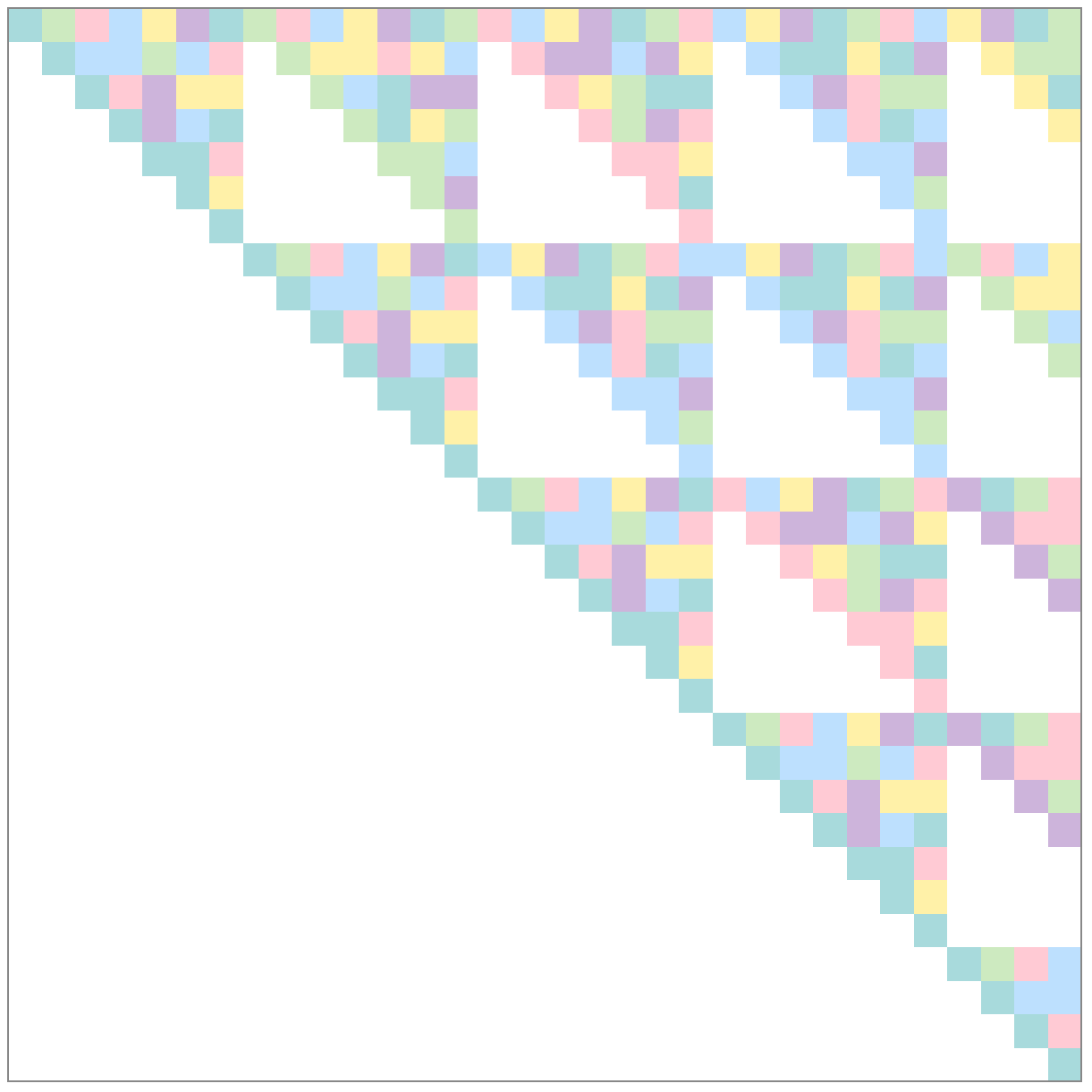}&
    \includegraphics[width=1.5cm]{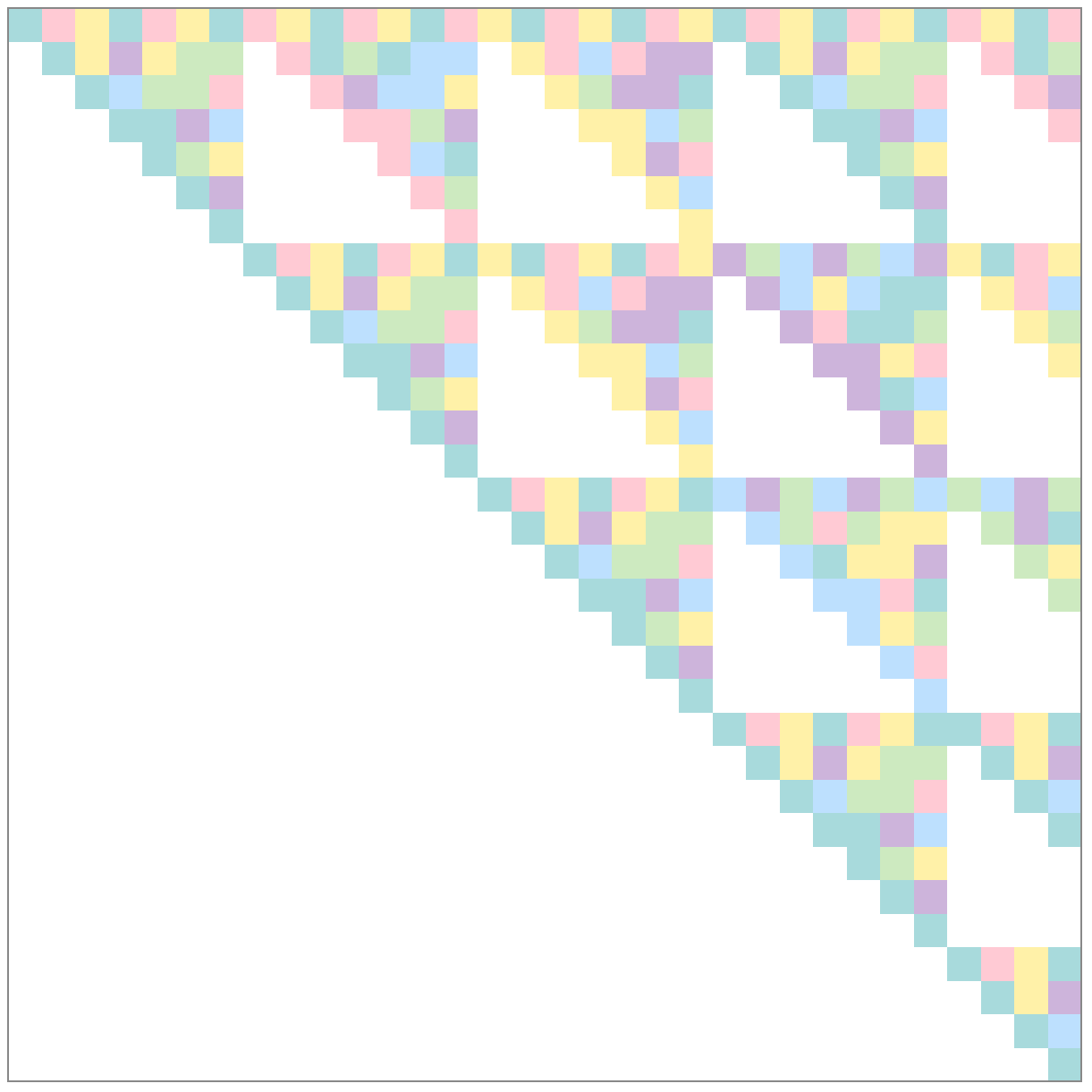}&
    \includegraphics[width=1.5cm]{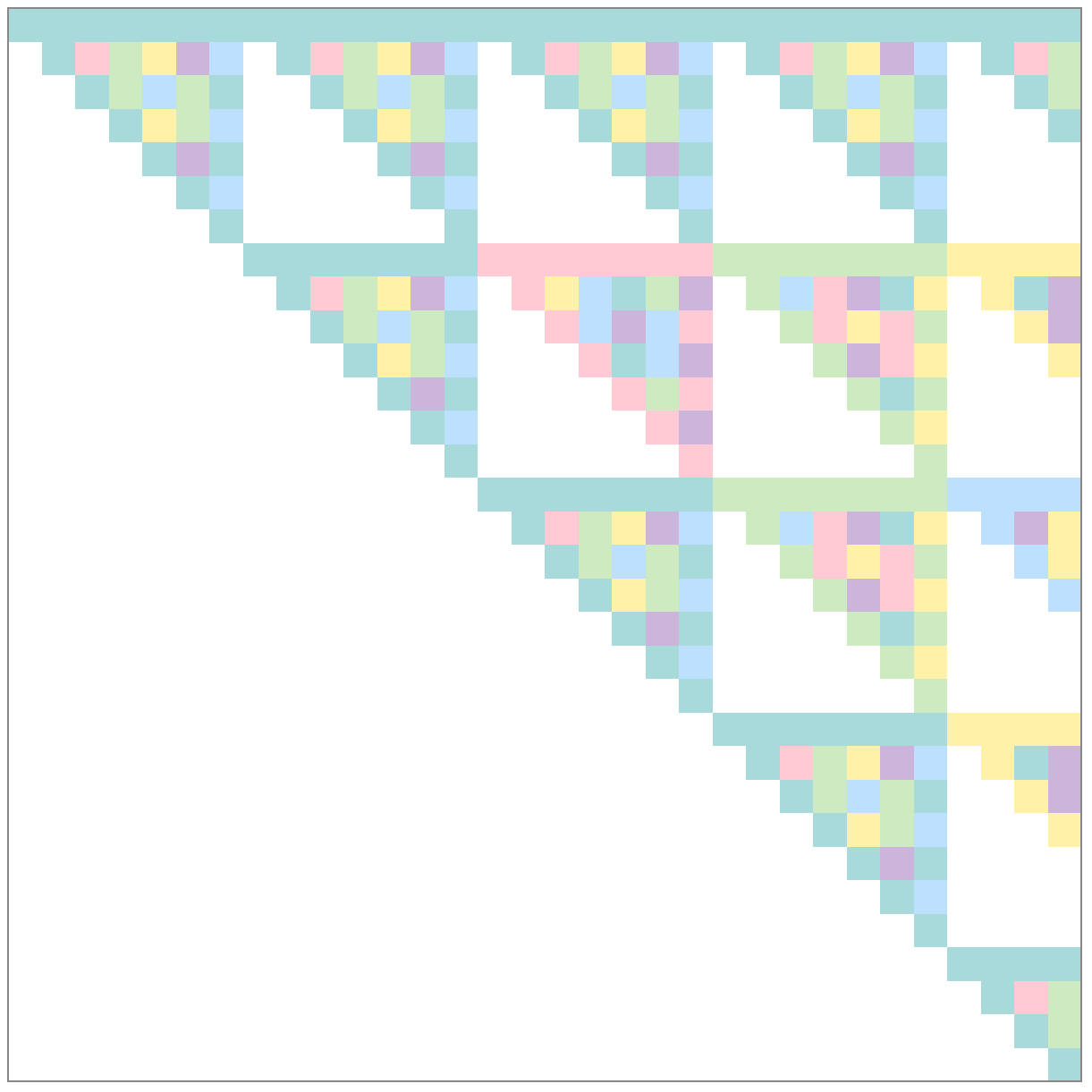}&\\
  & \rotatebox{90}{\quad\quad$AS$ }& & \includegraphics[width=1.5cm]{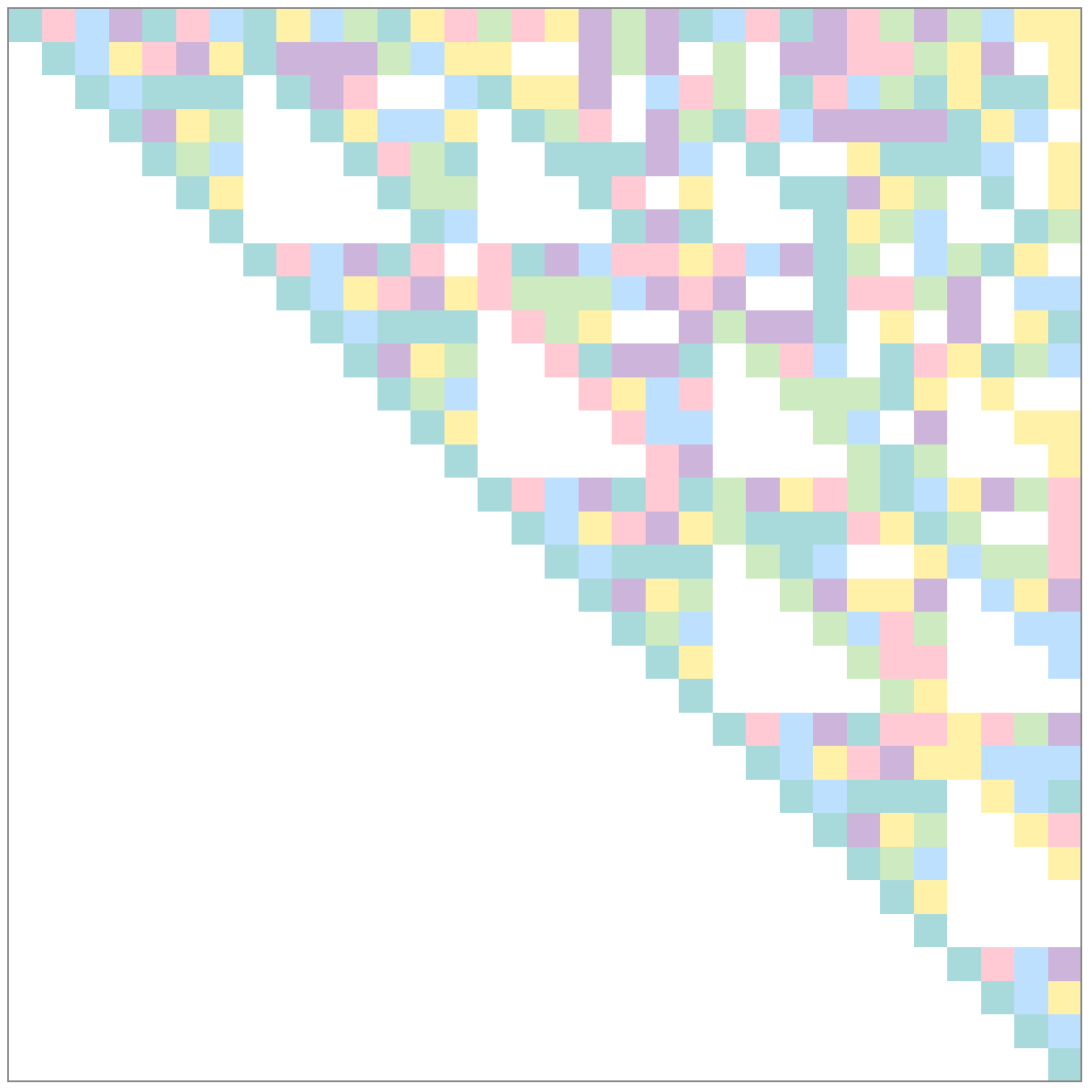}&
    \includegraphics[width=1.5cm]{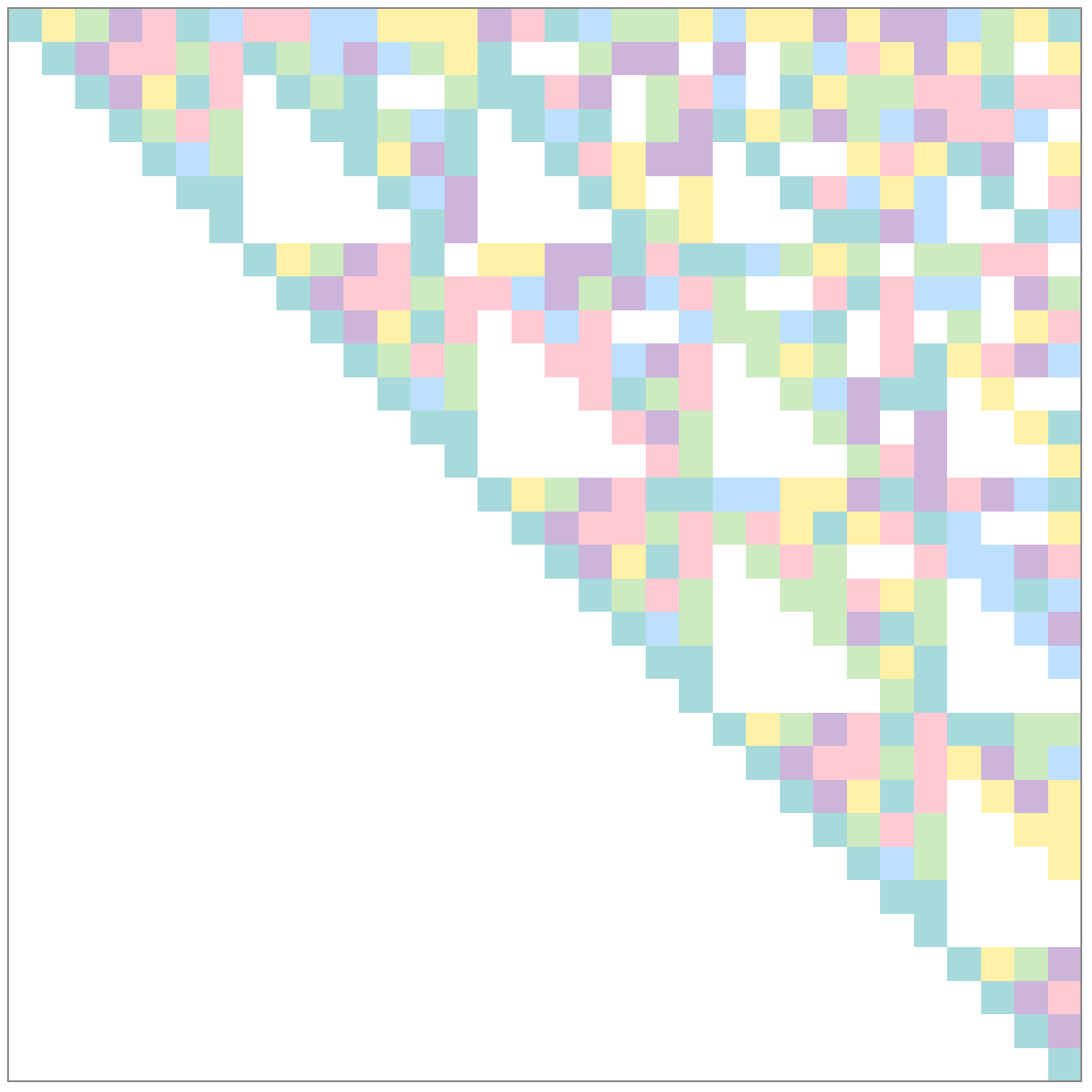}&
    \includegraphics[width=1.5cm]{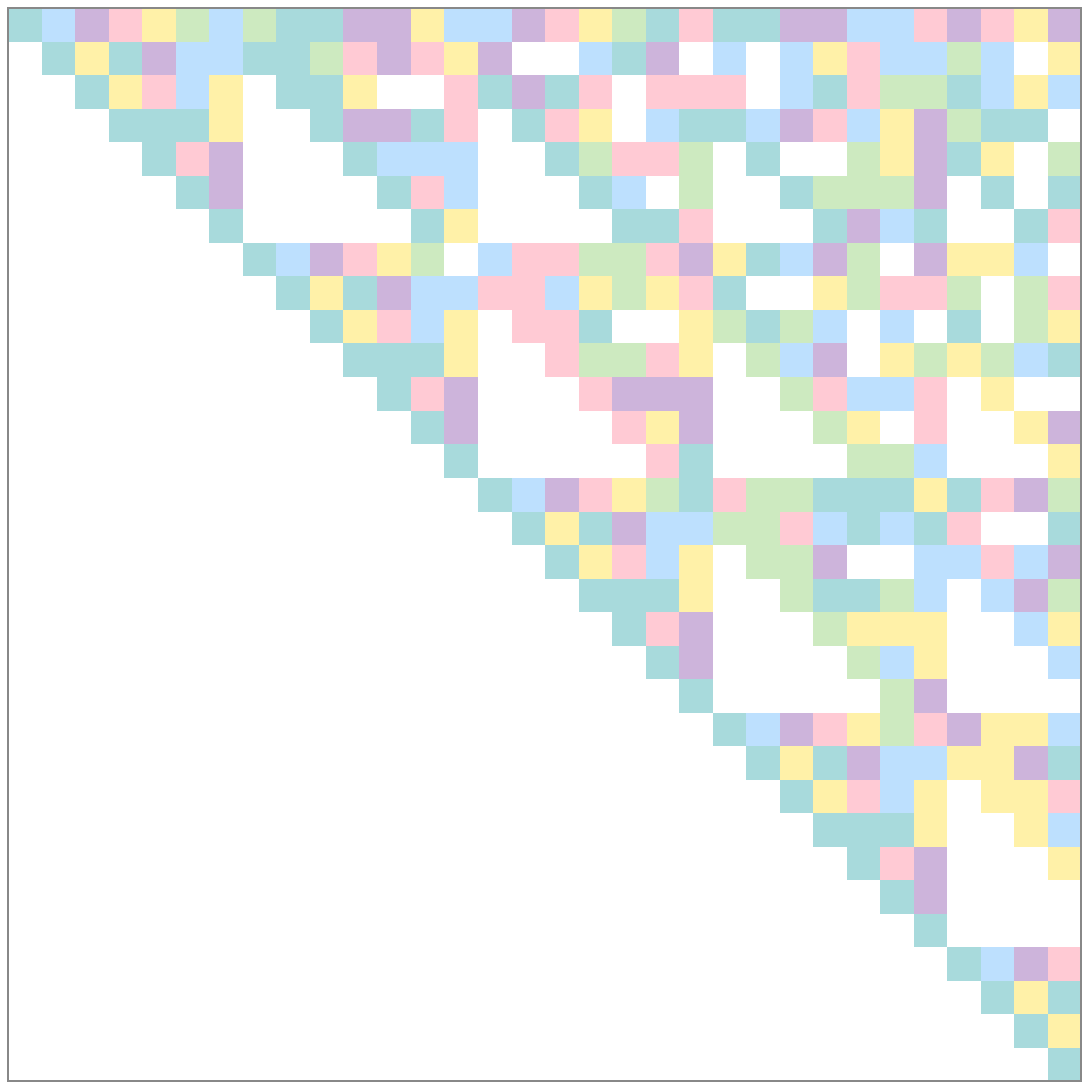}&
    \includegraphics[width=1.5cm]{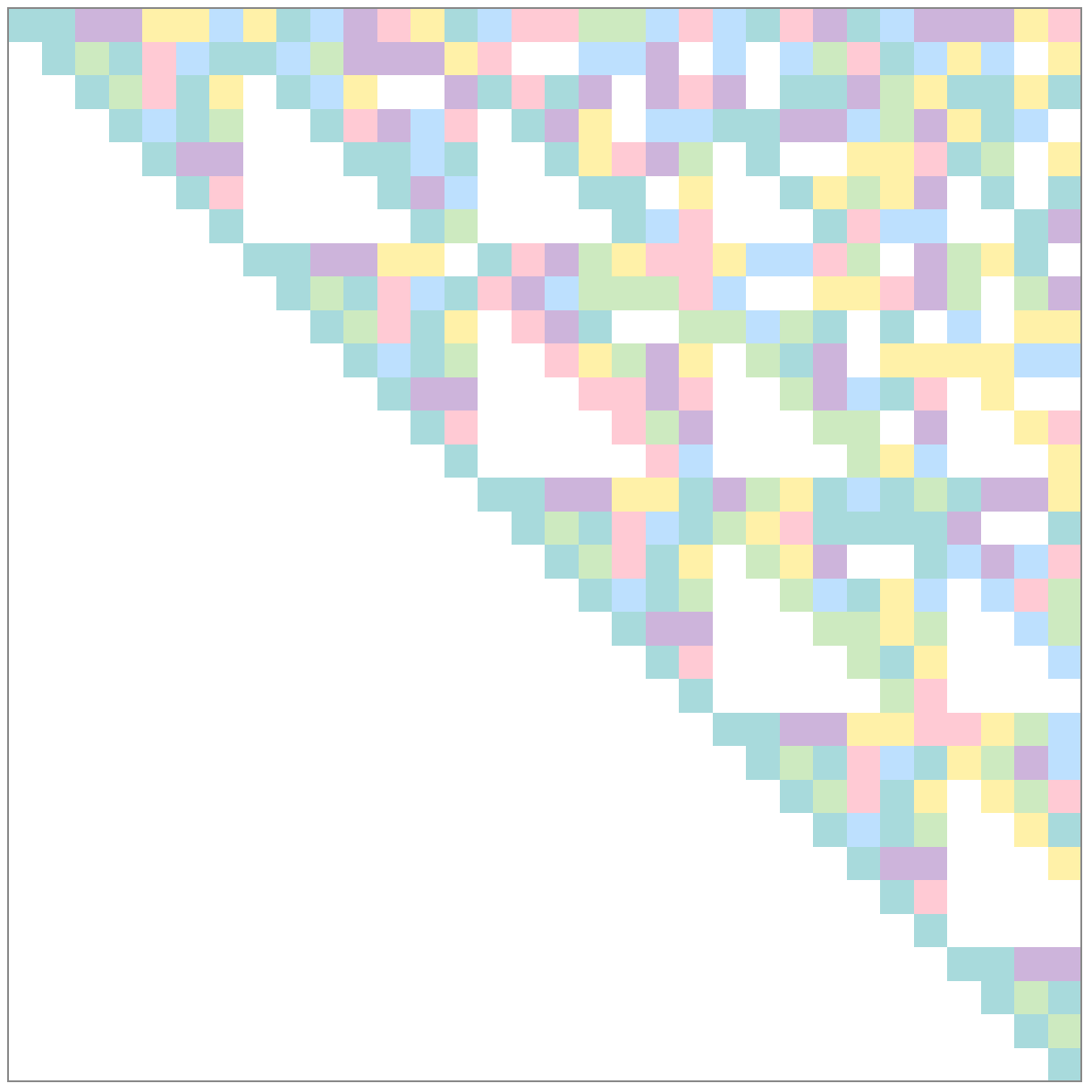}&
    \includegraphics[width=1.5cm]{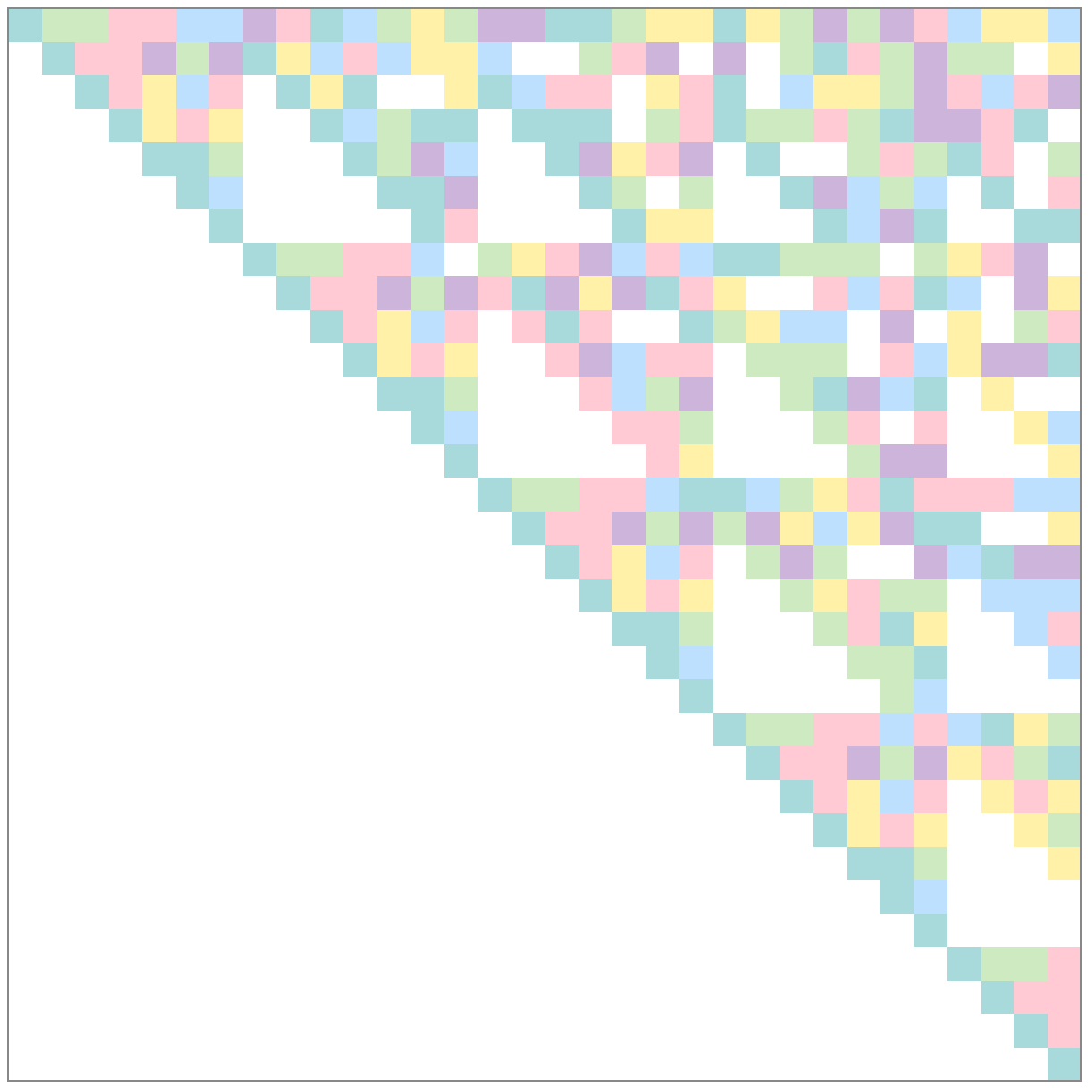}&
    \includegraphics[width=1.5cm]{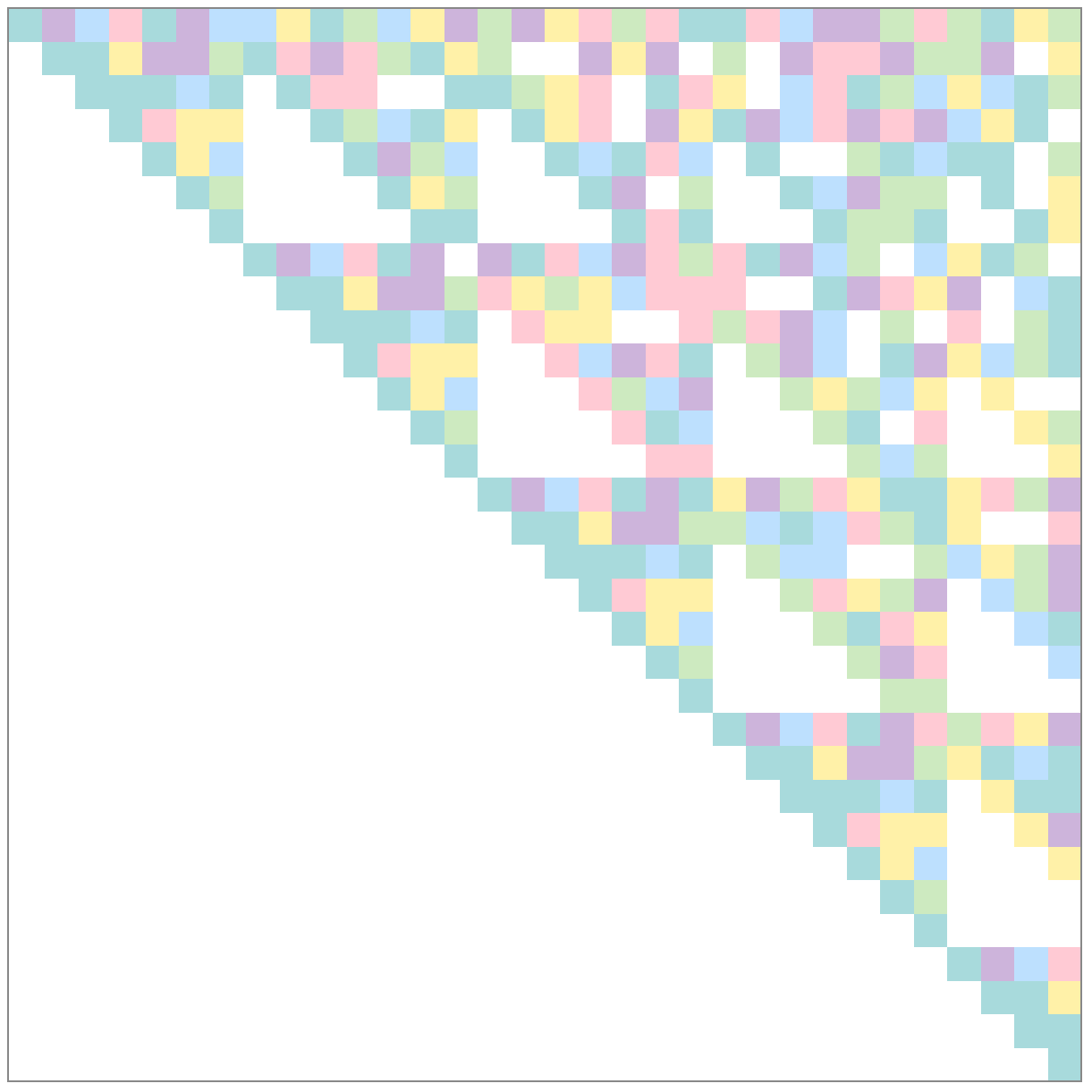}&\\\hline\\
   \multirow{2}{*}{\rotatebox{90}{\quad$GF(11)$}}&  \rotatebox{90}{~\quad$\{x+c_i\}$}&\includegraphics[width=1.5cm]{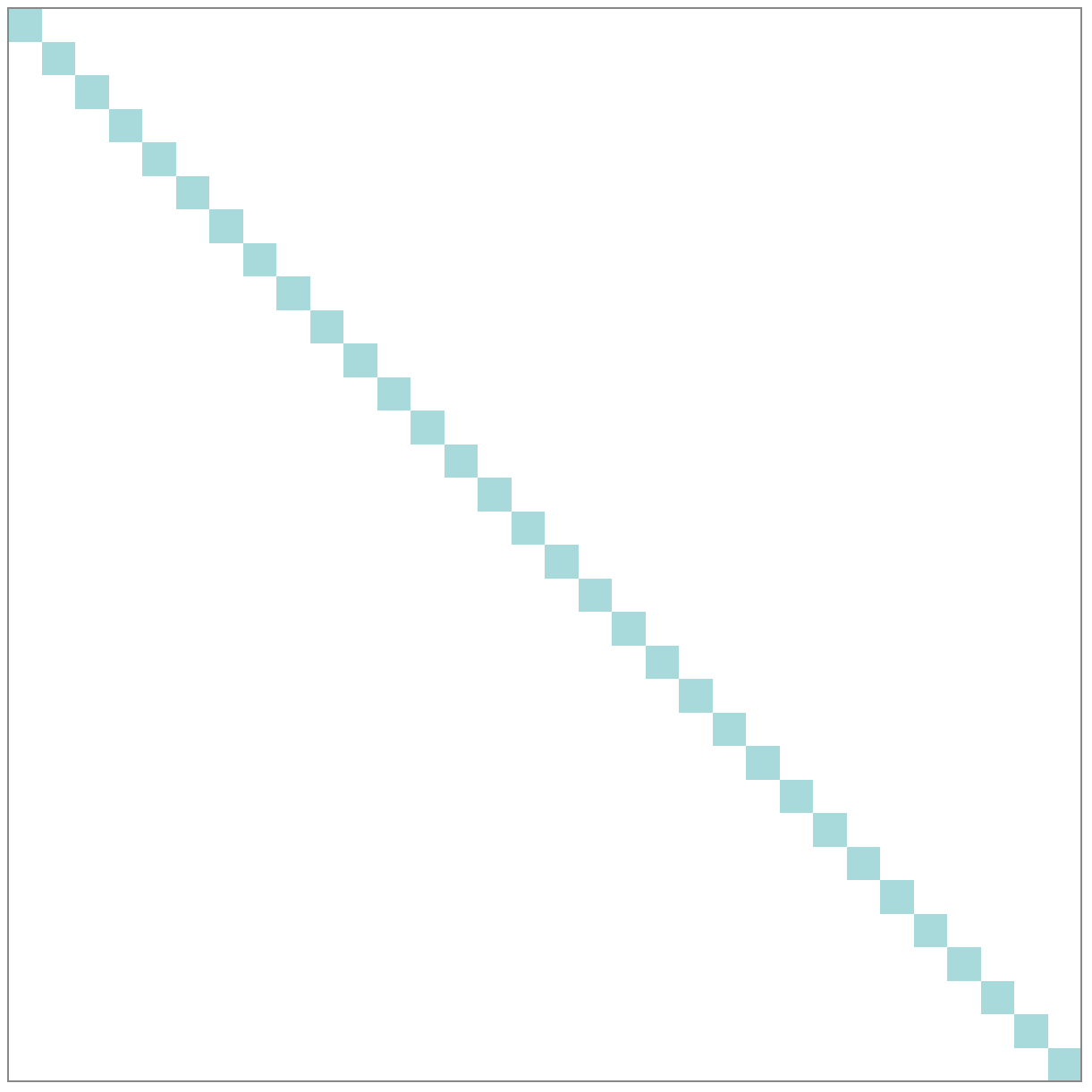}&
    \includegraphics[width=1.5cm]{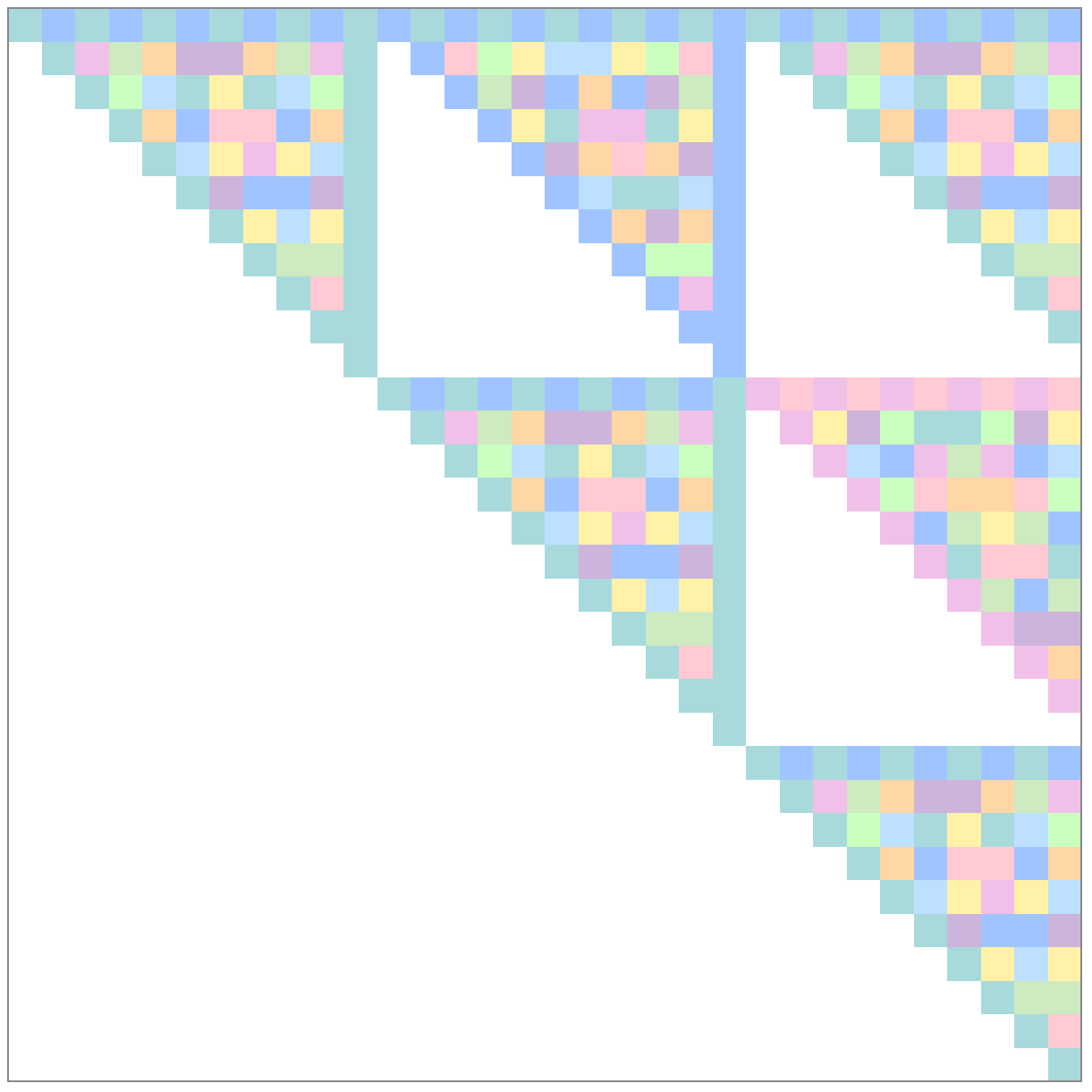}&
    \includegraphics[width=1.5cm]{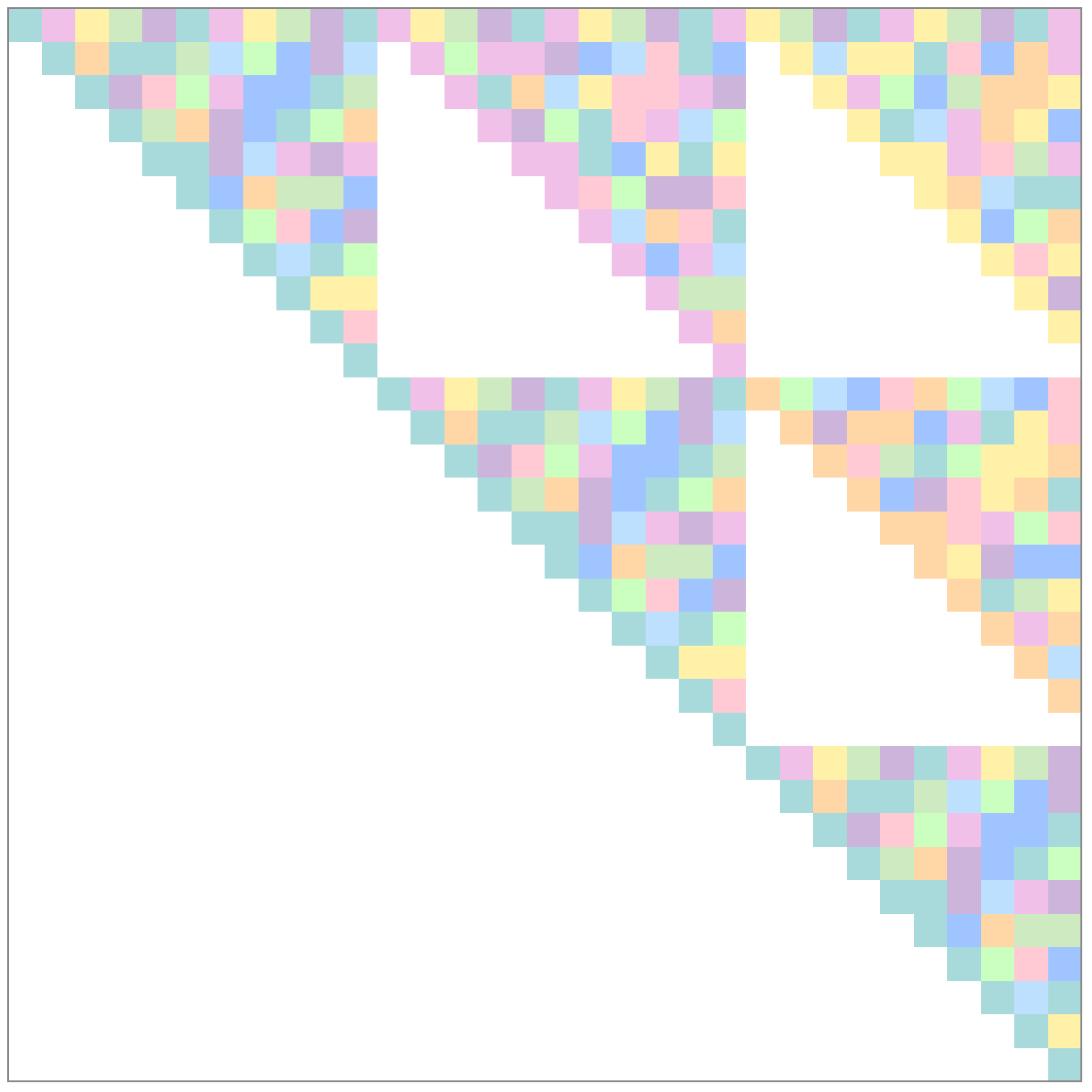}&
    \includegraphics[width=1.5cm]{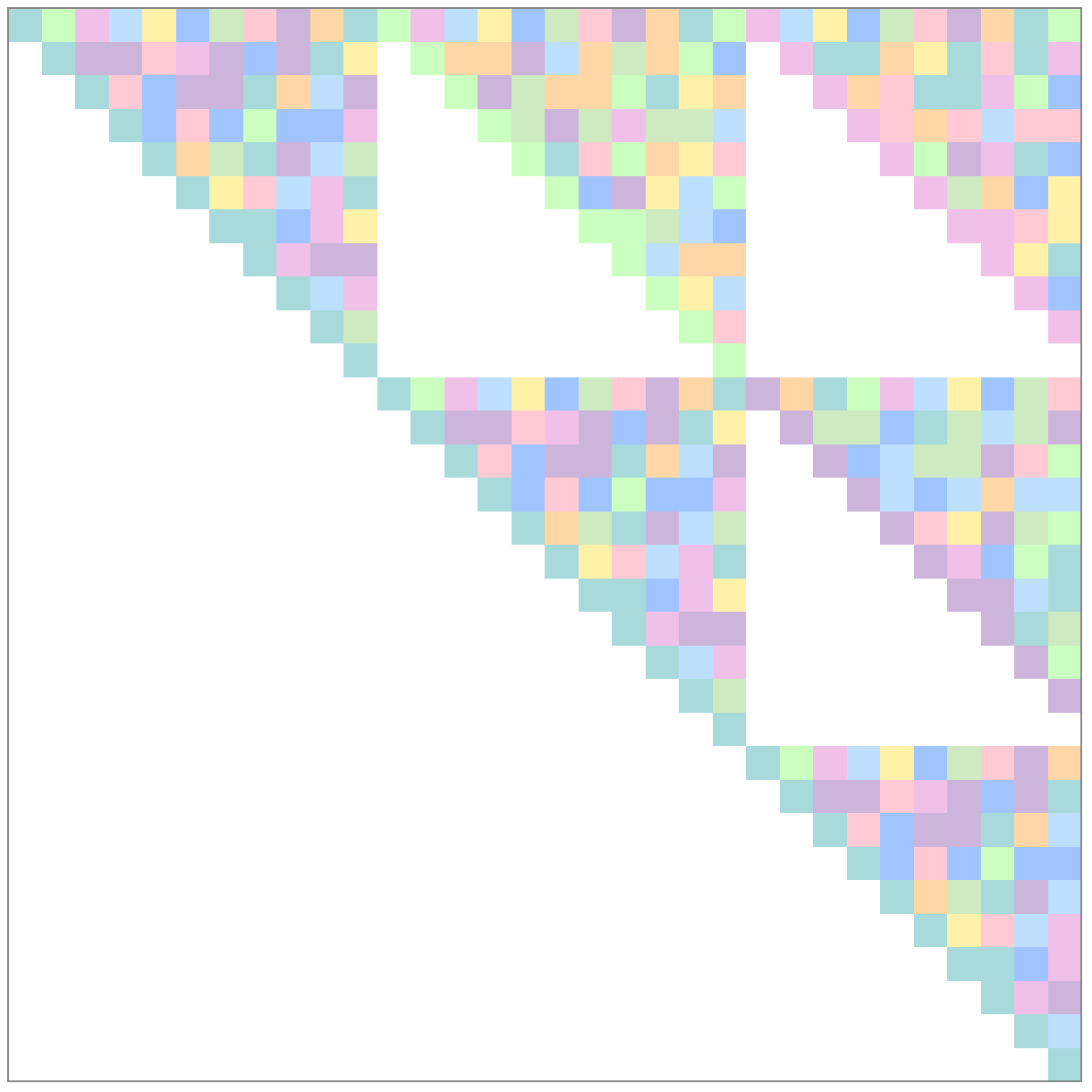}&
    \includegraphics[width=1.5cm]{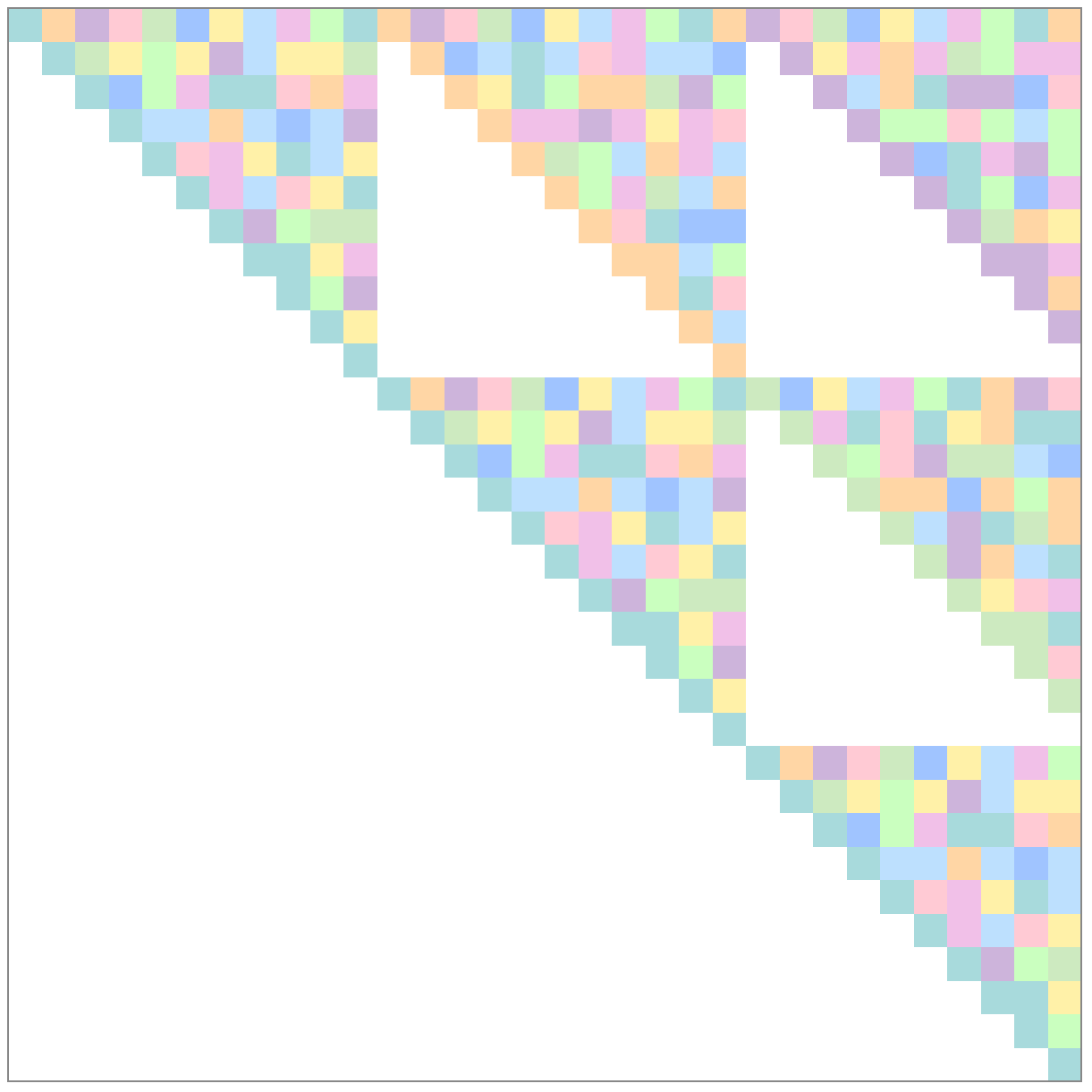}&
    \includegraphics[width=1.5cm]{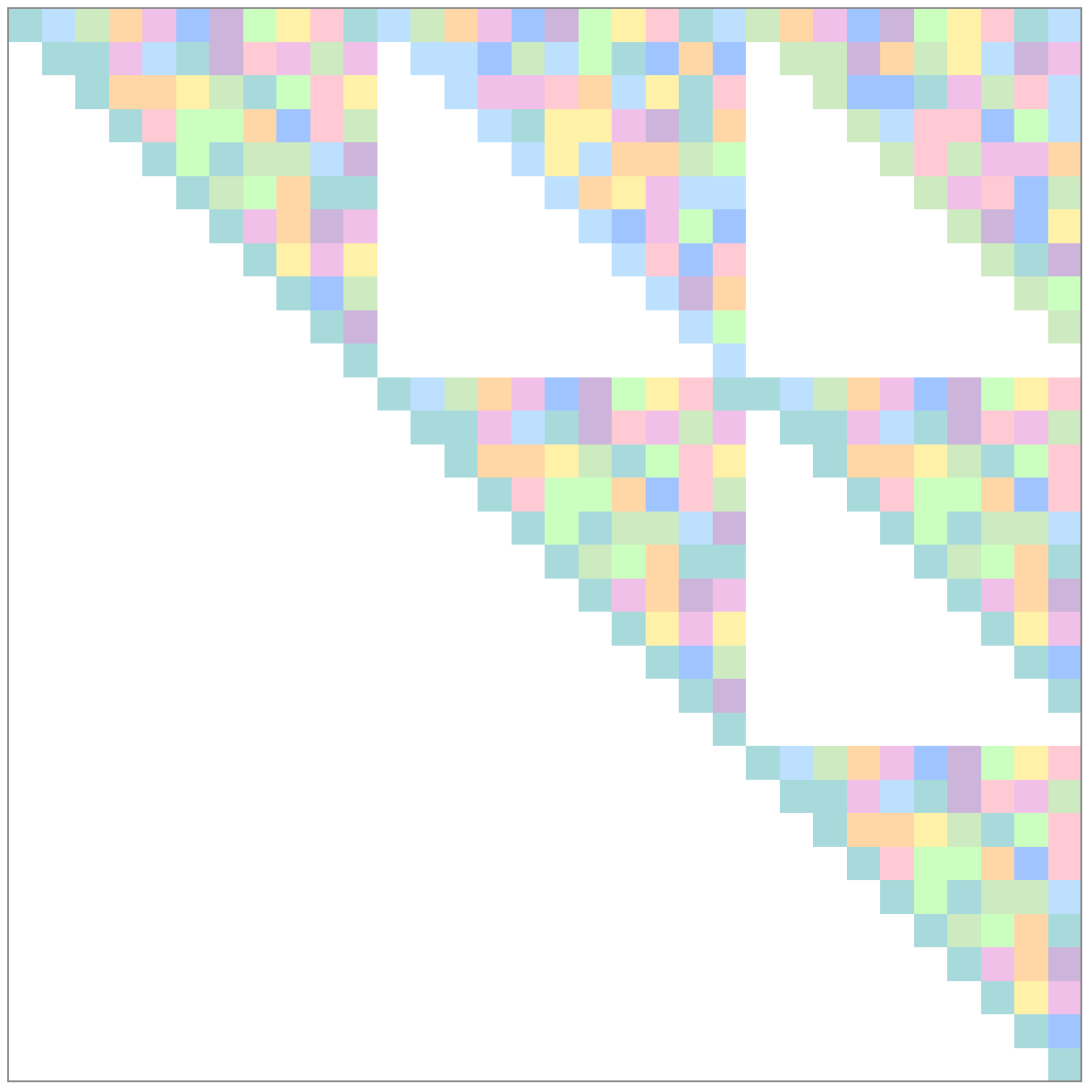}&
    \includegraphics[width=1.5cm]{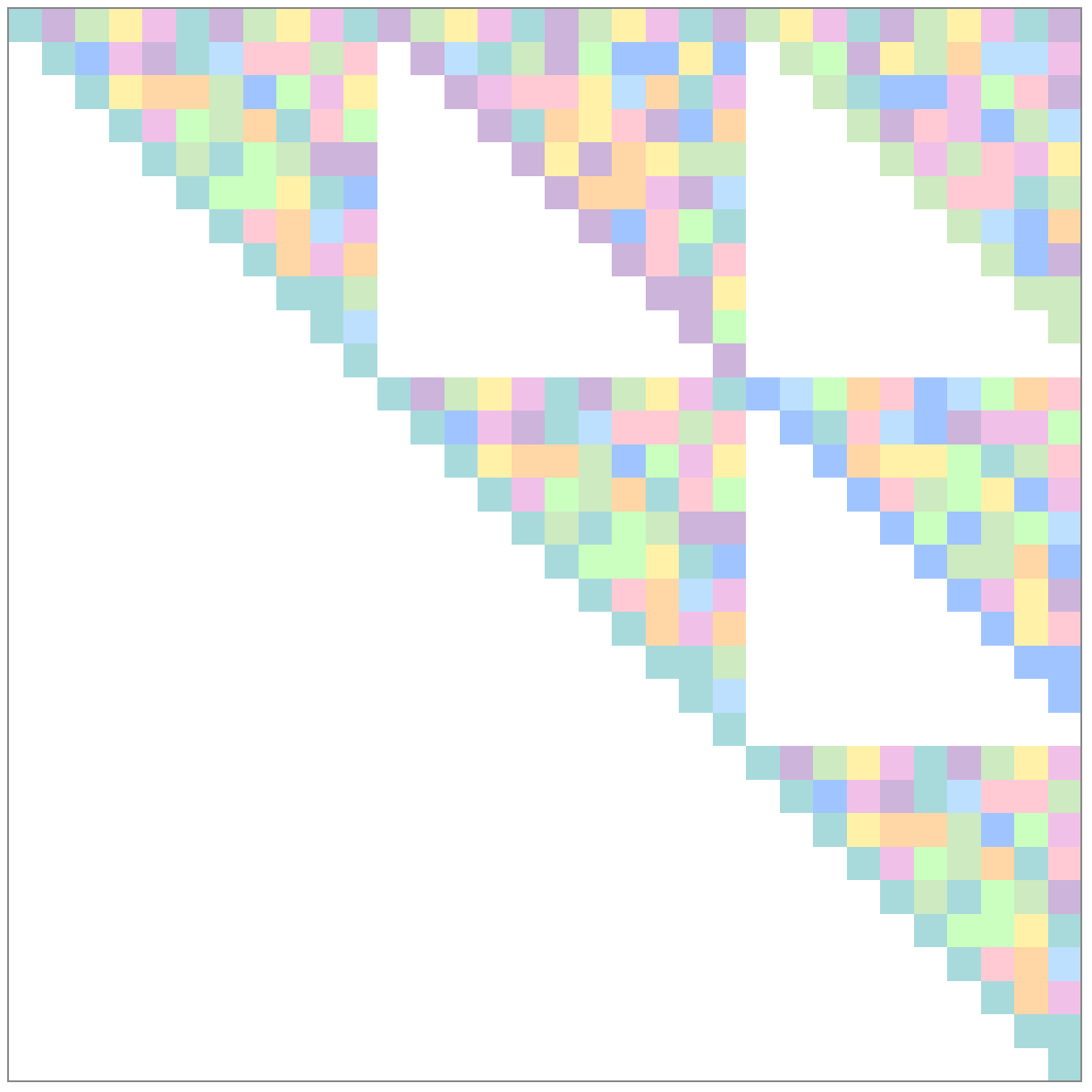}&
    \includegraphics[width=1.5cm]{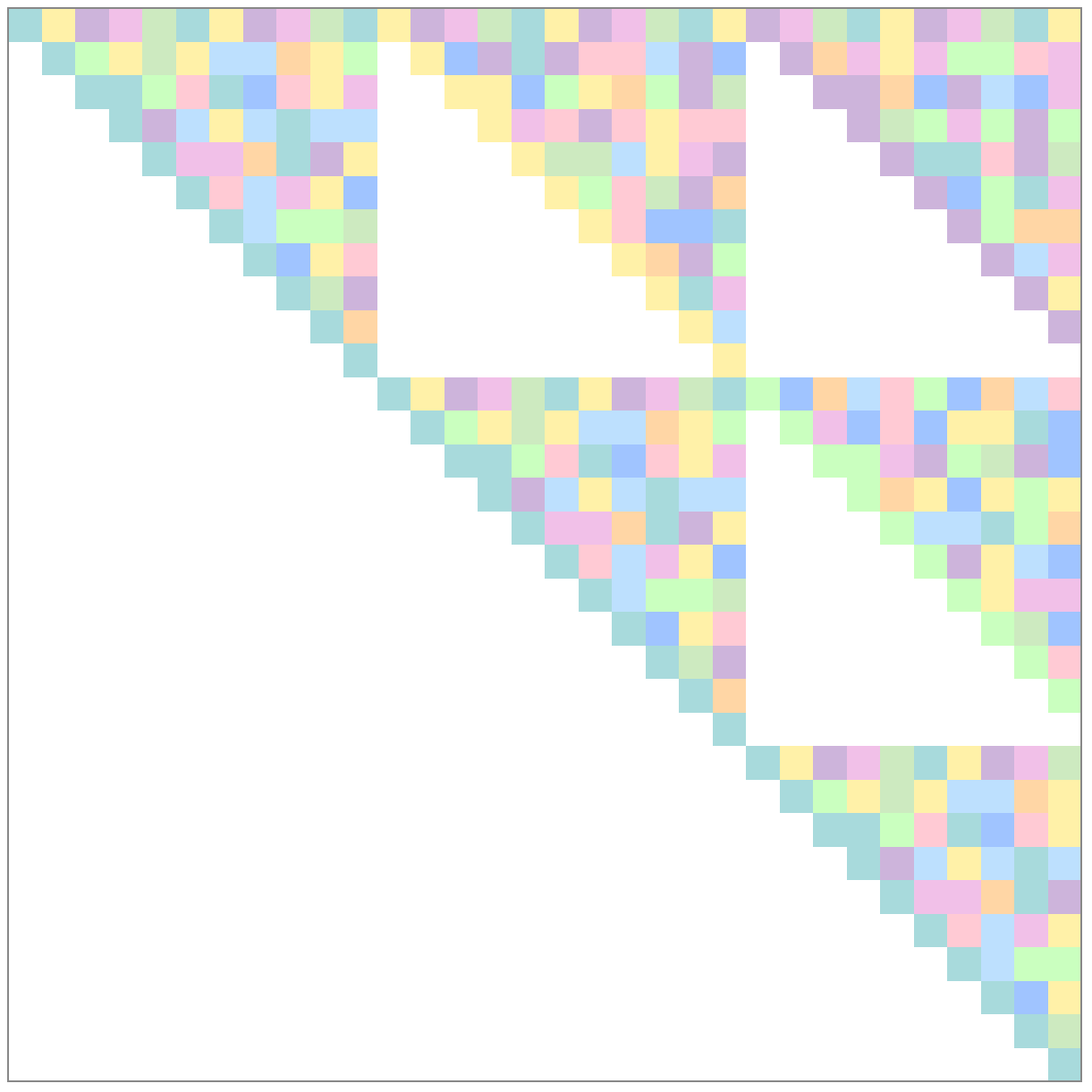}&
    \includegraphics[width=1.5cm]{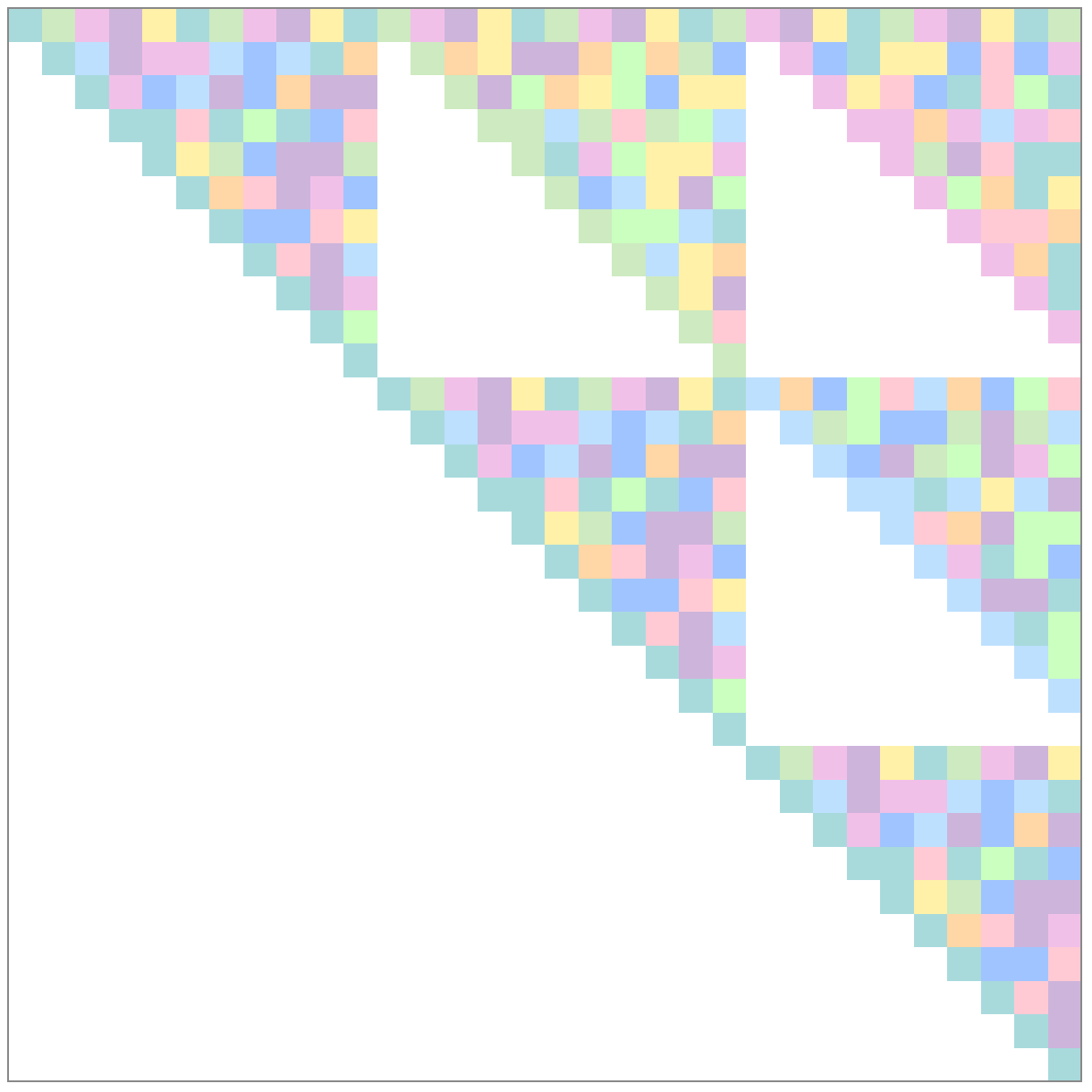}&
    \includegraphics[width=1.5cm]{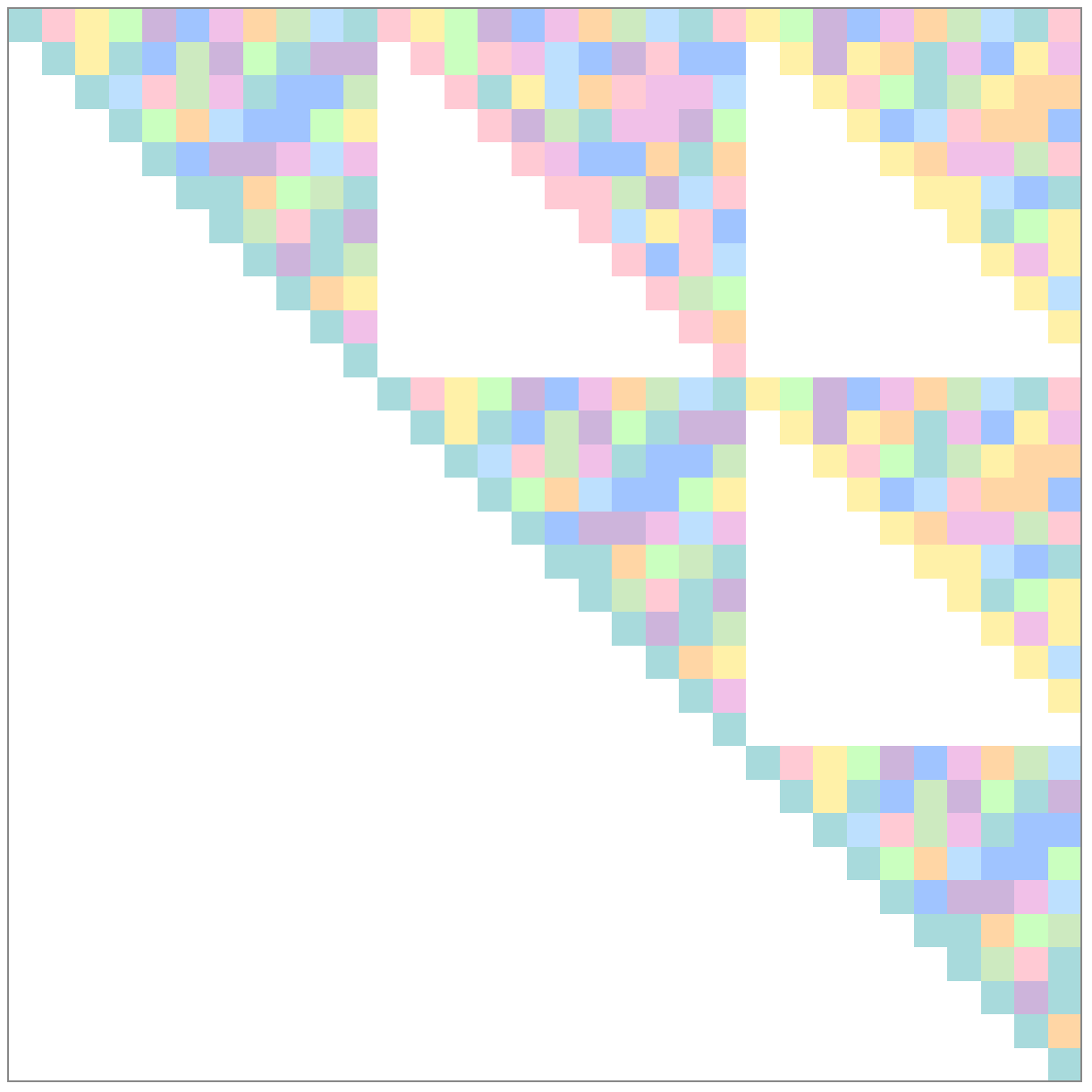}&
    \includegraphics[width=1.5cm]{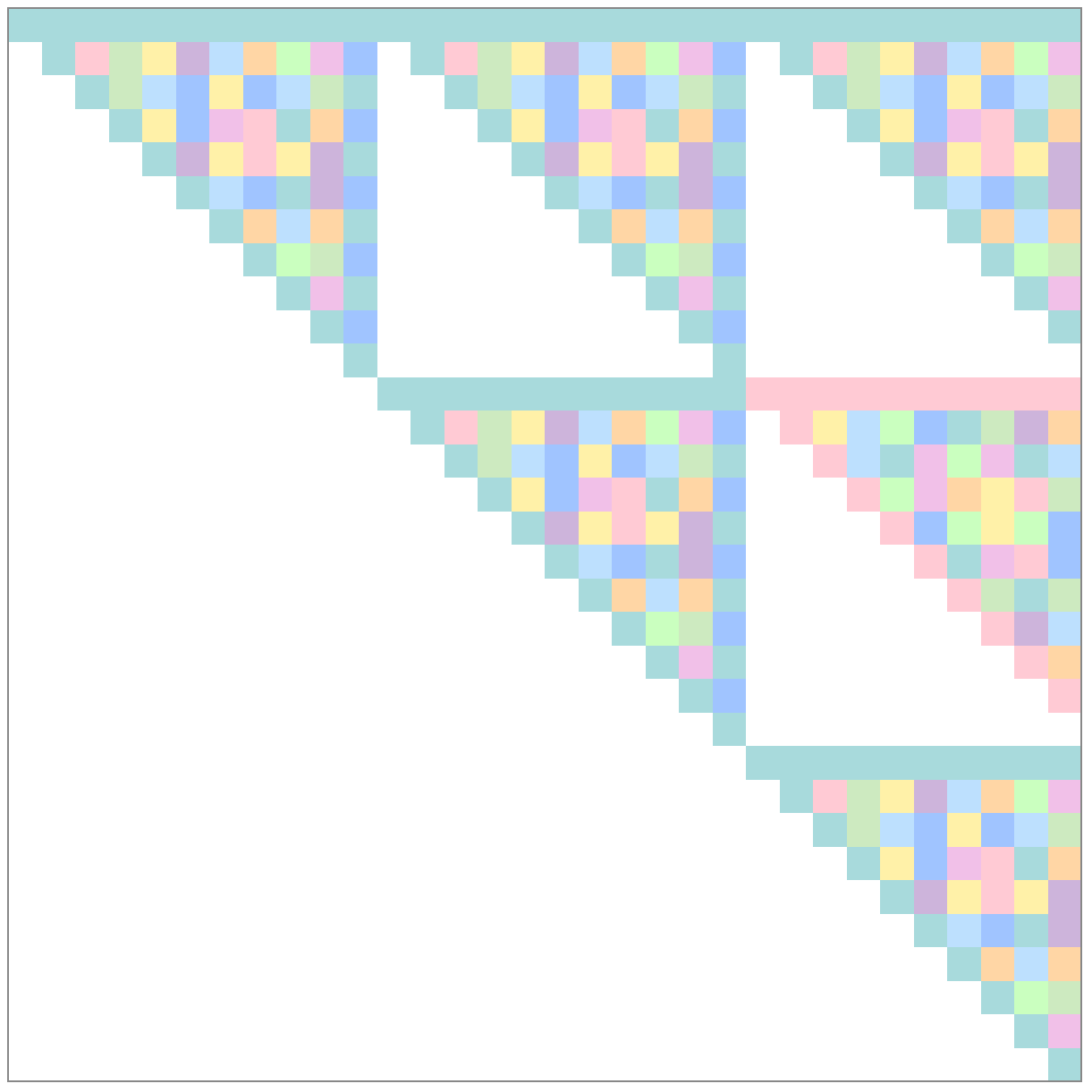}\\
  &   \rotatebox{90}{$\quad\quad$AS$$}&&\includegraphics[width=1.5cm]{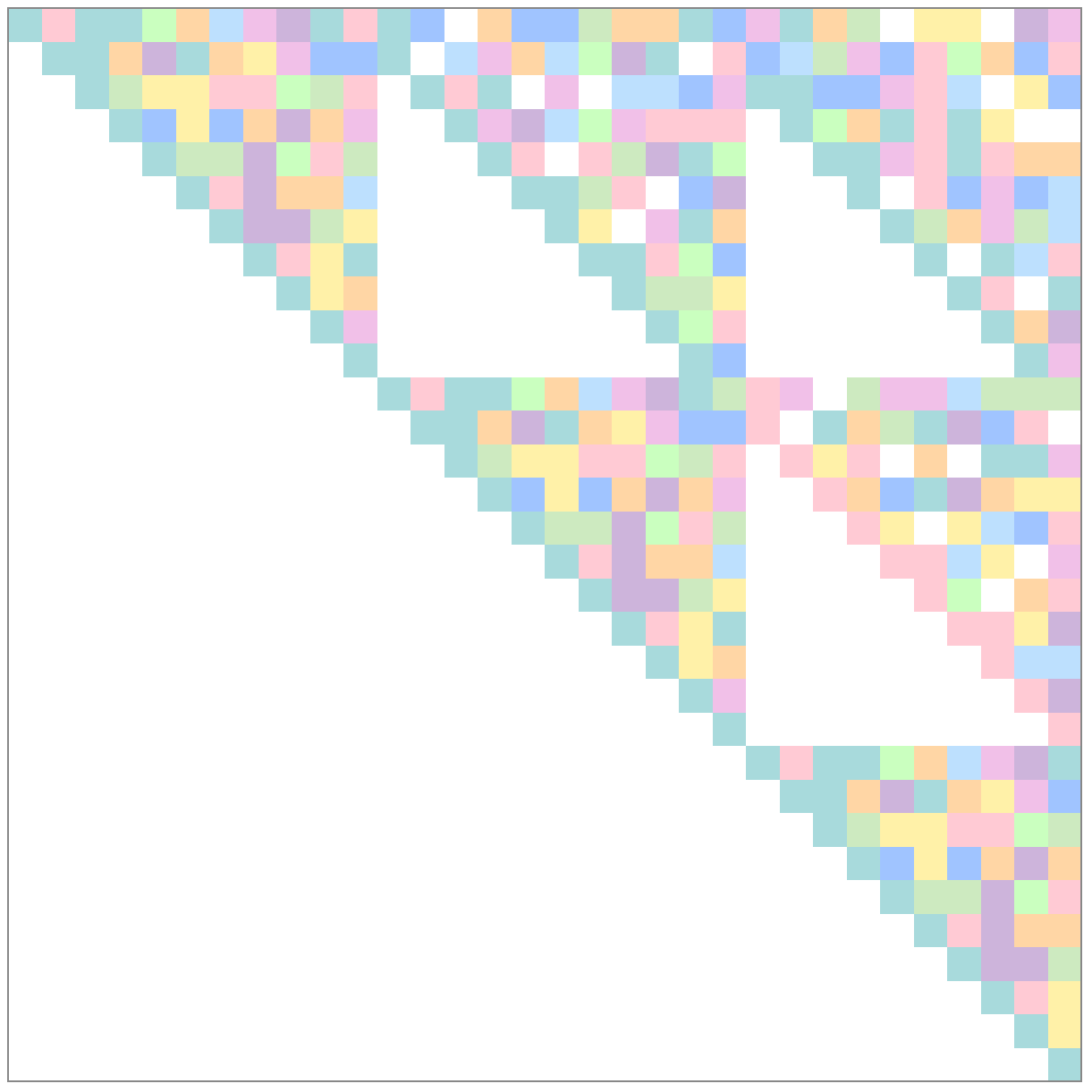}&
    \includegraphics[width=1.5cm]{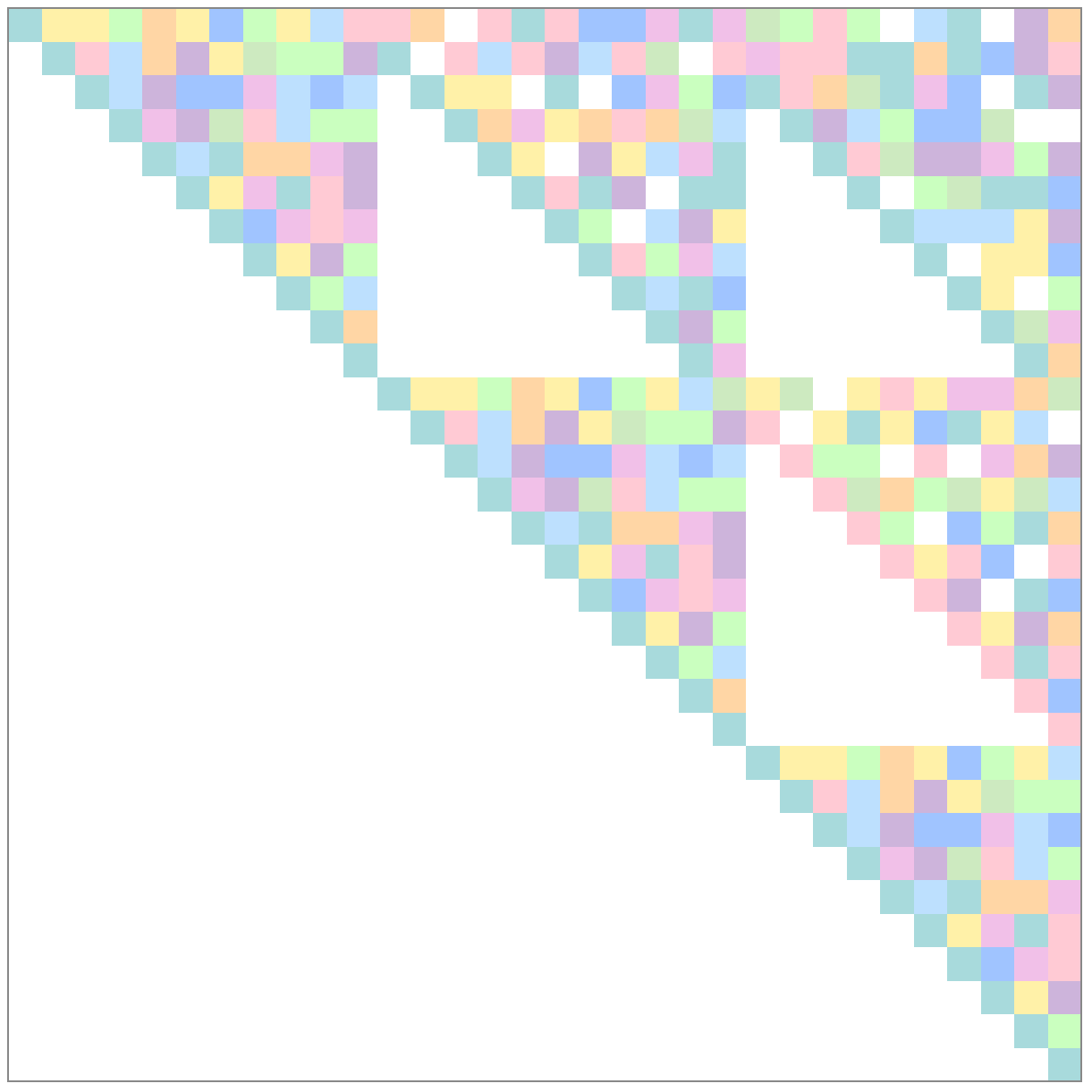}&
    \includegraphics[width=1.5cm]{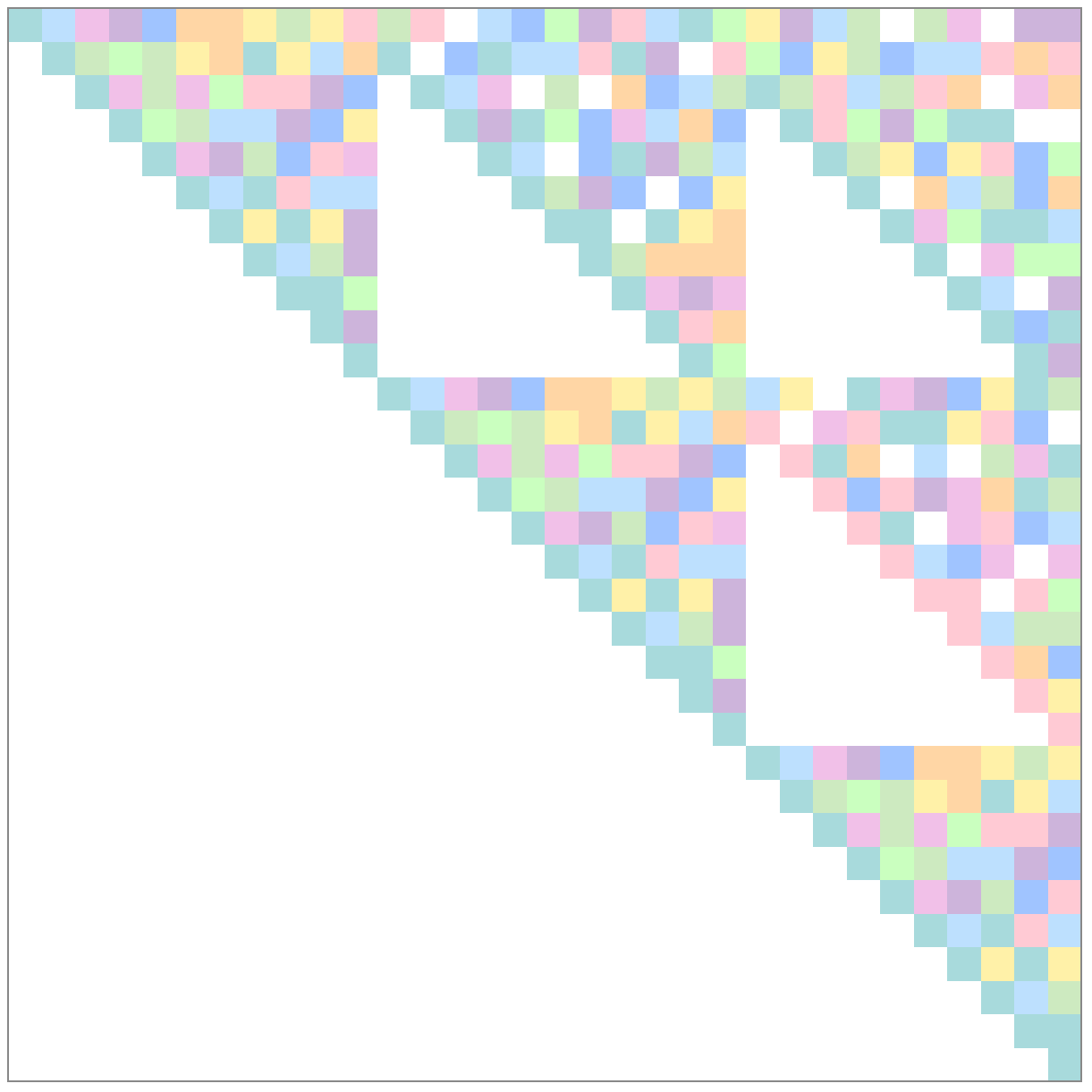}&
    \includegraphics[width=1.5cm]{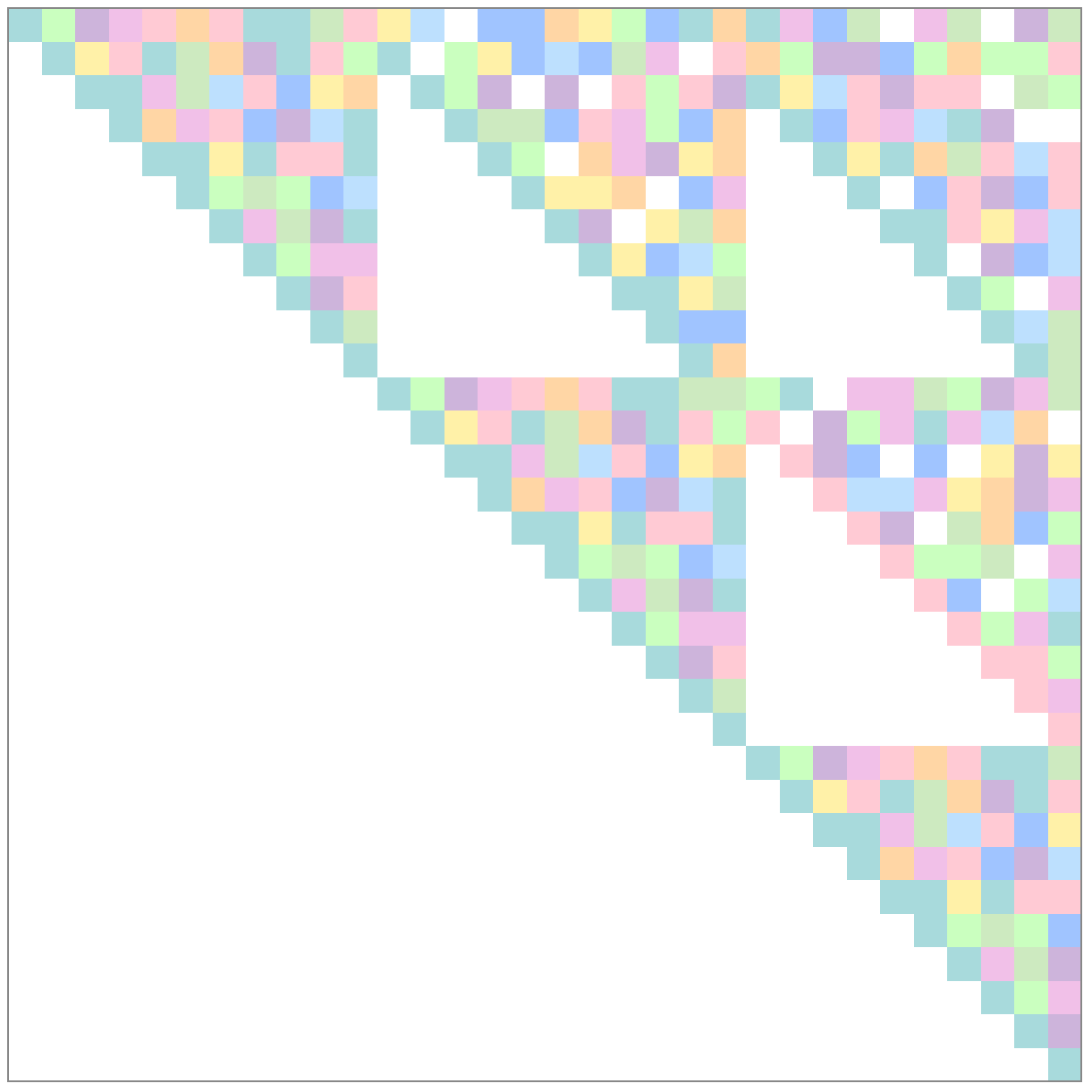}&
    \includegraphics[width=1.5cm]{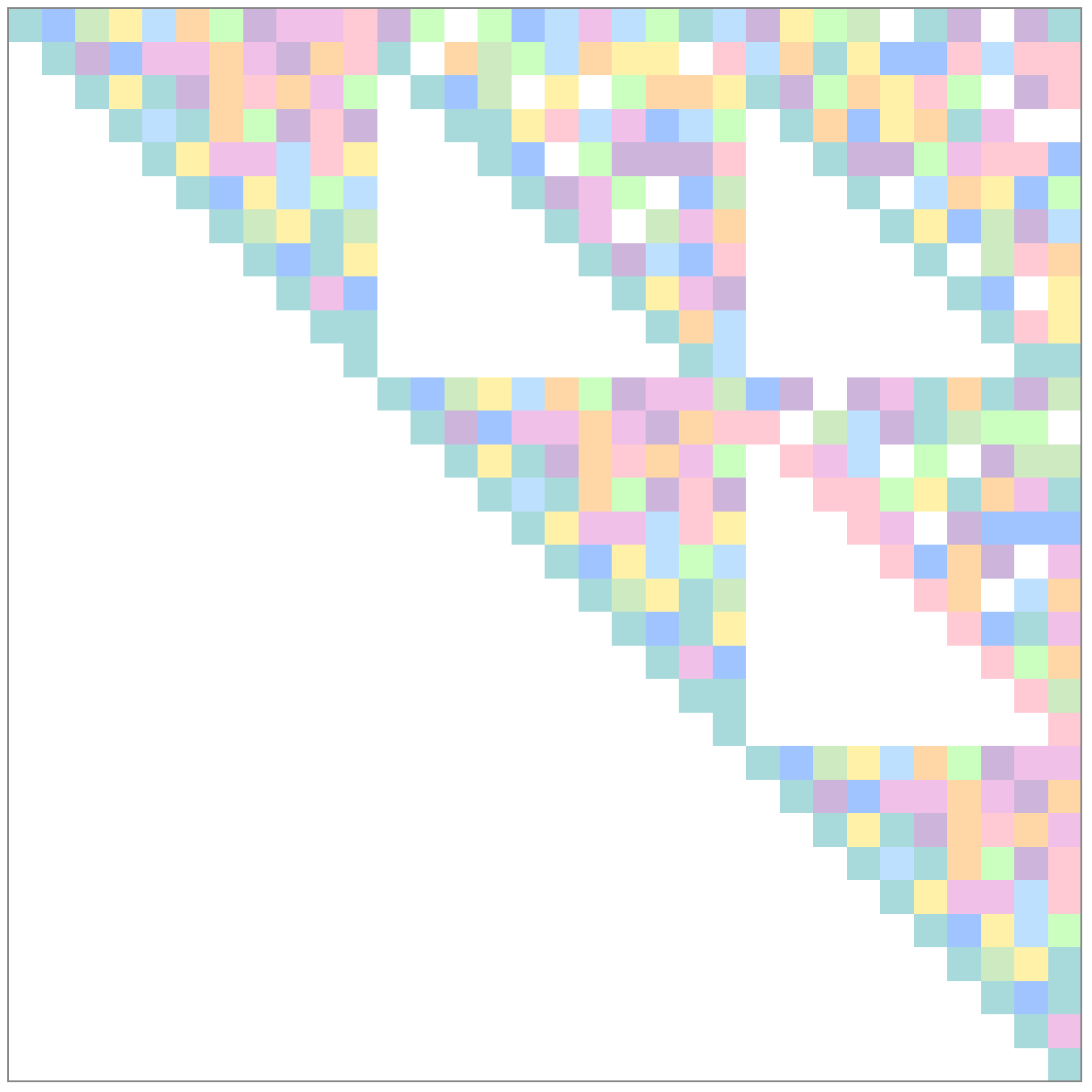}&
    \includegraphics[width=1.5cm]{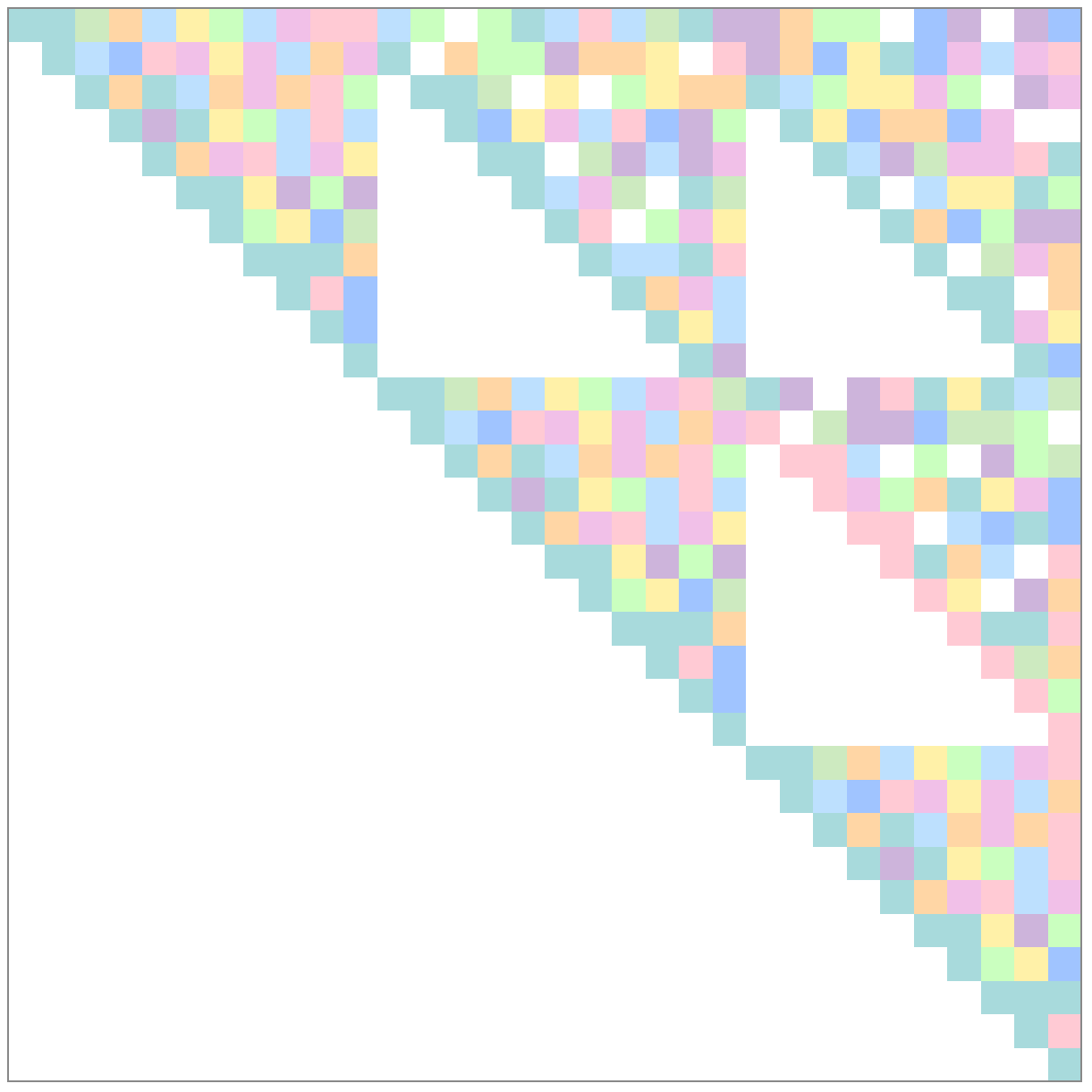}&
    \includegraphics[width=1.5cm]{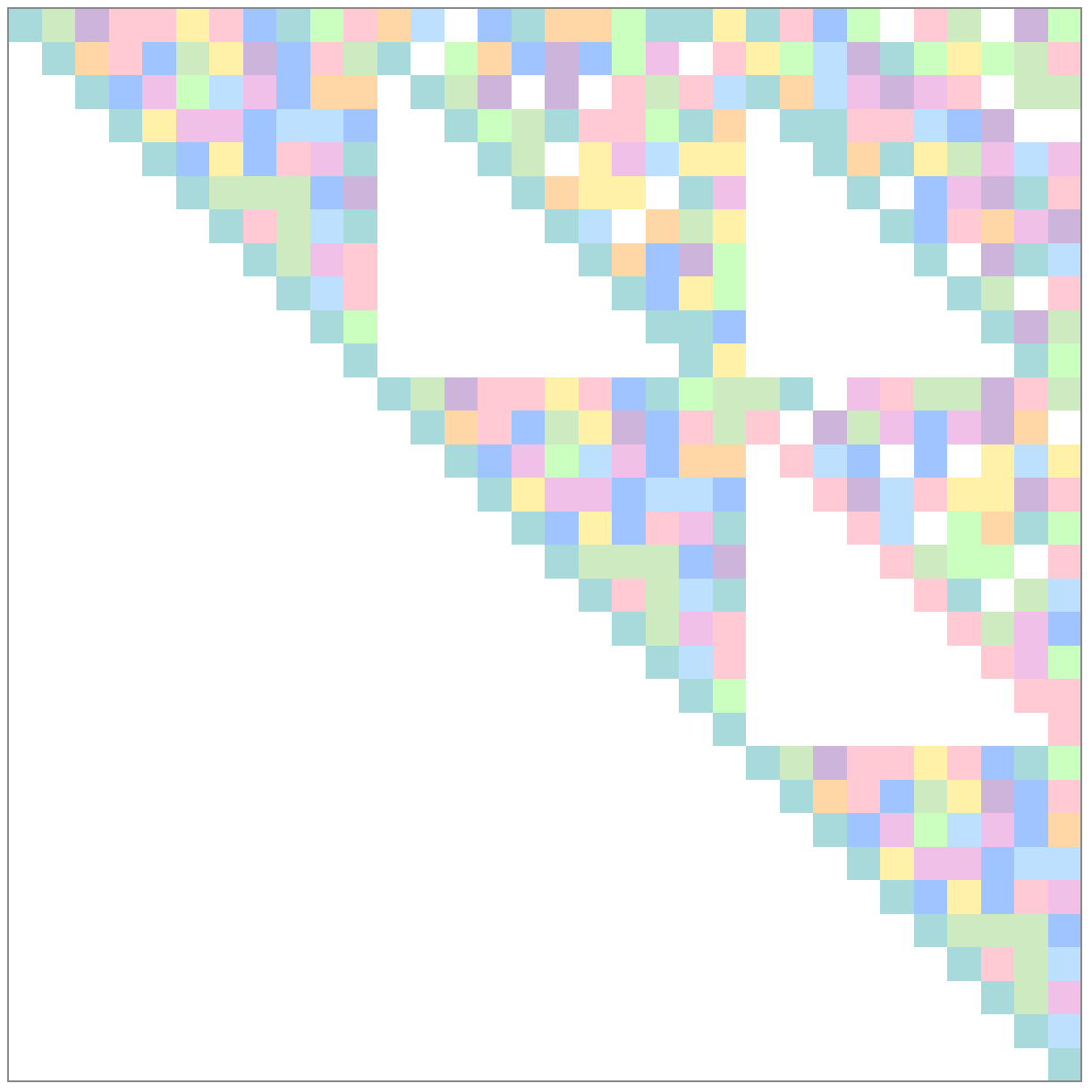}&
    \includegraphics[width=1.5cm]{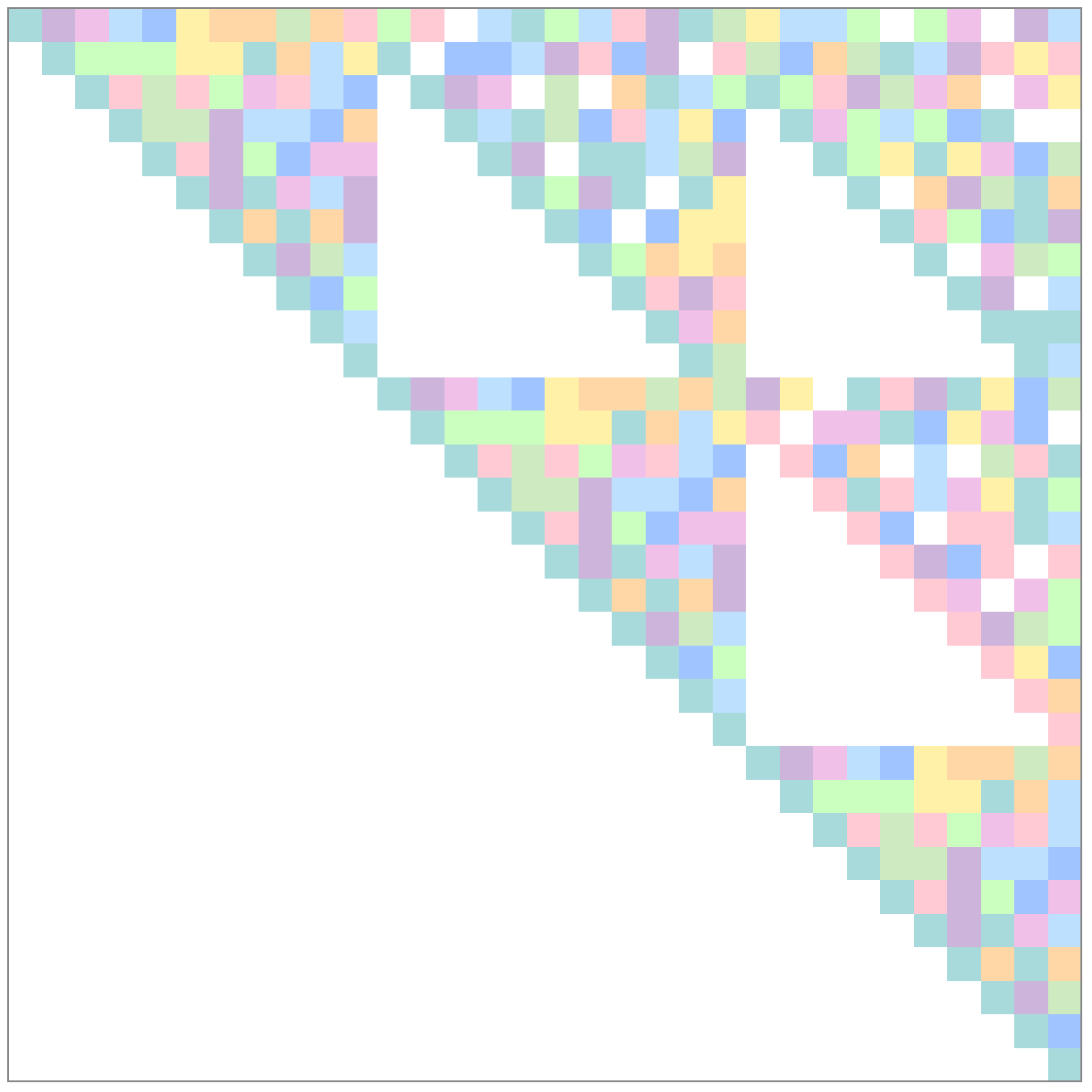}&
    \includegraphics[width=1.5cm]{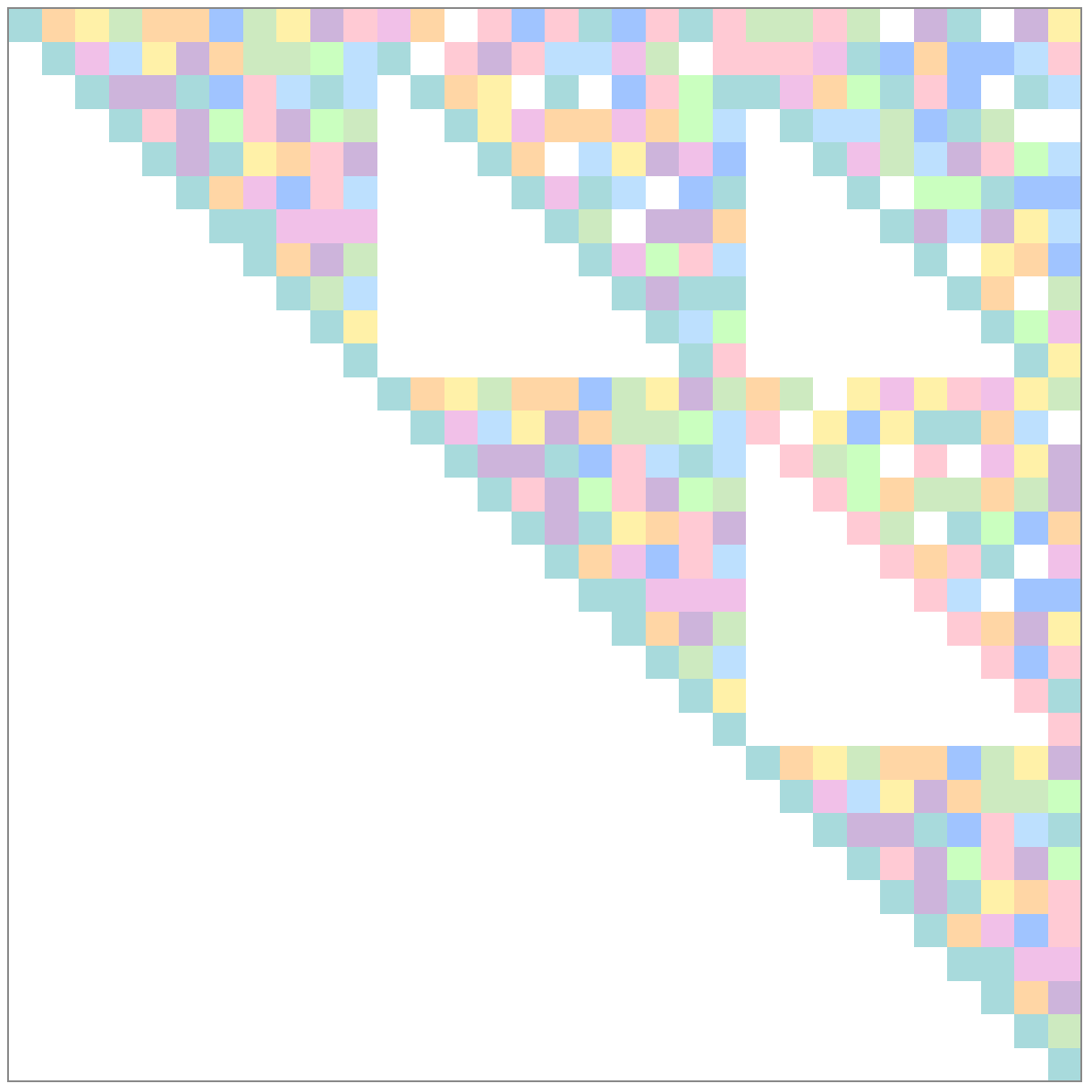}&
    \includegraphics[width=1.5cm]{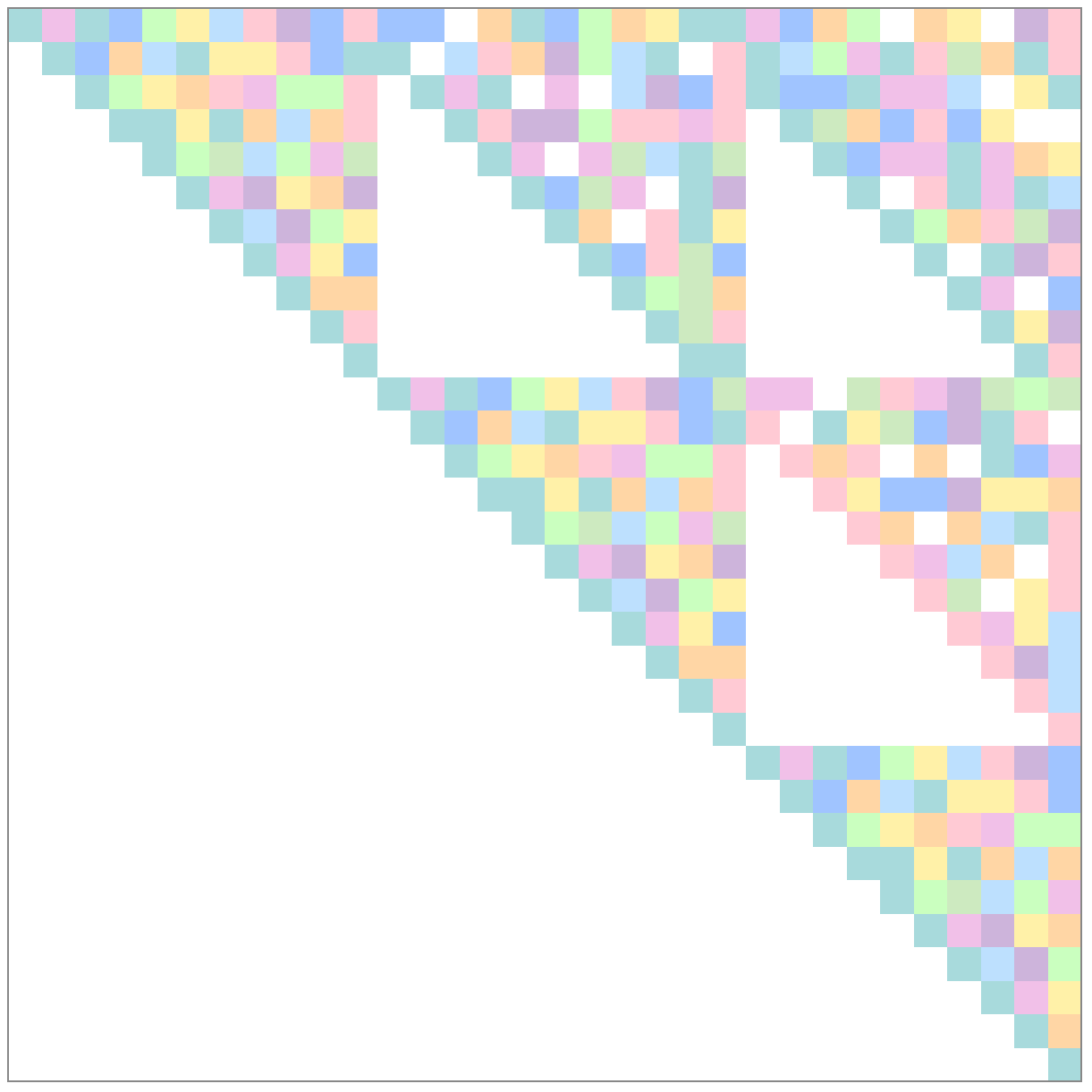}
    \end{tabular}
\caption{Color-coded representation of Sobol' matrices produced by consecutive polynomials $p_i(x) = x + c_i$ (top row) and Artin-Schreier polynomials $p_i(x) = x^b - x + c_i$ (bottom row) in $\mathrm{GF}(5)$, $\mathrm{GF}(7)$ and $\mathrm{GF}(11)$.\label{fig:matrices}}
\end{figure*}

\section{Consecutive polynomials produce Pascal powers}
\label{sec:pascal}

We first focus on the specific case of two-dimensional sequences, produced by Sobol' recurrence with consecutive polynomials, producing matrices $M^{(c_1)}$ and $M^{(c_2)}$. 

\paragraph*{Characteristic matrices}
For two-dimensional sequences, Ahmed et al.~\cite{ahmed2023analysis} define the \textit{characteristic matrix} as  $K = M^{(c_2)}\left(M^{(c_1)}\right)^{-1}$.
They show that this matrix uniquely characterizes a digital net: the 2D point set generated by $(M^{(c_1)},M^{(c_2)})$ is the same as the one given by $(\mathrm{Id},K)$, up to reordering.


We show that consecutive polynomials yield characteristic matrices in the form of tensorized powers of Pascal matrices:

\begin{lemma}\label{lemma:pascal}
Given two consecutive polynomials over $\mathrm{GF}(b)$ of degree $e$ and their corresponding Sobol' matrices $M^{(c_1)}$ and $M^{(c_2)}$, 
the characteristic matrix $K = M^{(c_2)} \left(M^{(c_1)}\right)^{-1}$ restricted to a size $em \times em$ (denoted $K_{[em \times em]}$) for  $m\geq 1$, is a Pascal matrix at the power $\alpha=c_1-c_2$ 
tensorized with an initialization matrix $K_{[e \times e]}$:
$$ K_{[em \times em]} = \cP^{\alpha} \otimes K_{[e \times e]}\,,$$
where $\cP^{\alpha}$ is an $m \times m$ matrix of coefficients in $\mathrm{GF}(b)$:
$$ \cP^{\alpha}_{i,j} = \binom{j-1}{i-1}\alpha^{j-i}\,. $$

Equivalently, the block coefficients of the (infinite) matrix $K$ are given by:
\begin{equation}\label{eq:rec}
  K_{i,j} = \alpha\,K_{i,j-1}+K_{i-1,j-1},
  \qquad K_{i,1}=\delta_{i,1}\,K_{[e \times e]}.
\end{equation}

\end{lemma}

\begin{proof}

We show that matrix $L := \cP^{\alpha} \otimes K_{[e \times e]}$ satisfies $L M^{(c_1)}_{[em \times em]} = M^{(c_2)}_{[em \times em]}$.
We exploit an identity relating powers of Pascal matrices~\cite{faure1982discrepance}: $\cP^a_{i,j} = a^{j-i} \binom{j-1}{i-1}$, and rewrite it 
for $e \times e$ block $(i,j)$ of matrix $L$: $L_{i,j+1} = \alpha L_{i,j}+L_{i-1,j}$. The matrices of equations~\ref{eq:q} and \ref{eq:f} are such that $F_{p+c_i} = F_p$ and $Q_{p+c_i}=Q_p-c_i \,\text{Id}_{e\times e}$. 
Using these identities, $\alpha = c_1-c_2$ and equation~\ref{eq:faurelemieux}, and denoting $N = L M^{(c_1)}_{[em \times em]}$, we develop the block matrix-vector multiplication as
\begin{align*}
N_{i,j} & =\sum_{k} L_{i,k}\,M^{(c_1)}_{k,j-1}\,(Q_p-c_1 \text{Id}_{e\times e})F_p + L_{i,k}\,M^{(c_1)}_{k-1,j-1}\,F_p\\
          & = N_{i,j-1}(Q_p-c_1 \text{Id}_{e\times e})F_p  + \sum_{k} L_{i,k+1} M^{(c_1)}_{k,j-1} F_p\\
&= N_{i,j-1}(Q_p-c_1 \text{Id}_{e\times e})F_p  +\bigl(\alpha\,N_{i,j-1}+N_{i-1,j-1}\bigr)F_p\\
&=\bigl(N_{i,j-1}(Q_p-c_2 \text{Id}_{e\times e})+N_{i-1,j-1}\bigr)F_p,
\end{align*}
Thus $N$ satisfies the Faure and Lemieux recurrence of equation~\ref{eq:faurelemieux} for polynomial $p+c_2$. The initialization is trivial, with $N_{i,1}=L_{i,1} M^{(c_1)}_{[e \times e]} = \delta_{i,1} K_{[e \times e]} M^{(c_1)}_{[e \times e]} = \delta_{i,1} M^{(c_2)}_{[e \times e]}$, and thus $N=M_{[em\times em]}^{c_2}$. $L$ is thus the characteristic matrix of $M^{(c_1)}$ and $M^{(c_2)}$; in other words $L = K_{[em \times em]}$ which concludes the proof.


\end{proof}

\section{Pascal powers may give $(0,s)$-sequences}
\label{sec:conditions}
We now move to the more general $s$-dimensional setting. 
The two-dimensional observation of Ahmed et al.~\cite{ahmed2023analysis} has been shown in the higher dimensional settings by Faure and Tezuka~\cite{faure2002another}. Given matrices $M^{(c_1)}, M^{(c_2)}, \dots,  M^{(c_{s})}$ producing a $(t,s)$-sequence, right-multiplying these matrices by the same invertible matrix results in the same set of points, up to reordering, and thus they also produce a $(t,s)$-sequence. One may choose right-multiplying by $\left(M^{(c_1)}\right)^{-1}$. This brings the study of Sobol' sequences produced by consecutive polynomials to the simpler study of tensorized powers of Pascal matrices.

We show that consecutive polynomials may lead to $(0,s)$-sequences.
\begin{theorem} 
\label{thm:main}
Given $s$ distinct consecutive polynomials of degree $e$ over $\mathrm{GF}(b)$, 
if the $s-$dimensional Sobol' point set of $b^m$ samples forms a 
$(0,m,s)$-net for all $m\le s\,(e-1)$ then it forms a $(0, s)-$sequence in base $b$.
\end{theorem}

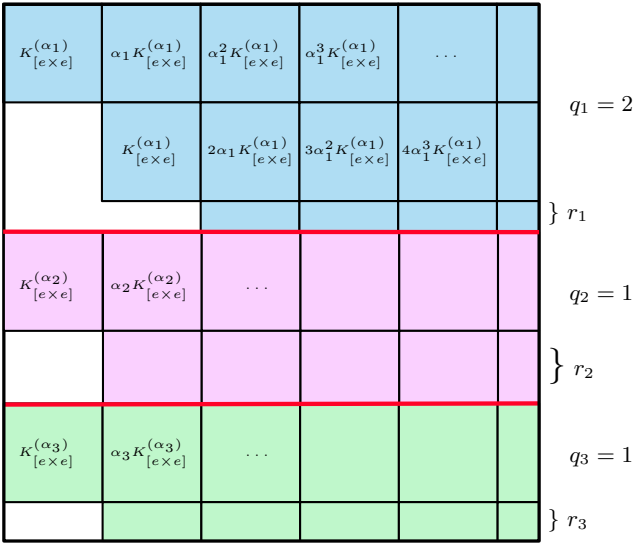
\begin{figure}
\scalebox{0.95} {
  \tiny
\begin{overpic}[width=7.5cm]{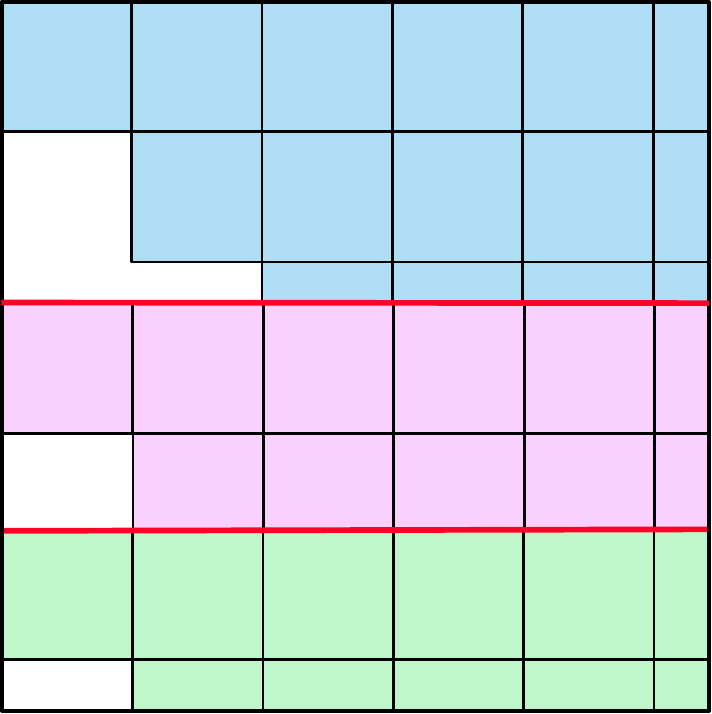}
\put(3,90){$K^{(\alpha_1)}_{[e\times e]}$}
\put(20,90){$\alpha_1K^{(\alpha_1)}_{[e\times e]}$}
\put(22,72){$K^{(\alpha_1)}_{[e\times e]}$}
\put(38,90){$\alpha_1^2K^{(\alpha_1)}_{[e\times e]}$}
\put(38,72){$2\alpha_1K^{(\alpha_1)}_{[e\times e]}$}
\put(56,90){$\alpha_1^3K^{(\alpha_1)}_{[e\times e]}$}
\put(56,72){$3\alpha_1^2K^{(\alpha_1)}_{[e\times e]}$}
\put(80,90){$\ldots$}
\put(74,72){$4\alpha_1^3K^{(\alpha_1)}_{[e\times e]}$}
\put(3,47){$K^{(\alpha_2)}_{[e\times e]}$}
\put(20,47){$\alpha_2K^{(\alpha_2)}_{[e\times e]}$}
\put(45,47){$\ldots$}
\put(3,16){$K^{(\alpha_3)}_{[e\times e]}$}
\put(20,16){$\alpha_3K^{(\alpha_3)}_{[e\times e]}$}
\put(45,16){$\ldots$}
\put(105,80){\small$q_1=2$}
\put(105,45){\small$q_2=1$}
\put(105,15){\small$q_3=1$}
\put(101,60){\small\}\,\,$r_1$}
\put(101,31){\Large\}\small\,\,$r_2$}
\put(101,3){\small\}\,\,$r_3$}
\end{overpic}
}
\caption{Compound matrix $\cK_{\bf d}$ formed by stacking the first $d_i$ rows of $K^{(\alpha_i)}$ for $i=1\dots s$.  $e \times e$ blocks follow the recursive construction of equation~\ref{eq:rec}.}
\label{fig:partition}
\end{figure}

\begin{proof} For $1\leq i\leq s$, with $\alpha_i = c_1-c_{i}$, and $K_{[m \times m]}^{(\alpha_i)}:=M_{[m \times m]}^{(c_i)}\left(M_{[m \times m]}^{(c_1)}\right)^{-1}$, the
  matrices $\left\{K^{(\alpha_i)}_{[m \times m]} \right\}_{i=1\dots s}$  produce the same sequence as matrices $\{M_{[m \times m]}^{(c_i)}\}_{i=1\dots s}$. In the following, we drop the $[m \times m]$ subscript for clarity.
As noted in Sec.~\ref{sec:background}, the first $b^m$ points form a $(0,m,s)$-net if and only if every possible $m \times m$ compound matrix $\cK_{\bf{d}}$, for all ${\bf d} = \{d_i\}_{i=1\dots s}$, formed by stacking the first $d_i$ rows of $K^{(\alpha_i)}$ has $\text{det}(\cK_{{\bf d}}) \neq 0$ (see figure~\ref{fig:partition}).
We assume that this is the case up to $m \leq s (e-1)$ and argue by contradiction. 
Suppose that some compound matrix of total size $m > s (e-1)$ is singular. Thus, for some row counts ${\bf d} = d_1,\dots,d_s$ with
$\sum_i d_i=m$, there exists a nonzero vector $\mathbf{v}\in\mathrm{GF}(b)^m$ annihilated by the selected rows.

We group the entries of $\mathbf{v}$ into blocks of length $e$,
$$ \mathbf{v}=(\mathbf{v}_0,\mathbf{v}_1,\ldots),\qquad \mathbf{v}_\ell\in\mathrm{GF}(b)^e, $$
padding the last block with zeros if necessary, and form the vector polynomial
$$F(x)=\sum_{\ell\ge 0} \mathbf{v}_\ell x^\ell.$$
Thus singularity of $\cK_{\bf d}$ is equivalent to the existence of a nonzero coefficient vector, or equivalently a nonzero vector polynomial $F$, such that $\cK_{\bf d}\mathbf{v}=0$.

We note that powers of Pascal matrices act as polynomial shifts, i.e., for a scalar-valued polynomial $p(x) = \sum_i a_i x^i$, with $\mathbf{a}=(a_i)$, denoting $\mathbf{a}' = \cP^y \mathbf{a}$, then $\sum_{i} a'_i x^i = p(x+y)$.

For our vector-valued polynomial $F$ and tensorized Pascal matrices, the $q$-th block row of one matrix $K^{(\alpha_i)} = \cP^{\alpha_i}\otimes K^{(\alpha_i)}_{[e \times e]}$  hence computes
$$  K^{(\alpha_{i})}_{[e \times e]} \, \text{coeff}_{x^q} F(x+\alpha_i), $$
where $\text{coeff}_{x^q} F$ denotes the coefficient of $x^q$ in $F$. 

Turning to the compound matrix $\cK_{\bf d}$, we write each row count as $d_i=e q_i+r_i$, for $0\le r_i<e$.
The first $e q_i$ rows of $K^{(\alpha_i)}$ consist of $q_i$ complete block rows. Since $K^{(\alpha_i)}_{e \times e}$ is invertible, these rows force
$$  \text{coeff}_{x^q} F(x+\alpha_i)=0,  \qquad q=0,\ldots,q_i-1.$$
Equivalently, $F(x)$ is divisible by $(x-\alpha_i)^{q_i}$.
Since the $\alpha_i$ are distinct, we can write for some vector polynomial $H(x)\neq 0$:
$$  F(x)=\Phi(x)H(x),  \qquad  \Phi(x)=\prod_{i=1}^s (x-\alpha_i)^{q_i}.$$

Let $ Q=\sum_i q_i = \deg \Phi$, and $R=\sum_i r_i$. 
Since $\Phi$ is monic of degree $Q$ and $F$ has at most
$m=eQ+R$ scalar coefficients, $H$ has at most $R$ scalar
coefficients. If $R=0$, this implies $H=0$, contradicting
$H\neq 0$. We may thus assume $R>0$ and handle the
remaining incomplete blocks of rows.

For a fixed $i$, the remaining $r_i$ selected rows of $K^{(\alpha_i)}$ are precisely the first
$r_i$ rows inside the next block row $q_i$. They therefore impose that the first
$r_i$ components of $ K^{(\alpha_i)}_{[e \times e]} \,\text{coeff}_{x^{q_i}}F(x+\alpha_i)$ vanish. 

We factor $\Phi$ by $(x-\alpha_i)^{q_i}$:
$$ \Phi(x)=(x-\alpha_i)^{q_i}\Phi_i(x),  \qquad  \Phi_i(x)=\prod_{j\ne i}(x-\alpha_j)^{q_j}\,,$$
then
$$  F(x+\alpha_i)  =  x^{q_i}\Phi_i(x+\alpha_i)H(x+\alpha_i)\,,$$
and therefore
$$  \text{coeff}_{x^{q_i}}F(x+\alpha_i)  =  \Phi_i(\alpha_i)H(\alpha_i)\,.$$

Since $\Phi_i(\alpha_i)=\prod_{j\ne i}(\alpha_i-\alpha_j)^{q_j}\ne 0$, this is equivalent to the vanishing of
the first $r_i$ components of $K^{(\alpha_i)}_{[e \times e]} H(\alpha_i)$ for every $i$.
In other words, the nonzero coefficient vector of $H$, of length $R=\sum_i r_i\le s(e-1)$, is annihilated by the first $r_i$ rows of each
$K^{(\alpha_i)}$, i.e., by the compound matrix $\cK_{(r_1,\dots,r_s)}$, 
which the hypothesis forbids. Hence every compound matrix is non-singular
and the sequence is a $(0,s)$-sequence.




\end{proof}

\section{Guaranteeing $t=0$}
\label{sec:exist}
The conditions of Sec.~\ref{sec:conditions} would suggest searching for $s$ initialization matrices, each of size $e \times e$, and identifying cases that satisfy conditions of theorem~\ref{thm:main}. While this works in practice for small values of bases $b$,  $e$ and $s$, it becomes intractable even for moderate values. Even focusing on invertible upper triangular initialization matrices, this would lead to $b^{e(e-1) s/2} (b-1)^{e\, s}$ degrees of freedom, such that even 3-dimensional sequences in base $b=5$ using degree $e=5$ polynomials would require a search over $10^{30}$ candidates.

When $e=b$, we formulate the following condition to guarantee that sequences have $t=0$.

\begin{proposition}\label{prop:init}
Considering $s$ distinct consecutive polynomials of degree $e=b$ over  $\mathrm{GF}(b)$, initializing their $s$ Sobol' matrices with 
$$M^{(c_i)}_{[e \times e]} = D \cP^{c_i}_{[e\times e]}  D^{-1}$$
with $D$ a diagonal invertible matrix, and $\cP_{[e\times e]}$ the $e\times e$ Pascal matrix always produces $(0,s)$-sequences.
\end{proposition}

To show proposition~\ref{prop:init},  we first exhibit the self-similar structure of powers of Pascal matrices.

\begin{lemma}[Self-similarity of Pascal powers]\label{lemma:tensor}
Let $b$ be prime and $a\in\mathrm{GF}(b)$. Over $\mathrm{GF}(b)$,
\begin{equation}\label{eq:tensor-fine}
\cP^{a}_{[bm\times bm]}
=\cP^{a}_{[m\times m]}\otimes\cP^{a}_{[b\times b]}
\qquad\text{for every }m\ge1,
\end{equation}
and, when $m$ is a power of $b$,
\begin{equation}\label{eq:tensor-coarse}
\cP^{a}_{[bm\times bm]}
=\cP^{a}_{[b\times b]}\otimes\cP^{a}_{[m\times m]}.
\end{equation}
\end{lemma}

\begin{proof}
We exploit the identity $\cP^a_{i,j} = a^{j-i} \binom{j-1}{i-1}$. 
For equation~\ref{eq:tensor-fine}, split off the last base-$b$ digit of the
(0-based) indices: $i=\hat\imath\,b+i_0$, $j=\hat\jmath\,b+j_0$ with
$0\le i_0,j_0<b$. Lucas' congruence \cite{lucas1891theorie,fine1947binomial} gives
$\binom{j}{i}\equiv\binom{\hat\jmath}{\hat\imath}\binom{j_0}{i_0}\pmod b$, and
Fermat's little theorem gives
$a^{\,j-i}=\bigl(a^{\,b}\bigr)^{\hat\jmath-\hat\imath}a^{\,j_0-i_0}
=a^{\,\hat\jmath-\hat\imath}\,a^{\,j_0-i_0}$; the product of the two
congruences is the $(i,j)$ entry of
$\cP^{a}_{[m\times m]}\otimes\cP^{a}_{[b\times b]}$. For
\eqref{eq:tensor-coarse}, split off the \emph{leading} digit instead,
$i=I\,m+i'$, $j=J\,m+j'$ with $0\le I,J<b$: this is a single base-$b$ digit exactly when
$m$ is a power of $b$, and Fermat iterates to $a^{\,m}=a$; the same two
congruences conclude.
\end{proof}

We now turn to the proof of proposition~\ref{prop:init}.

\begin{proof}
As noted by Faure and Lemieux~\cite{faure2016irreducible}, a Sobol' matrix obtained using an initialization $M_{[e\times e]}$ can also be obtained by initializing the recurrence with an identity matrix and left-multiplying the result by a block-diagonal matrix with $M_{[e\times e]}$ as the diagonal blocks. Denoting $W_i = M^{(c_i)}_{[e \times e]} = D\,\cP^{c_i}_{[e\times e]}D^{-1}$ the initialization matrix: $M^{(c_i)}=\mathrm{diag}(W_i,W_i,\dots)\,M^{(c_i)}_{\mathrm{Id}}$. Using lemma~\ref{lemma:pascal} for the identity-initialized matrices,
\begin{align}
K^{(\alpha_i)} &=\mathrm{diag}(W_i,\dots)\,\bigl(\cP^{\alpha_i}\otimes \mathrm{Id}_e\bigr)\,  \mathrm{diag}(W_1,\dots)^{-1} \\
                          &=\cP^{\alpha_i}\otimes W_iW_1^{-1}.
\end{align}
Pascal powers compose as the coordinate shifts they encode,
$\cP^{a}\cP^{a'}=\cP^{a+a'}$, and conjugating by a geometric diagonal rescales
the power, $\mathrm{diag}(\lambda^u)_u\,\cP^{a}\,\mathrm{diag}(\lambda^{-u})_u
=\cP^{a/\lambda}$ (since $\lambda^{u}\binom{v}{u}a^{v-u}\lambda^{-v}
=\binom{v}{u}(a/\lambda)^{v-u}$); hence, with $S=\mathrm{diag}\bigl((-1)^u\bigr)_u$ and
$D'=DS$,
$$
W_iW_1^{-1}=D\,\cP^{\,c_i-c_1}_{[e\times e]}\,D^{-1}
=D'\,\cP^{\,\alpha_i}_{[e\times e]}\,D'^{-1},
~~~ \alpha_i=c_1-c_i .
$$
Because $e=b$, lemma~\ref{lemma:tensor} gives 
$\cP^{\alpha}\otimes\cP^{\alpha}_{[b\times b]}=\cP^{\alpha}$, and therefore
$$
K^{(\alpha_i)}=\Delta\;\cP^{\alpha_i}\;\Delta^{-1},
\qquad \Delta=\mathrm{diag}(D',D',\dots),
$$
with the \emph{same} diagonal $\Delta$ for every $i$. Diagonal conjugation
preserves the nonsingularity of compound matrices: the left factor rescales each selected row by
a nonzero scalar, and the right factor is a common non-singular triangular
factor. The compound matrices of the family $\{\cP^{\alpha_i}\}$, with
pairwise-distinct powers $\alpha_i$  are all
non-singular by Faure's theorem \cite{faure1982discrepance}. All
compound determinants are therefore nonzero, and the sequence is a
$(0,s)$-sequence.
\end{proof}

\paragraph*{Artin-Schreier polynomials}
While our proofs have not relied on polynomial irreducibility, enforcing the use of irreducible polynomials allows to maintain compatibility with Sobol' sequences. By using $s_1$ irreducible consecutive polynomials and $s_2$ other polynomials commonly used for Sobol' sequences \cite{joe2008constructing}, one may obtain samples of $s_1+s_2$ dimensions with $t=0$ guarantees in the $s_1$ subspace, and preserving a bounded $t$ value for the full $(s_1+s_2)$-dimensional sequence. The question arises of whether $s$ consecutive irreducible Sobol' sequences always exist. The answer lies within Artin-Schreier theory~\cite{artin1927kennzeichnung}. This theory states that for prime $b$, the $b-1$ polynomials in  $\mathrm{GF}(b)$ of the form $p_i(x) = x^b - x +c_i$ are always irreducible. Such polynomials perfectly fit our framework since proposition~\ref{prop:init} requires $e=b$, which also results from Artin-Schreier theory, and since $s=b-1$ is the maximum number of possible different consecutive irreducible polynomials in $\mathrm{GF}(b)$, apart from the $b$ linear polynomials $p_i(x) = x +c_i$.
In practice, we use Artin-Schreier polynomials to define our sequences and for our numerical experiments. Illustrations of Sobol' matrices produced by $\{x+c_i\}$ and Artin-Schreier polynomials are shown in figure~\ref{fig:matrices} for various bases.

\section{Numerical experiments}
\label{sec:numerics}

\subsection{Artin-Schreier sequences}

\begin{figure*}[h]
  \setlength{\tabcolsep}{0pt}
  \begin{tabular}{cc}
  \begin{overpic}[width=9cm]{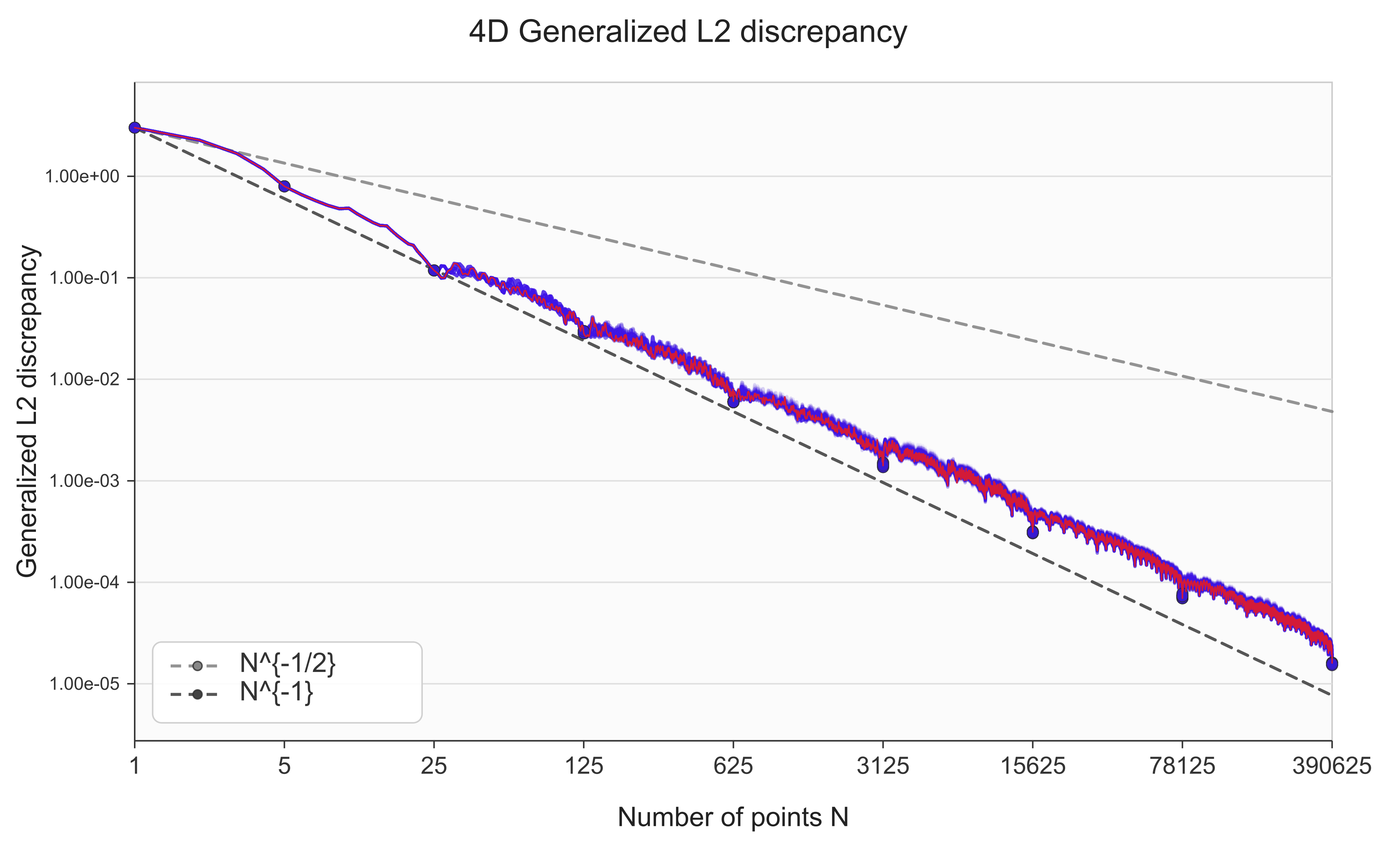}
    \put(85,53){\fcolorbox{black}{white}{$\mathrm{GF}(5)$}}
  \end{overpic}
  & \begin{overpic}[width=9cm]{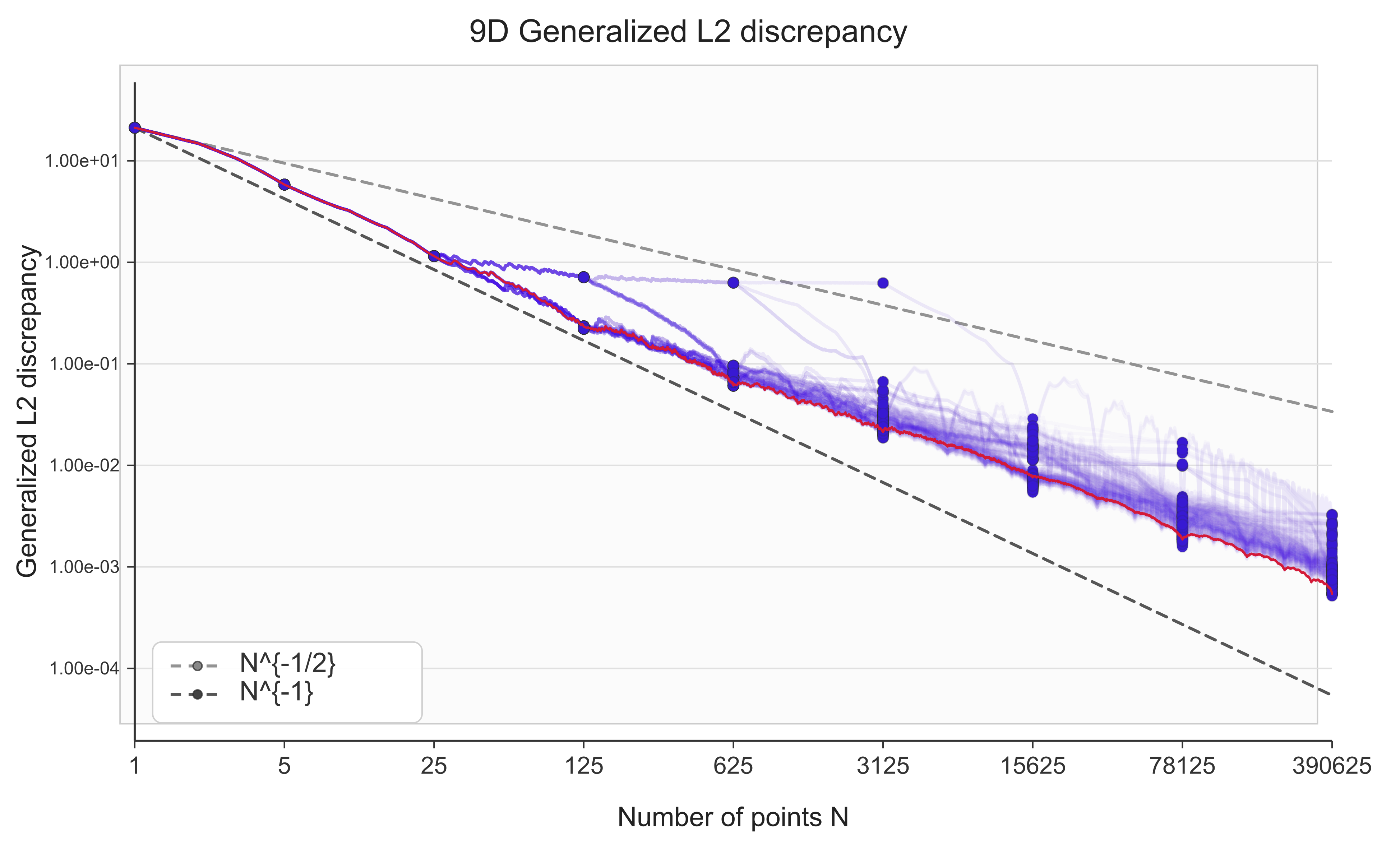}
    \put(85,53){\fcolorbox{black}{white}{$\mathrm{GF}(5)$}}
  \end{overpic}
  \end{tabular}
\caption{In $\mathrm{GF}(5)$, we illustrate the behavior of the discrepancies of all $256$ initializations of Artin-Schreier polynomials that produce $(0,4)$-sequences in 4D (left), and their concatenations with affine polynomials that produce $(0,5)$-sequences to produce 9-dimensional sequences (right). While initialization does not affect the 4D behavior much, it significantly matters when considering the combined 9D sequence.}
\label{fig:allinits}
\end{figure*}

In figure~\ref{fig:allinits} (left), we evaluate a uniformity measure called generalized $\ell^2$ discrepancy, which relates Monte Carlo integration error to point set uniformity~\cite{Lemieux2009} (lower is better) for our sequences with all $4^4 = 256$ possible initialization $D$ matrices in $\mathrm{GF}(5)$,  using the \texttt{tms\_lib}   library~\cite{tmslib}. 
For a chosen $D$, we show in figure~\ref{fig:as} the pairwise 2D projections of point sets with $b^m \approx 1200$ samples (for the nearest perfect power of $b$) corresponding to Sobol' sequences with Artin-Schreier polynomials in bases $b=5, 7$ and $11$ for up to 5 dimensions (the case $b=3$ corresponds dimensions 13--14 of the $\mathrm{GF}(3)$ sequence of Ostromoukhov et al.~\cite{OBCI24}). We also show other Sobol'-based sequences, such as the improvement of Faure and Lemieux~\cite{faure2016irreducible} using irreducible polynomials in $\mathrm{GF}(2)$, the sequence of Bonneel et al.~\cite{BCIO25} that uses polynomials of the form $p$ and $p^2+p+1$ to guarantee $t=1$ in consecutive pairs of dimensions in $\mathrm{GF}(2)$, and the sequence of Ostromoukhov et al.~\cite{OBCI24} that numerically optimizes quadruples of Sobol' dimensions in $\mathrm{GF}(3)$.
In figure~\ref{fig:discrepancies} (top), we compare the generalized $\ell^2$ discrepancy of our sequences to these alternatives, as well as Joe and Kuo's Sobol' sequence~\cite{joe2008constructing}, where discrepancy is computed after randomization by a base-$b$ Owen scrambling~\cite{owen1997monte} and averaged over 32 realizations.

\begin{figure*}[h]
\begin{tabular}{c@{}c@{}c}
\hspace*{-2cm}\raisebox{2.1cm}{\includegraphics[width=3.3cm]{"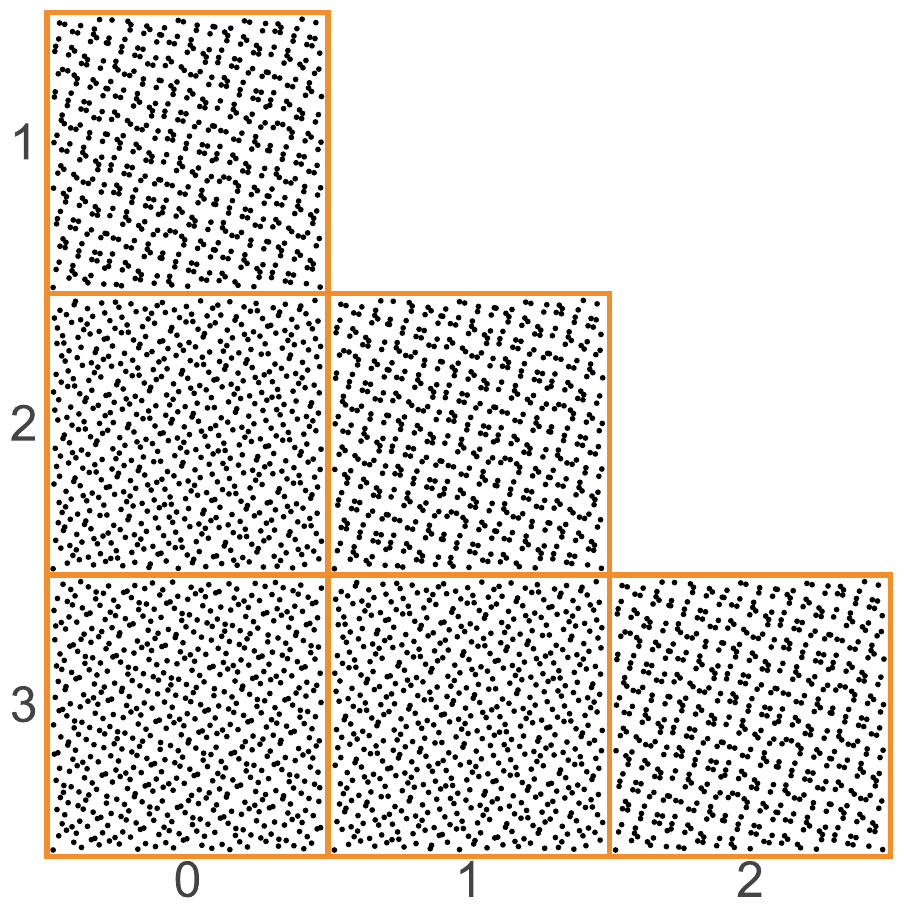"}} &
\includegraphics[width=5.5cm]{"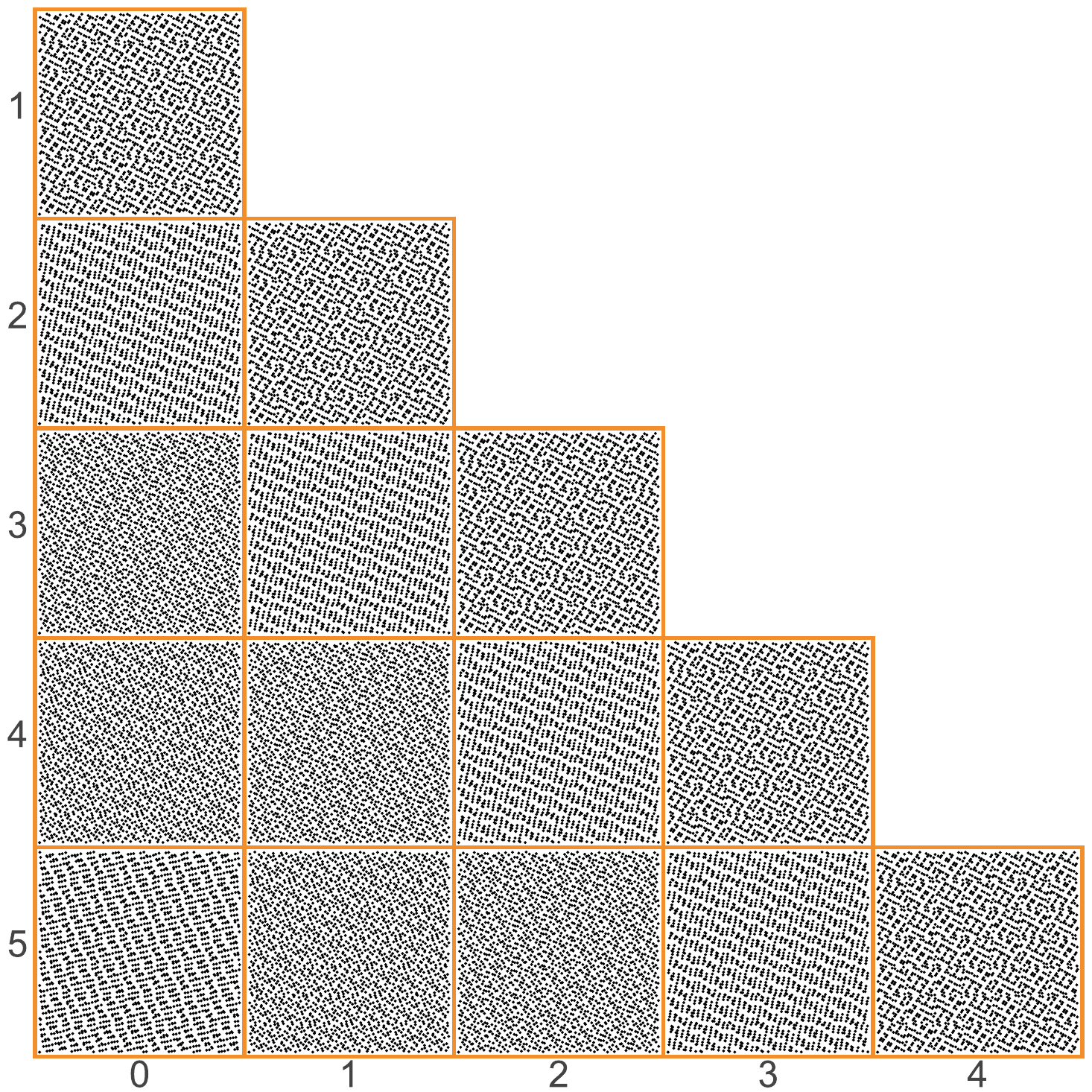"} &
\includegraphics[width=5.5cm]{"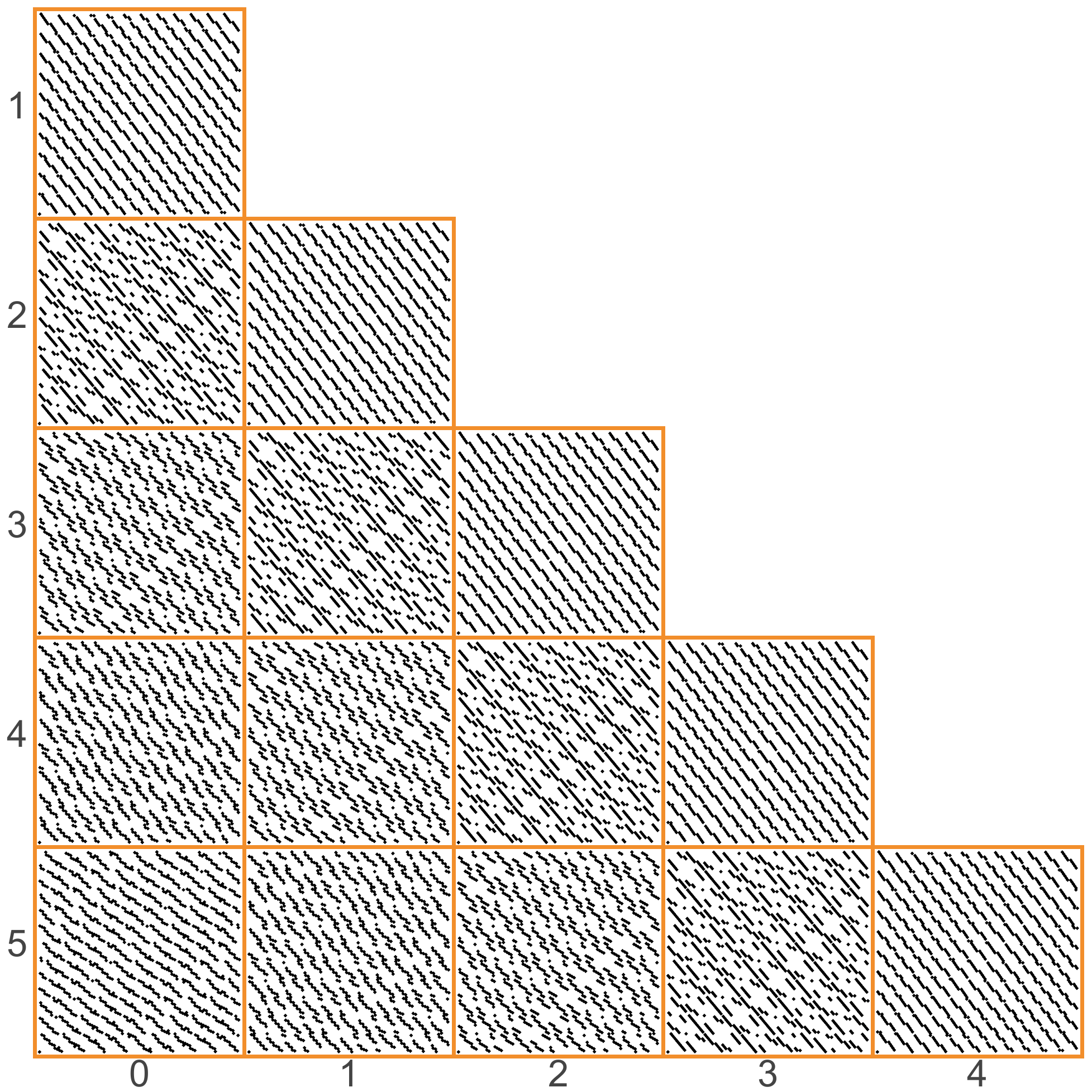"} \\
AS $\mathrm{GF}(5)$ & AS $\mathrm{GF}(7)$ & AS $\mathrm{GF}(11)$\\
\includegraphics[width=5.5cm]{"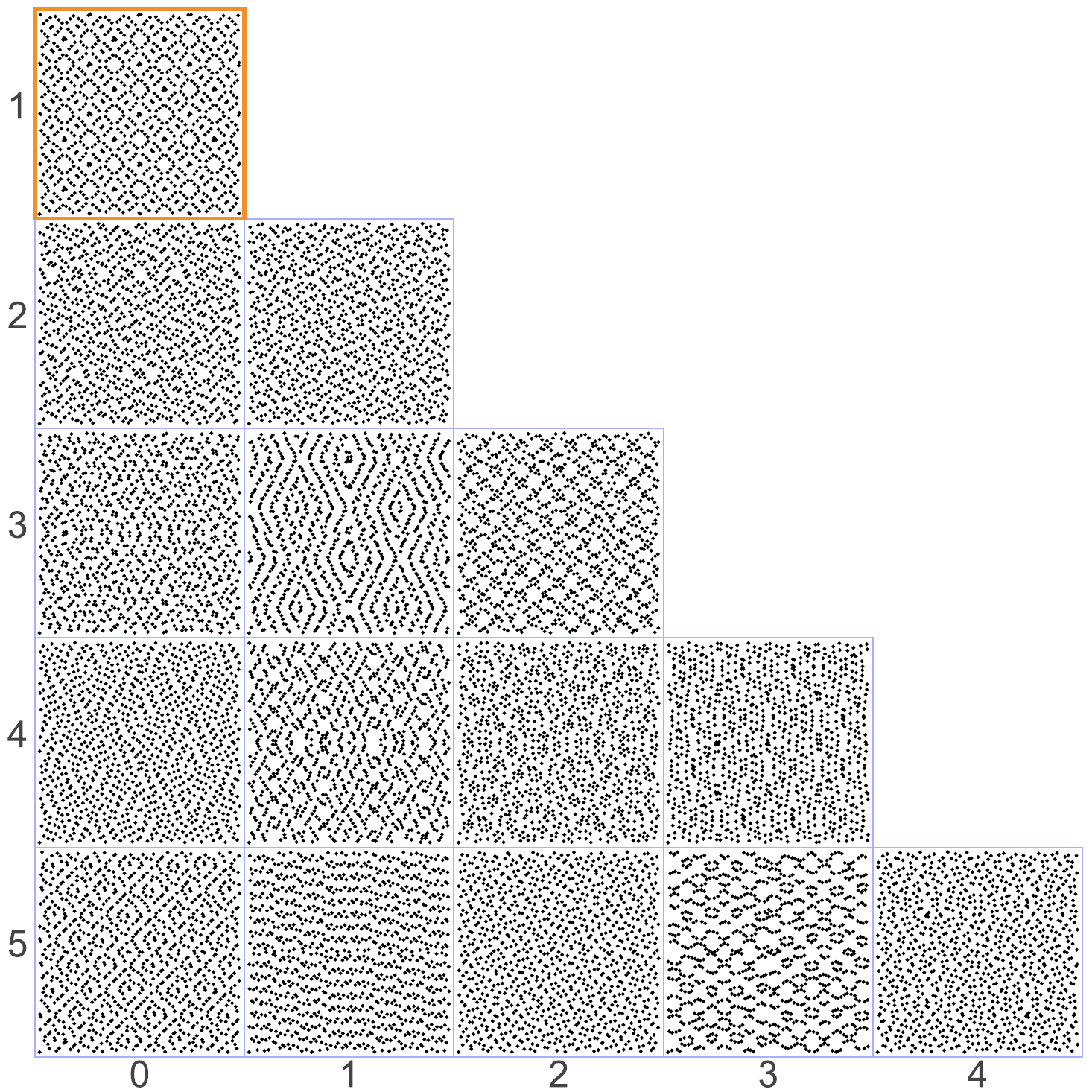"} &
\includegraphics[width=5.5cm]{"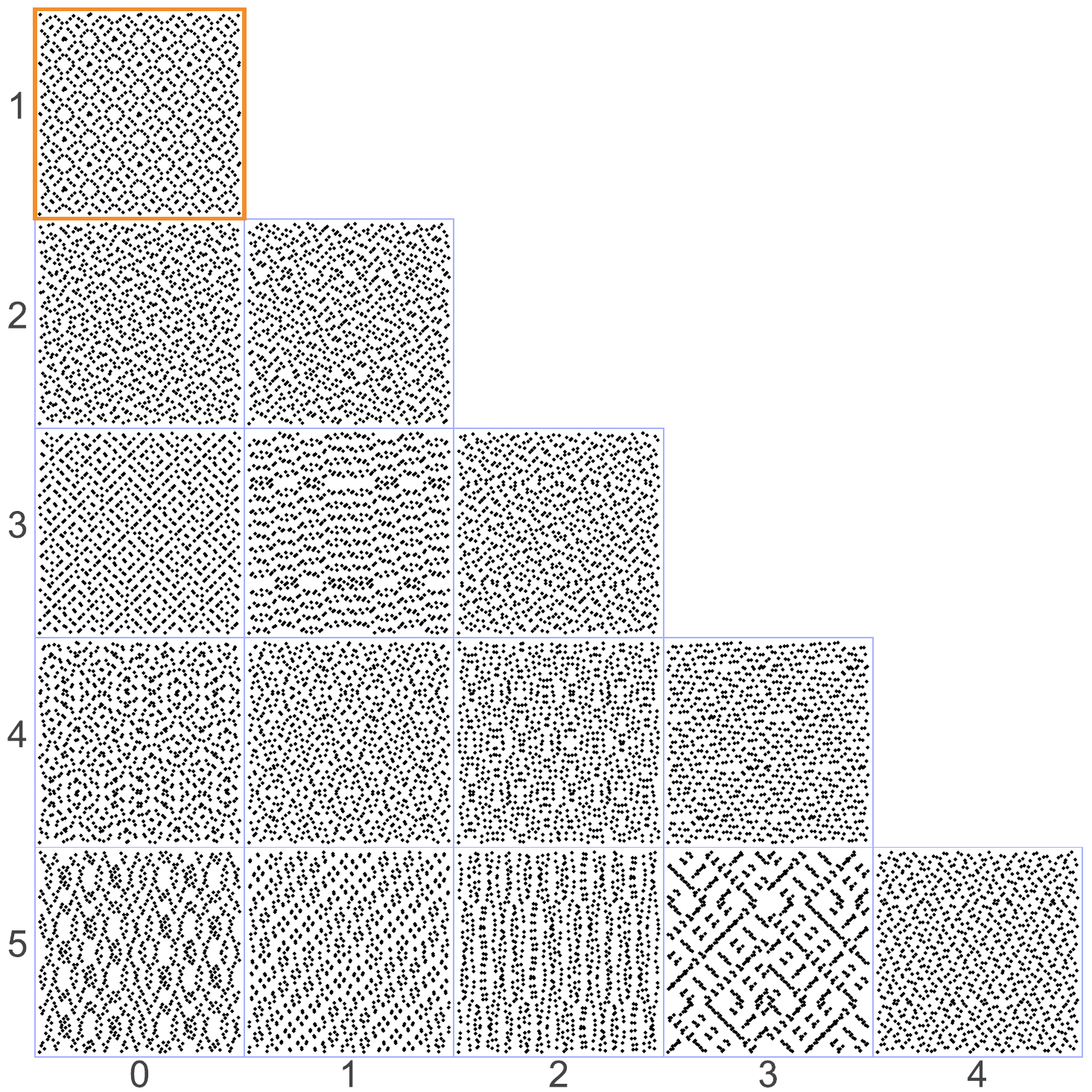"} &
\includegraphics[width=5.5cm]{"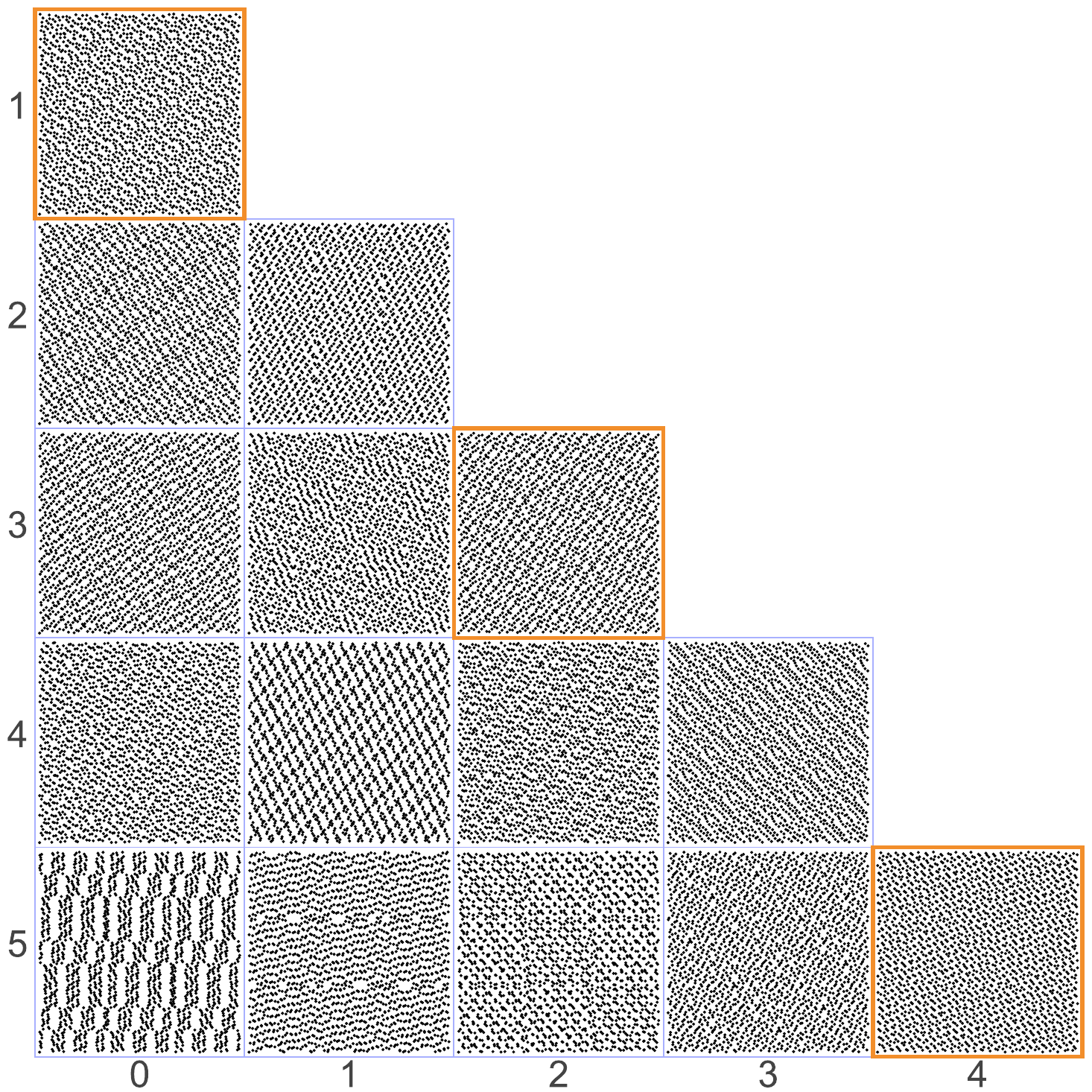"} \\
Faure and Lemieux~\cite{faure2016irreducible} in $\mathrm{GF}(2)$ & Bonneel et al.~\cite{BCIO25} in  $\mathrm{GF}(2)$ & Ostromoukhov et al.~\cite{OBCI24} in $\mathrm{GF}(3)$\\
\end{tabular}
\caption{
Pairwise 2D projections for $2^{10} = 1024$ samples in $\mathrm{GF}(2)$, $3^7 = 2187$ samples in $\mathrm{GF}(3)$, $5^4 = 625$ samples in $\mathrm{GF}(5)$,   $7^4 = 2401$ samples in $\mathrm{GF}(7)$, $11^3 = 1331$ samples in $\mathrm{GF}(11)$ for various samples, including our Artin-Schreier (AS) sampler, for up to 6 dimensions.
}
\label{fig:as}
\end{figure*}


\begin{figure*}
  \setlength{\tabcolsep}{0pt}
  \begin{center}
  \begin{tabular}{ccc}
  \begin{overpic}[width=6cm]{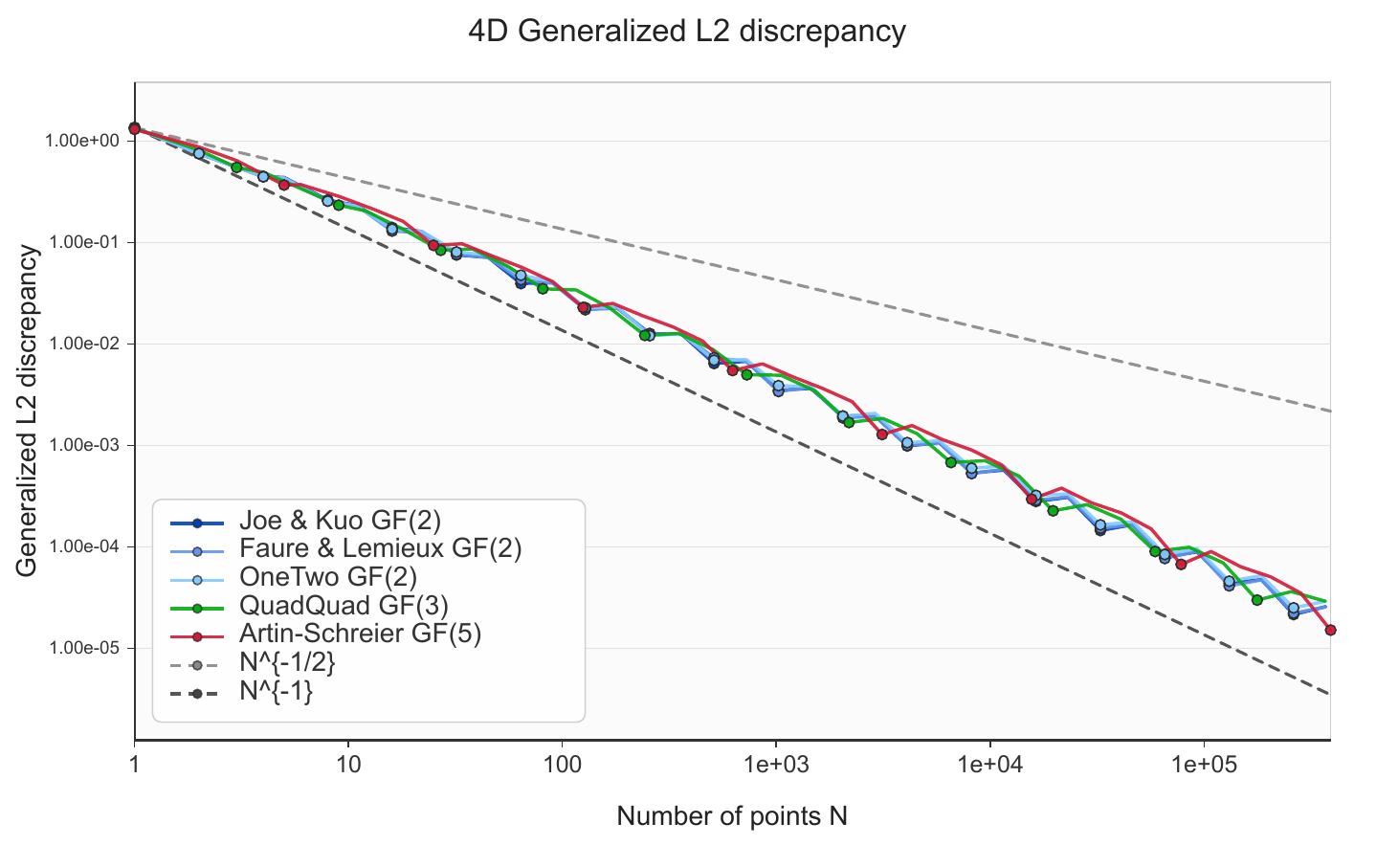}
    \put(80,51){\fcolorbox{black}{white}{$\mathrm{GF}(5)$}}
  \end{overpic}&
 \begin{overpic}[width=6cm]{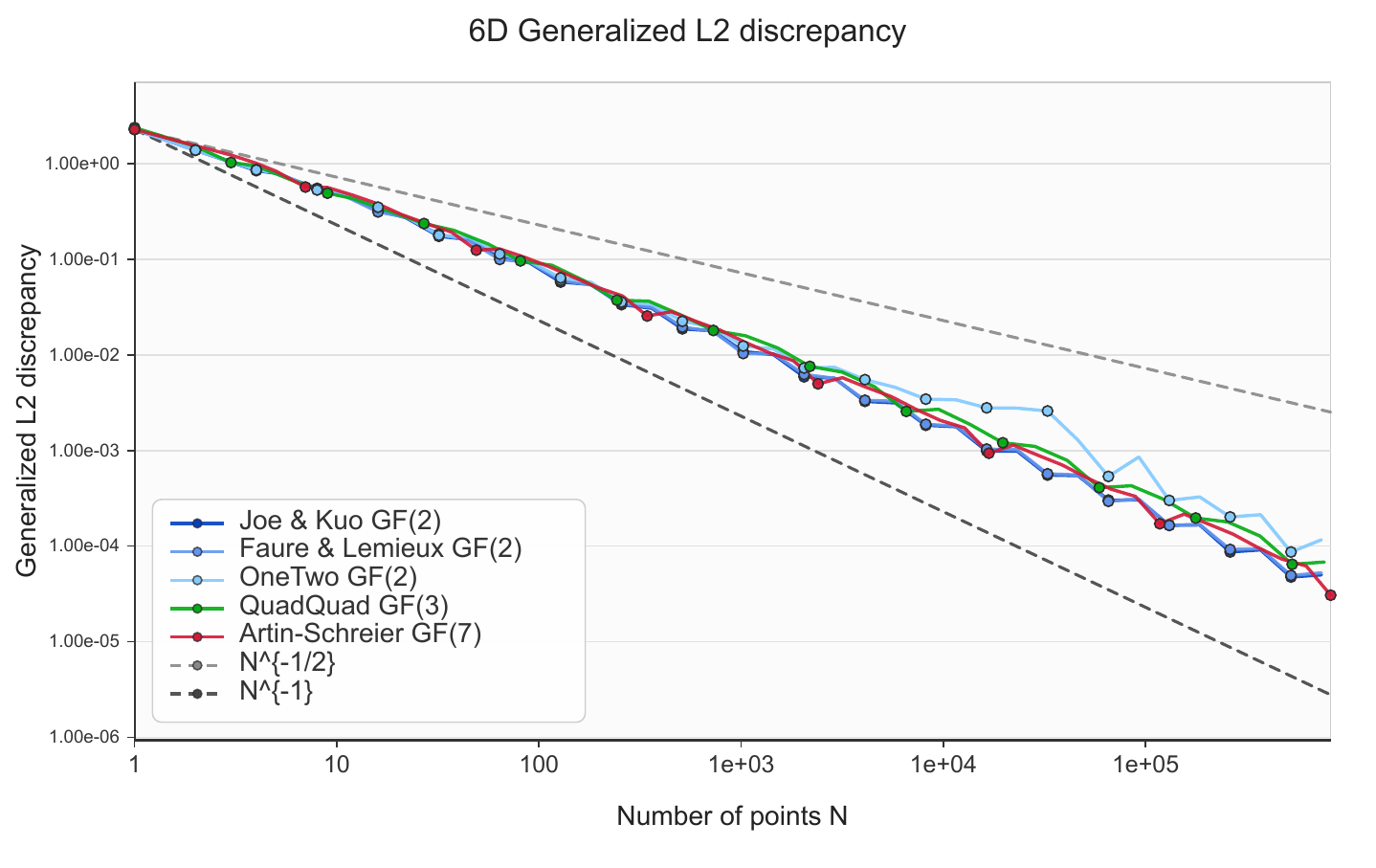}
    \put(80,51){\fcolorbox{black}{white}{$\mathrm{GF}(7)$}}
  \end{overpic}& 
 \begin{overpic}[width=6cm]{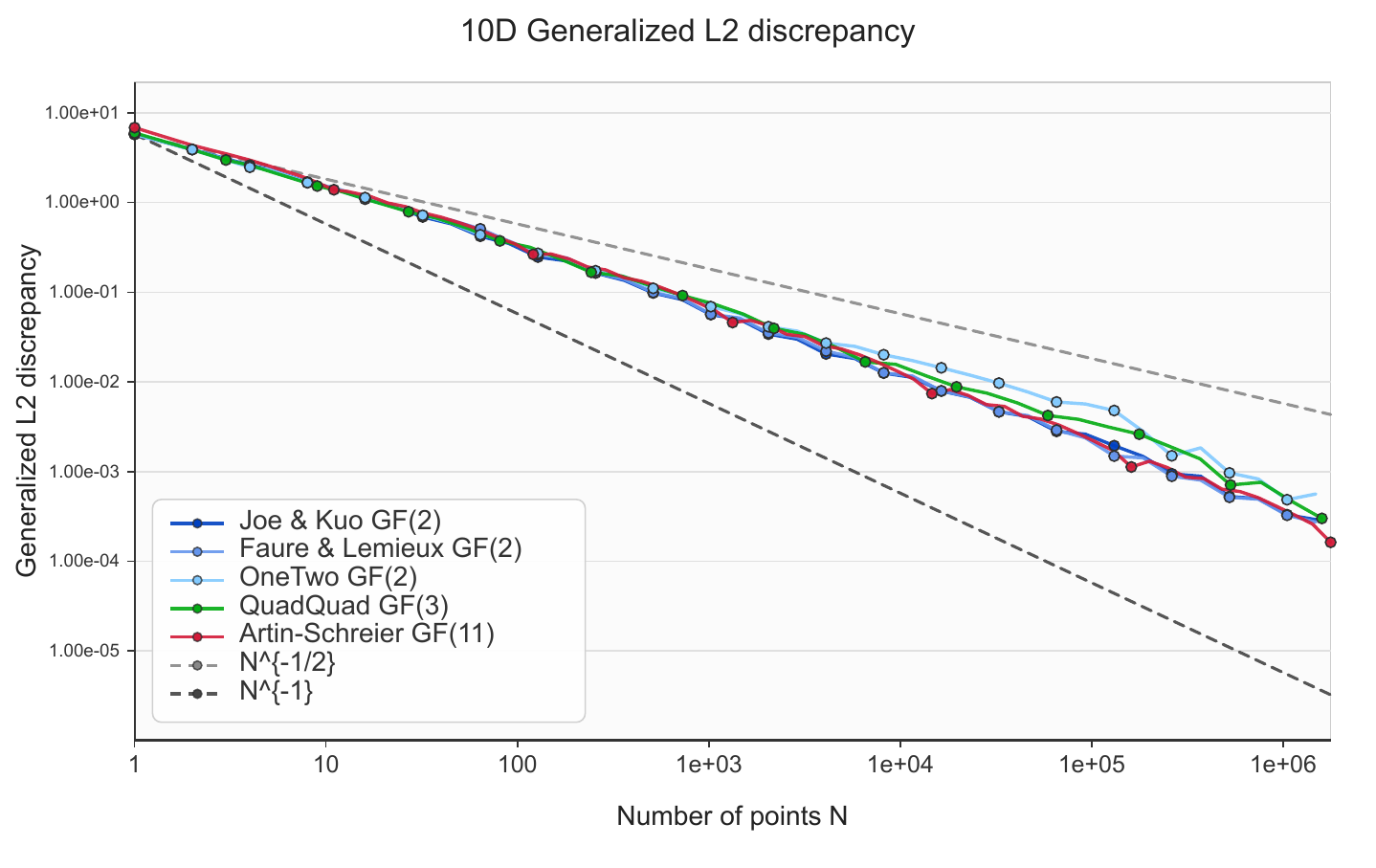}
    \put(78,51){\fcolorbox{black}{white}{$\mathrm{GF}(11)$}}
  \end{overpic}
  \end{tabular}
\end{center}
    \begin{center}
        \setlength{\tabcolsep}{2pt}\Large
\begin{tabular}{ccc}
 \begin{overpic}[width=7cm]{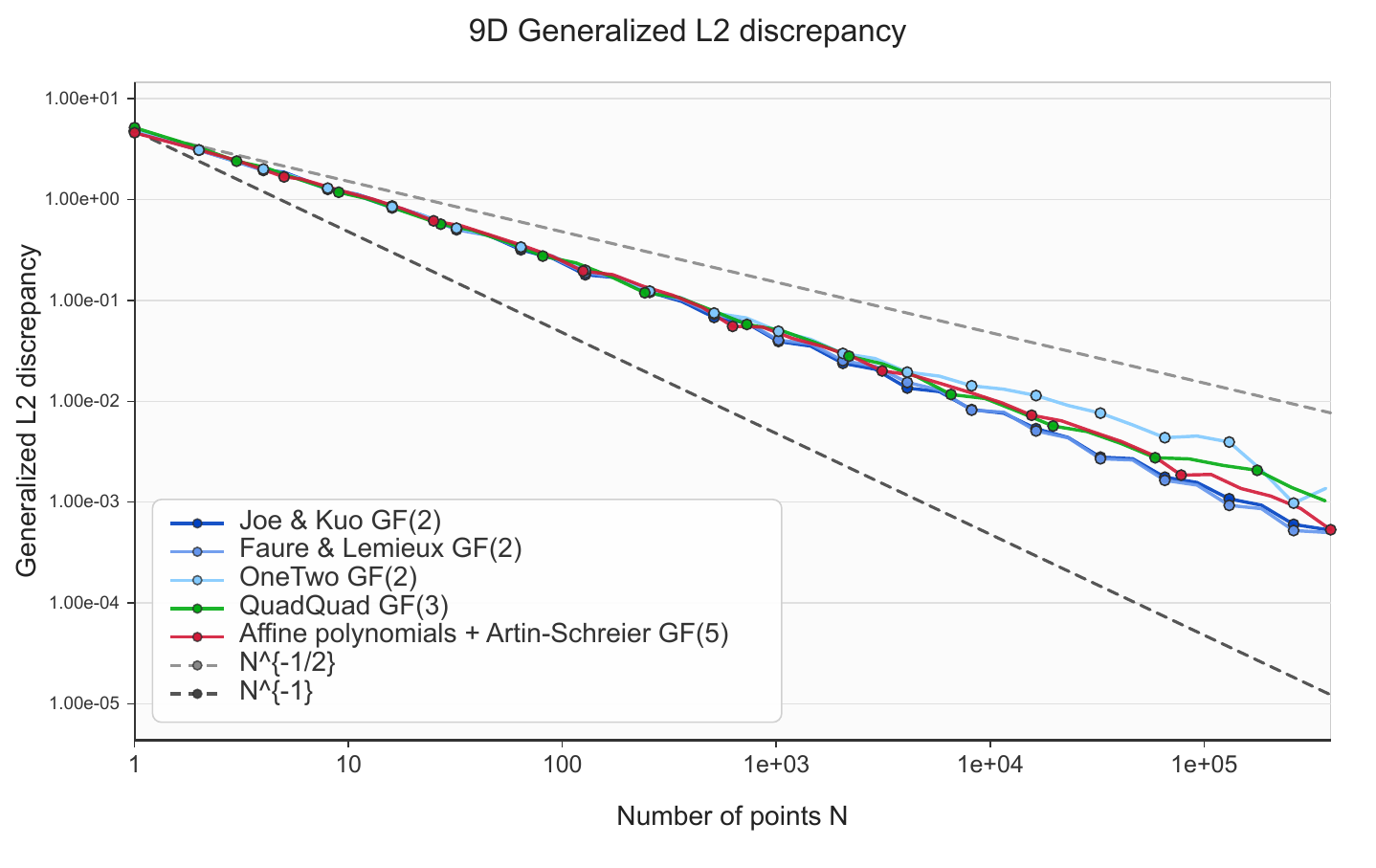}
    \put(82,51){\fcolorbox{black}{white}{$\mathrm{GF}(5)$}}
  \end{overpic}
  &\raisebox{2cm}{$\Rightarrow$}&\begin{overpic}[width=7cm]{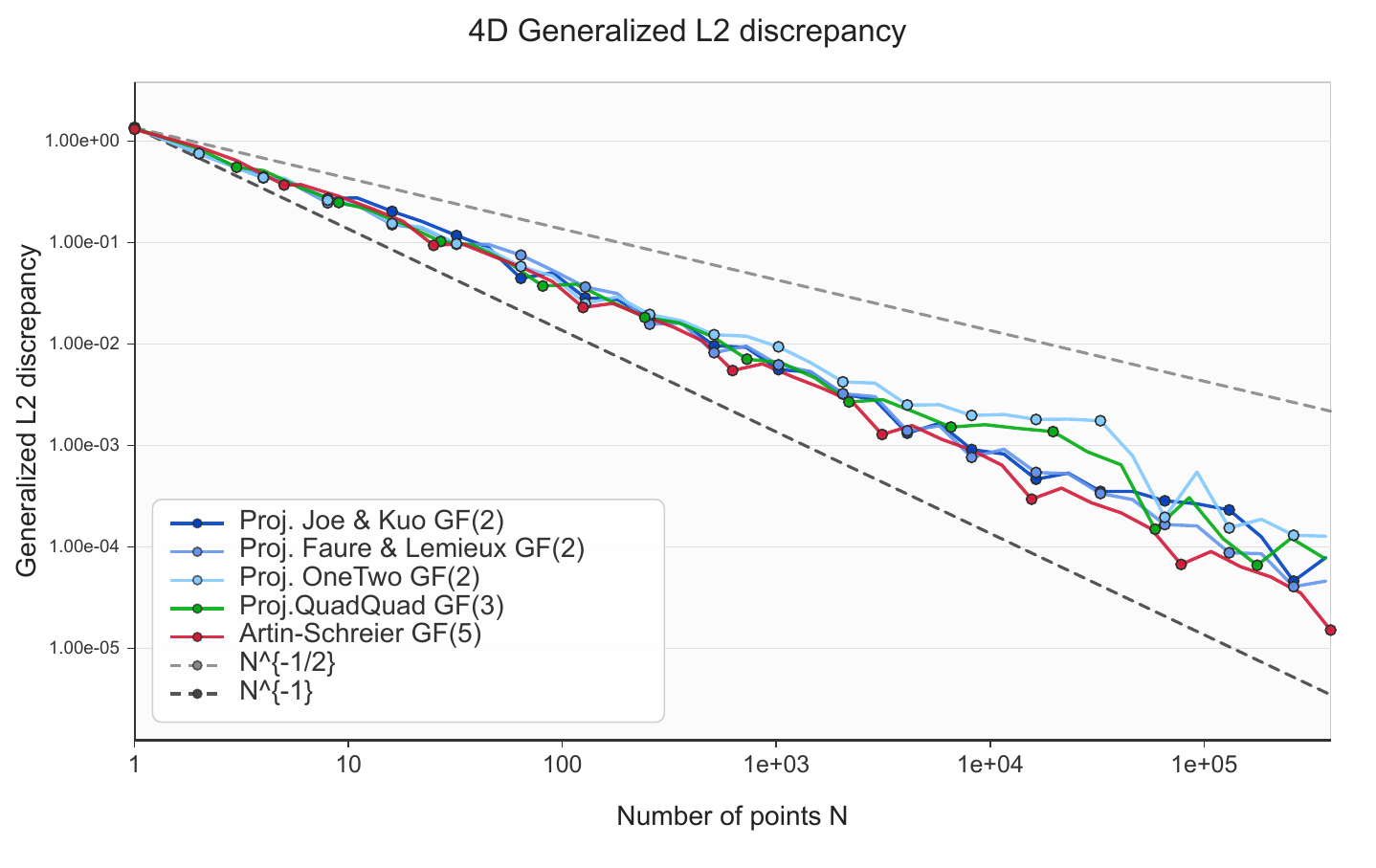}
    \put(82,51){\fcolorbox{black}{white}{$\mathrm{GF}(5)$}}
  \end{overpic}\\
 \begin{overpic}[width=7cm]{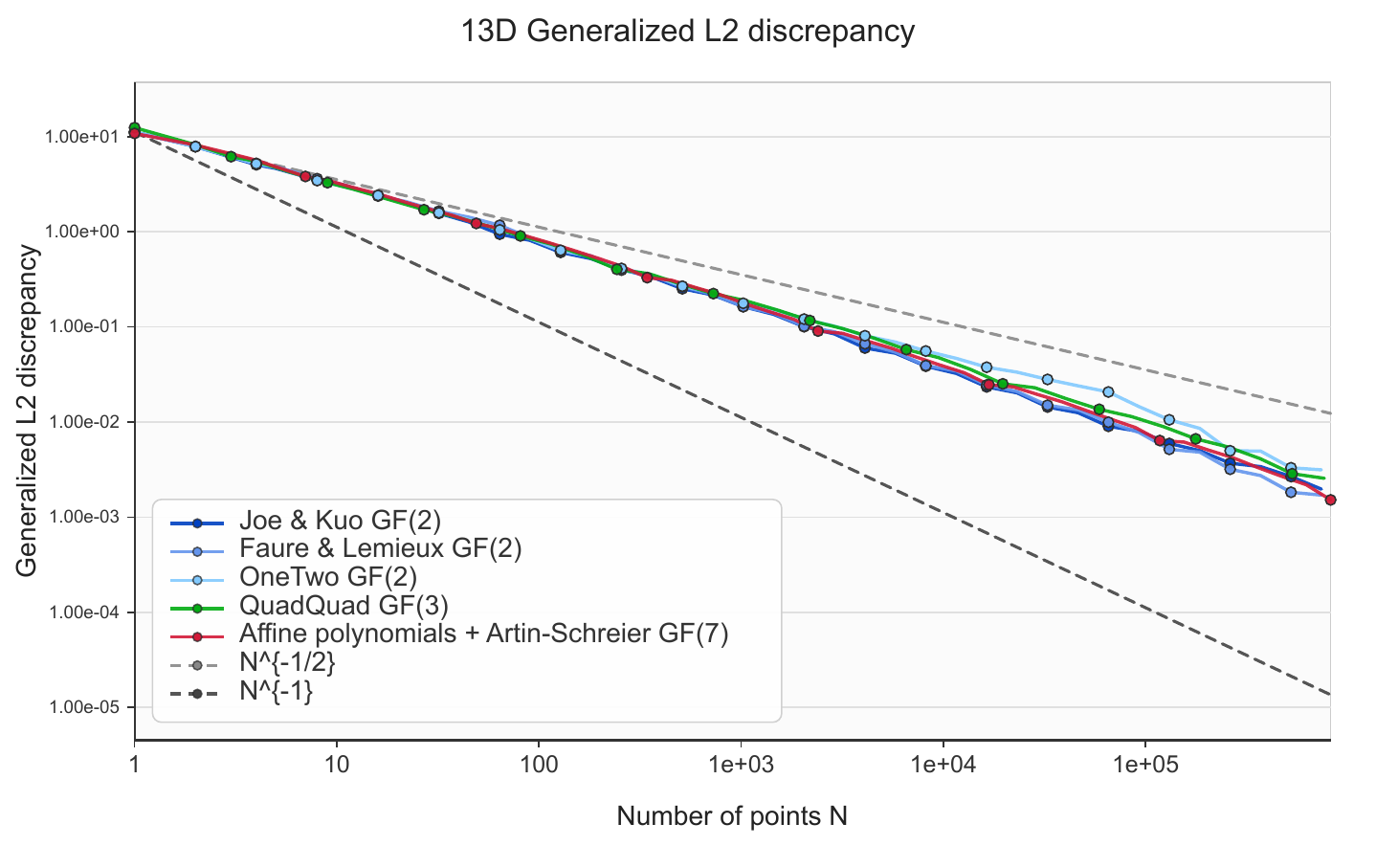}
    \put(82,51){\fcolorbox{black}{white}{$\mathrm{GF}(7)$}}
  \end{overpic}
  &\raisebox{2cm}{$\Rightarrow$}&\begin{overpic}[width=7cm]{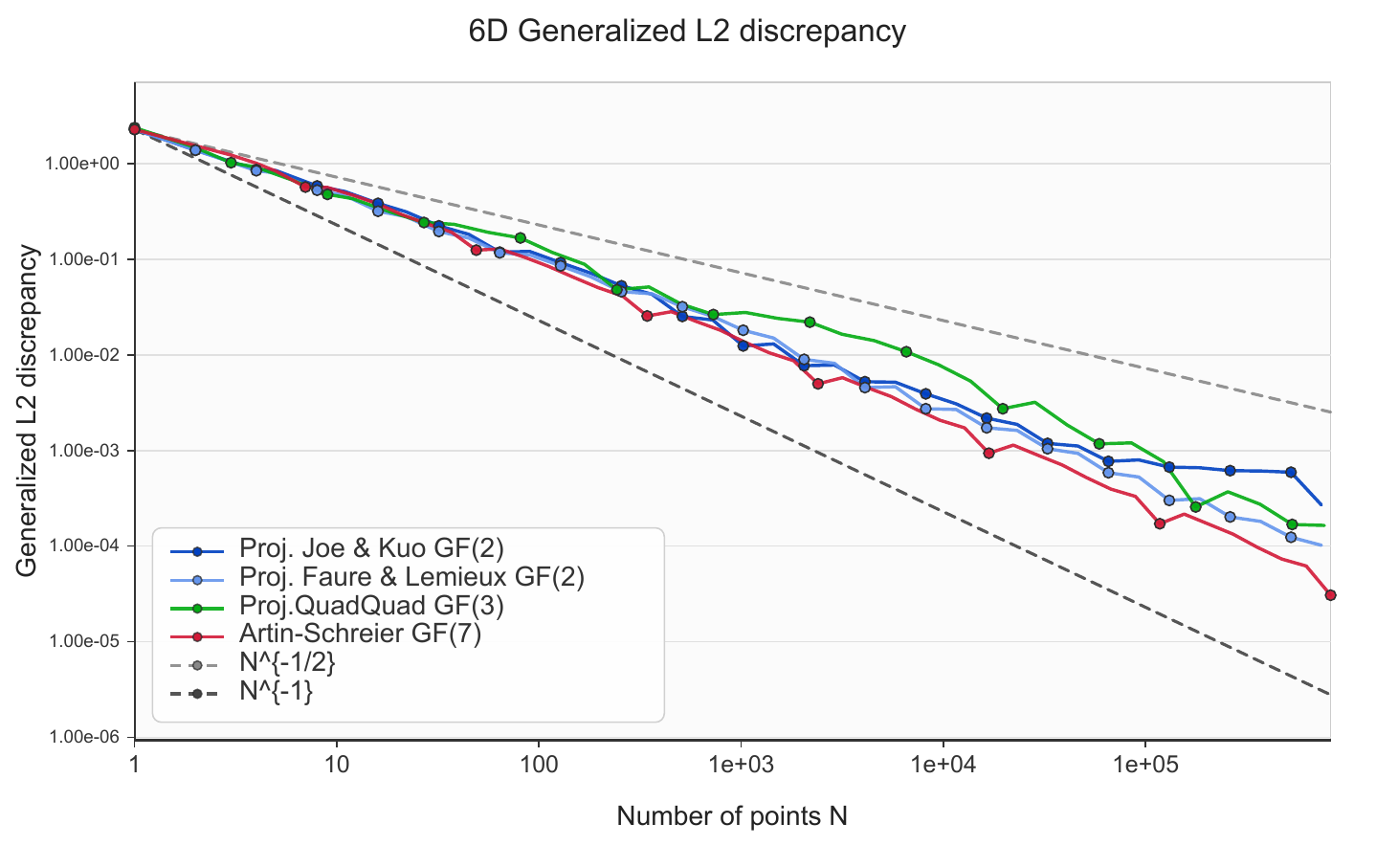}
    \put(82,51){\fcolorbox{black}{white}{$\mathrm{GF}(7)$}}
  \end{overpic}\\
  \begin{overpic}[width=7cm]{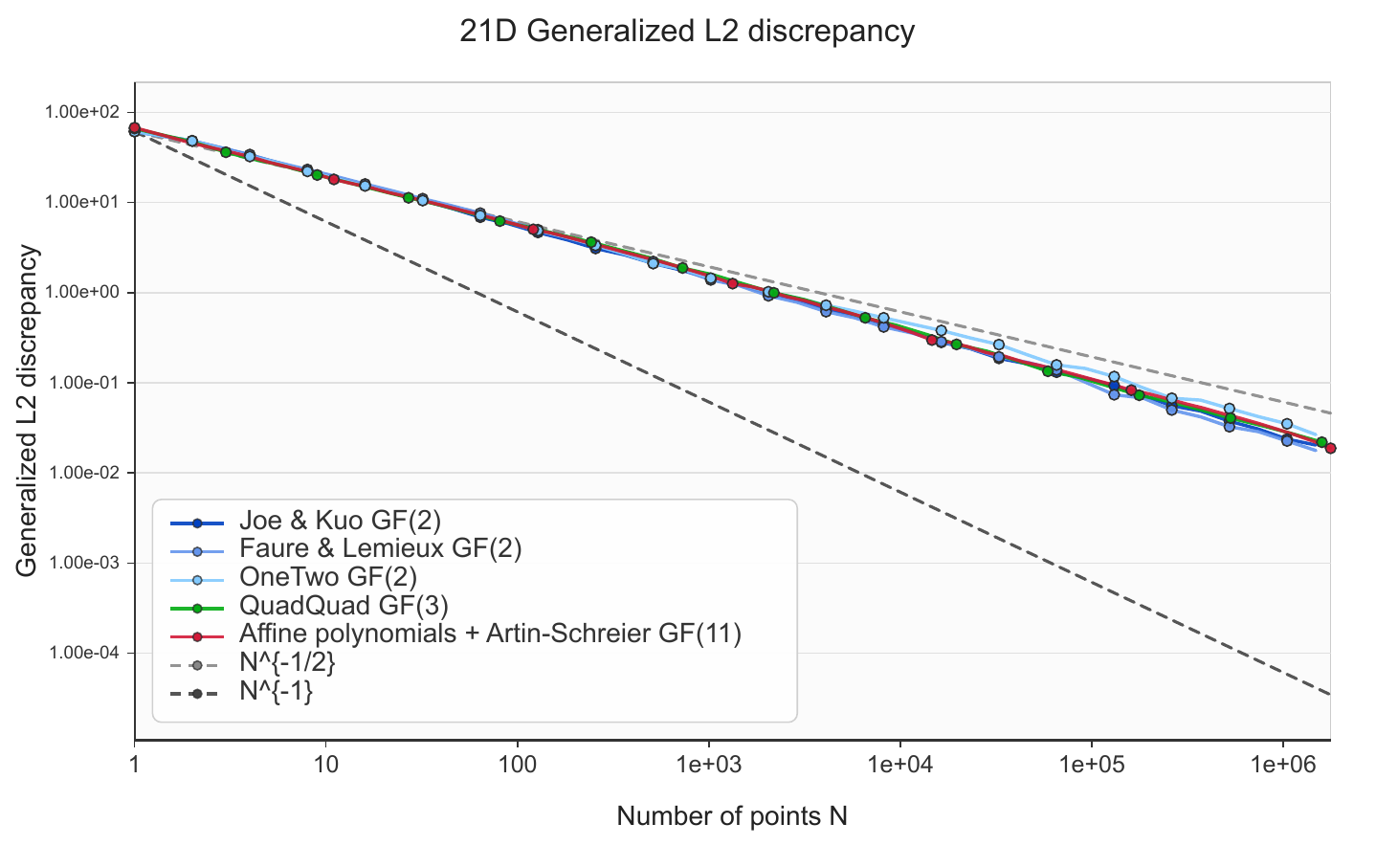}
    \put(80,51){\fcolorbox{black}{white}{$\mathrm{GF}(11)$}}
  \end{overpic}
  &\raisebox{2cm}{$\Rightarrow$}&\begin{overpic}[width=7cm]{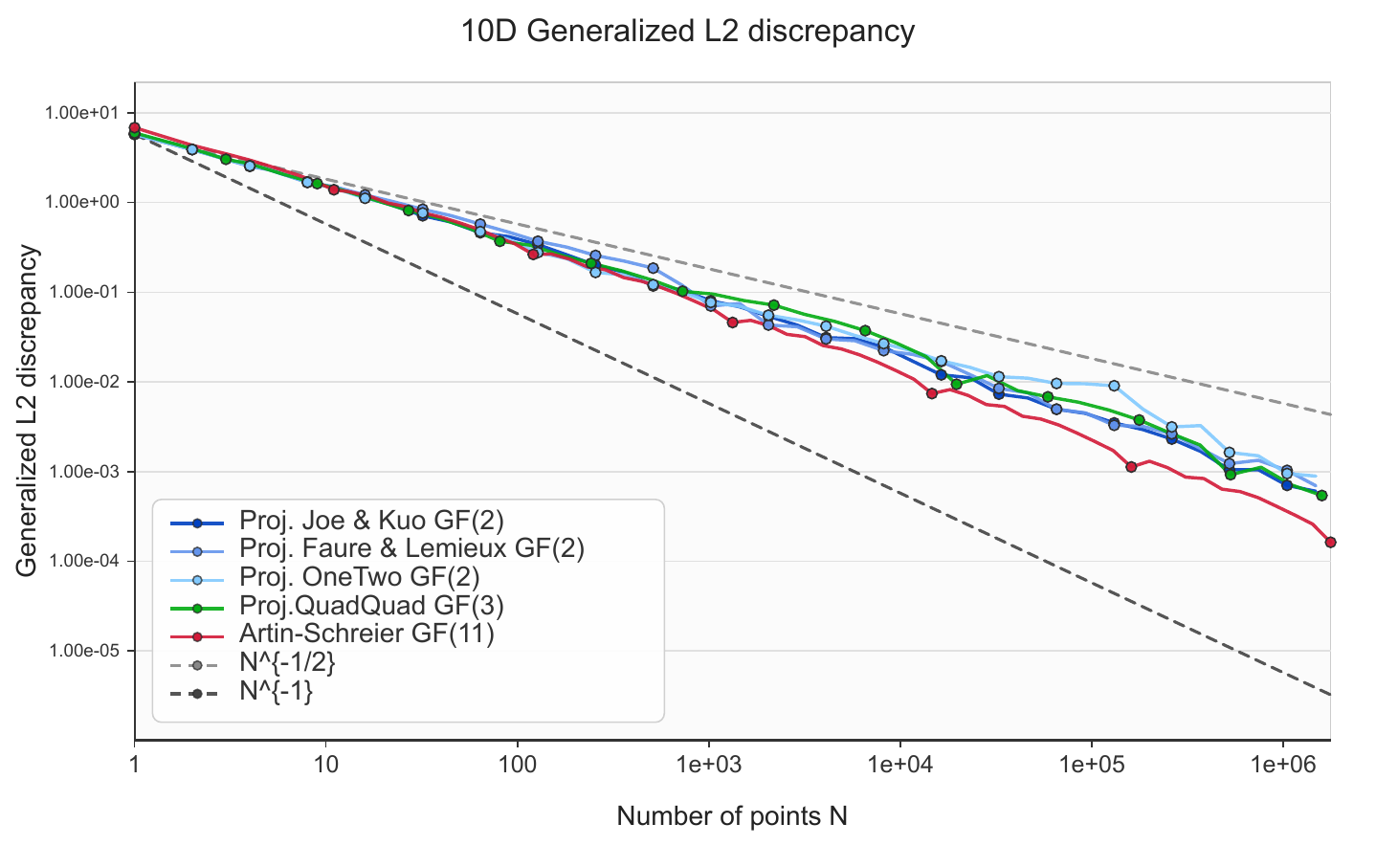}
    \put(80,51){\fcolorbox{black}{white}{$\mathrm{GF}(11)$}}
  \end{overpic}
  \end{tabular}
\end{center}
\caption{We evaluate the $(b-1)$-dimensional discrepancy of our sequences in $\mathrm{GF}(b)$ for $b=\{5, 7, 11\}$  (top 3). Then, we consider our sequence combined with the $b$ affine Sobol' polynomials in $\mathrm{GF}(b)$, and compute the combined $(2b-1)$-dimensional discrepancy (bottom left column) and compare the discrepancy of its projection on the $(b-1)$ last dimensions (bottom right column) to other Sobol' based sequences (``Proj'' denotes the last $b-1$ dimensions of these $(2b-1)$-dimensional sequences).}
\label{fig:discrepancies}
\end{figure*}

\subsection{Combined Sobol'-Artin-Schreier sequences}

In this section we investigate the behavior of sequences obtained by combining the first $b$ irreducible affine polynomials that produce powers of Pascal Sobol' matrices (with no tensorization involved) and the $b-1$ Artin-Schreier polynomials to form $(2b-1)$-dimensional sequences.
In practice, proposition~\ref{prop:init} provides a way to construct $(b-1)^{b-1}$ different sequences (since multiplying $D$ by a constant does not change the initialization matrix), all guaranteeing $t=0$ in $b-1$ dimensions, by exploring potential diagonal matrices $D$. When combining these $b-1$ dimensions with other Sobol' dimensions, it can be desirable to optimize criteria over the combined $2b-1$ dimensions. For instance, while taking $D = \text{Id}_{b \times b}$ still produces a $t=0$ sequence in $b-1$ dimensions, the sequence obtained by combining the powers of Pascal matrices resulting from affine polynomials with the powers of Pascal matrices obtained by the initialization of proposition~\ref{prop:init} with an identity diagonal matrix would result in the worst possible behavior among possible diagonal matrices in $2b-1$ dimensions, and especially for sample counts $n < b^b$ (see figure~\ref{fig:allinits} for the combined $(2b-1)$-dimensional discrepancy for all initializations in $\mathrm{GF}(5)$). 
While we rely on exhaustive search over these $(b-1)^{b-1}$ values for up to  $\mathrm{GF}(7)$,  this can be intractable for larger $b$. We alleviate this issue using a greedy approach. We first define our criterion as a lexicographic order of $t$ values computed for each $m$ as computed by the fast approach of Marion et al.~\cite{marion2020algorithm}, and progressively test diagonal elements $D_{i,i}$ for increasing $i$ of matrix $D$, making the search only over $(b-1)^2$ values. The resulting $(2b-1)$-dimensional discrepancy of our optimized $D$ values for  $\mathrm{GF}(5)$,  $\mathrm{GF}(7)$ and  $\mathrm{GF}(11)$ can be seen in figure~\ref{fig:discrepancies}, compared to other Sobol' samplers. Our resulting diagonal matrices contain the following values in the diagonal: $[1, 2, 3, 1, 4]$ in $\mathrm{GF}(5)$, $[1, 2, 6, 5, 1, 2, 6]$ in $\mathrm{GF}(7)$ and $[1, 2, 1, 1, 8, 7, 6, 9, 5, 1, 2]$ in $\mathrm{GF}(11)$.

\section{Conclusions}

We showed that Sobol' sequences in $\mathrm{GF}(b)$ constructed with polynomials of degree $e=b$ that only differ by a constant relate to Faure sequences~\cite{faure1982discrepance} and correspond to tensorized powers of Pascal matrices. These can be initialized in such a way that they guarantee the highest quality $t=0$. While irreducibility is not required per se when only using consecutive polynomials, one can ensure compatibility with Sobol' sequences by using irreducible polynomials to guarantee that higher-dimensional sequences formed by concatenating Sobol' sequences obtained by consecutive polynomials and other irreducible polynomials still produce $(t,s)$-sequences. In this setting, choosing Artin-Schreier polynomials ensure that $b-1$ irreducible polynomials can be found. We show that the resulting sequences are competitive in terms of discrepancy.

\section{Materials and Methods}
\subsection*{Disclosure of Delegation to Generative AI}
The authors declare the use of generative AI in the research and writing process. According to the GAIDeT taxonomy (2025), the following tasks were delegated to GAI tools under full human supervision:
\begin{itemize}
\item Research design
\item Proofreading and editing
\end{itemize}

The GAI tool used was: ChatGPT 5.5, Claude Opus 4.8.
Responsibility for the final manuscript lies entirely with the authors.
GAI tools are not listed as authors and do not bear responsibility for the final outcomes.
Declaration submitted by: Nicolas Bonneel

\subsection*{Data, Materials and Software Availability}
Our code for optimizing over initializations and comparing discrepancy graphs to other sequences is available at \url{https://github.com/liris-origami/Artin-Schreier-sequences}

\bibliographystyle{alpha}
{\footnotesize
\bibliography{paper}
}

\end{document}